\pdfoutput=1
\PassOptionsToPackage{svgnames}{xcolor}

\documentclass[3p]{cas-sc}

\usepackage{booktabs}   \usepackage{subcaption} \usepackage{amssymb}
\usepackage{mathrsfs}
\usepackage[T1]{fontenc}

\usepackage{etoolbox}

\usepackage[nomargin,inline,index,status=draft ]{fixme} \fxusetheme{colorsig}
\FXRegisterAuthor{fxAS}{anfxAS}{AS}\FXRegisterAuthor{fxAB}{anfxAB}{AB}\FXRegisterAuthor{fxNY}{anfxNY}{NY}\FXRegisterAuthor{fxFZ}{anfxFZ}{FZ}
\FXRegisterAuthor{fxPH}{anfxPH}{PH}

\makeatletter
\DeclareFontFamily{OMX}{MnSymbolE}{}
\DeclareSymbolFont{MnLargeSymbols}{OMX}{MnSymbolE}{m}{n}
\SetSymbolFont{MnLargeSymbols}{bold}{OMX}{MnSymbolE}{b}{n}
\DeclareFontShape{OMX}{MnSymbolE}{m}{n}{
    <-6>  MnSymbolE5
   <6-7>  MnSymbolE6
   <7-8>  MnSymbolE7
   <8-9>  MnSymbolE8
   <9-10> MnSymbolE9
  <10-12> MnSymbolE10
  <12->   MnSymbolE12
}{}
\DeclareFontShape{OMX}{MnSymbolE}{b}{n}{
    <-6>  MnSymbolE-Bold5
   <6-7>  MnSymbolE-Bold6
   <7-8>  MnSymbolE-Bold7
   <8-9>  MnSymbolE-Bold8
   <9-10> MnSymbolE-Bold9
  <10-12> MnSymbolE-Bold10
  <12->   MnSymbolE-Bold12
}{}

\let\llangle\@undefined
\let\rrangle\@undefined
\let\llbracket\@undefined
\let\rrbracket\@undefined
\DeclareMathDelimiter{\llangle}{\mathopen}{MnLargeSymbols}{'164}{MnLargeSymbols}{'164}
\DeclareMathDelimiter{\rrangle}{\mathclose}{MnLargeSymbols}{'171}{MnLargeSymbols}{'171}
\DeclareMathDelimiter{\llbracket}{\mathopen}{MnLargeSymbols}{'102}{MnLargeSymbols}{'102}
\DeclareMathDelimiter{\rrbracket}{\mathclose}{MnLargeSymbols}{'107}{MnLargeSymbols}{'107}
\makeatother

\usepackage[inline]{enumitem}

\usepackage{nicefrac}

\usepackage{proof}

\usepackage[most]{tcolorbox}\newtcolorbox{myframe}[1][]{
  enhanced,
  arc=0pt,
  outer arc=0pt,
  colback=white,
  boxrule=0.5pt,
  boxsep=0mm,
  left=1mm,
  right=1mm,
  top=0.5mm,
  bottom=0.5mm,
  #1
}

\lstdefinelanguage{Scribble}{basicstyle=\footnotesize\ttfamily,
  stringstyle=\color{Blue},
  showstringspaces=false,
  keywords={nested,new,calls,and,as,at,by,catches,choice,continue,do,from,global,import,instantiates,interruptible,local,module,or,par,protocol,rec,role,sig,throws,to,type,with,int,aux,reliable,crash},
  morestring=[b]",
  morestring=[b]',
  morecomment=[l][\color{greencomments}]{//},
}

\lstdefinelanguage{nuScr}{basicstyle=\footnotesize\ttfamily,
  stringstyle=\color{Blue},
  showstringspaces=false,
  keywords={
    nested,new,calls,and,as,at,by,catches,choice,continue,do,from,global,import,instantiates,interruptible,local,module,or,par,protocol,rec,role,sig,throws,to,type,with,int,aux,
    safe
  },
  morestring=[b]",
  morestring=[b]',
  morecomment=[l][\color{greencomments}]{//},
  morecomment=[s][\color{magenta}]{(*}{*)},
}

\lstdefinelanguage{effpi}{
  keywords=[1]{
    case,class,sealed,abstract,object,extends,type,def,val,if,else,new,var,match
  },
  keywords=[2]{
    InChan,OutChan,RecVar,Rec,Out,In,InErr,Loop,
  },
  keywords=[3]{
    rec,send,receive,receiveErr,eval,par,Channel
  },
  keywordstyle=[1]{\color{blue}},
  keywordstyle=[2]{\color{ImperialIris}}, keywordstyle=[3]{\color{OliveGreen}},
  otherkeywords={=>,.type,<:,>>:},
morecomment=[l][\color{darkgray}]{//},
}

\usepackage{amsthm}
\usepackage{thmtools}\usepackage{arydshln}

\usepackage{xifthen}\usepackage{xspace}

\usepackage{mathtools}

\usepackage{balance}

\usepackage{hyperref}
\hypersetup{hidelinks}
\usepackage[capitalise]{cleveref}

\definecolor{ImperialBlue}{HTML}{003E74}
\definecolor{ImperialDarkGreen}{HTML}{02893B}
\definecolor{ImperialTangerine}{HTML}{EC7300}
\definecolor{ImperialIris}{HTML}{751E66}

\definecolor{RYB1}{RGB}{141, 211, 199}
\definecolor{RYB2}{RGB}{255, 255, 179}
\definecolor{RYB3}{RGB}{190, 186, 218}
\definecolor{RYB4}{RGB}{251, 128, 114}
\definecolor{RYB5}{RGB}{128, 177, 211}
\definecolor{RYB6}{RGB}{253, 180, 98}
\definecolor{RYB7}{RGB}{179, 222, 105}

\usetikzlibrary{automata, positioning, arrows, arrows.meta, calc, fit, matrix}
\tikzset{
  >=stealth,
  node distance=2cm,
  every state/.style={thick, fill=gray!10},
  initial text=$ $,
}
\usepackage{pgfplots}
\pgfplotscreateplotcyclelist{stackexchange}{
  {RYB1!50!black,fill=RYB1},
  {RYB4!50!black,fill=RYB4},
  {RYB3!50!black,fill=RYB3},
  {RYB2!50!black,fill=RYB2},
  {RYB5!50!black,fill=RYB5},
  {RYB6!50!black,fill=RYB6},
  {RYB7!50!black,fill=RYB7},
}
\pgfplotscreateplotcyclelist{stackexchangePlusOne}{
{RYB4!50!black,fill=RYB4},
  {RYB3!50!black,fill=RYB3},
  {RYB2!50!black,fill=RYB2},
  {RYB5!50!black,fill=RYB5},
  {RYB6!50!black,fill=RYB6},
  {RYB7!50!black,fill=RYB7},
}
\pgfplotsset{
  compat=1.8,
  /pgfplots/bar cycle list/.style={/pgfplots/cycle list={{brown!60!black,fill=brown!30!white,mark=none},
    {red,fill=red!30!white,mark=none},
    {blue,fill=blue!30!white,mark=none},
    {black,fill=gray,mark=none},
    }
  },
}
\usepackage{tabularx}
\newcolumntype{L}{>{$}l<{$}}
\newcolumntype{C}{>{$}c<{$}}
\newcolumntype{P}[1]{>{\centering\arraybackslash$}p{#1}<{$}}
\usepackage{multirow}

\usepackage[export]{adjustbox}

\newtheorem{theorem}{Theorem}
\newtheorem{lemma}[theorem]{Lemma}
\newdefinition{remark}{Remark}
\newdefinition{definition}{Definition}
\newdefinition{example}{Example}

\Crefname{section}{\S\!}{\S\!}\Crefname{subsection}{\S\!}{\S\!}\Crefname{subsubsection}{\S\!}{\S\!}\Crefname{appendix}{Appendix \S\!}{Appendix \S\!}
\Crefname{definition}{Def.\@}{Defs.\@}\Crefname{figure}{Fig.\@}{Figs.\@}\Crefname{example}{Ex.\@}{Exs.\@}\Crefname{corollary}{Cor.\@}{Cors.\@}\Crefname{theorem}{Thm.\@}{Thms.\@}\Crefname{proposition}{Prop.\@}{Props.\@}\Crefname{lemma}{Lem.\@}{Lems.\@}
\Crefname{equation}{Eq.\@}{Eqs.\@}

\crefname{section}{\S\!}{\S\!}\crefname{subsection}{\S\!}{\S\!}\crefname{subsubsection}{\S\!}{\S\!}\crefname{appendix}{Appendix \S\!}{Appendix \S\!}
\crefname{definition}{Def.\@}{Defs.\@}\crefname{figure}{Fig.\@}{Figs.\@}\crefname{example}{Ex.\@}{Exs.\@}\crefname{corollary}{Cor.\@}{Cors.\@}\crefname{theorem}{Thm.\@}{Thms.\@}\crefname{proposition}{Prop.\@}{Props.\@}\crefname{lemma}{Lem.\@}{Lems.\@}
\crefname{equation}{Eq.\@}{Eqs.\@}

\usepackage{cite} 

\usepackage{magicvariables}

\usepackage{verse}
\usepackage{epigraph} 
\usepackage{blindtext} 

\usepackage{centernot}
\newif\ifdraft \draftfalse \drafttrue 

\newcommand{\ifempty}[3]{\ifthenelse{\isempty{#1}}{#2}{#3}}

\newtoggle{techreport}

\newtoggle{cruft}

\renewcommand{\poemtitlefont}{\normalfont\large\itshape\centering}

\newcommand{\dom}[1]{{\color{black}\operatorname{dom}\!\left({#1}\right)}}\newcommand{\suchthat}{\colon}\newcommand{\fv}[1]{\operatorname{fv}\!\left({#1}\right)}\newcommand{\unfoldOne}[1]{{\color{black}\operatorname{unf}\!\left({#1}\right)}}\newcommand{\notImpliedBy}{\mathrel{{\kern 1em}{\not{\kern -1em}\impliedby}}}

\newcommand{\coloncolonequals}{\Coloneqq}\newcommand{\bnfdef}{\coloncolonequals}\newcommand{\bnfsep}{\mathbin{\;\big|\;}}

\def\etc{\emph{etc}\@\xspace}\def\eg{e.g.\@\xspace}\def\ie{i.e.\@\xspace}

\definecolor{ruleColor}{rgb}{0.1, 0.3, 0.1}\newcommand{\inferrule}[1]{{\color{ruleColor}\textsc{\scriptsize [#1]}}}\newcommand{\inference}[3][]{\infer[\ifempty{#1}{}{\inferrule{#1}}]{#3}{#2}}\newcommand{\cinference}[3][]{\infer=[\ifempty{#1}{}{\inferrule{#1}}]{#3}{#2}}

\newcommand{\RULE}[1]{{\color{ruleColor}\textsc{\scriptsize [#1]}}}

\newcommand{\bind}[2]{\nicefrac{#2}{#1}}\newcommand{\substenum}[1]{\mathord{\left\{{#1}\right\}}}\newcommand{\subst}[2]{\substenum{\bind{#1}{#2}}}

\definecolor{hlColor}{rgb}{0.65, 1.0, 0.65}

\newcommand{\lbbar}{\{\kern-0.2em|}
\newcommand{\rbbar}{|\kern-0.2em\}}

\newcommand{\hpFmt}[1]{{\color{black}{#1}}}

\newcommand{\hpMove}{\mathrel{\hpFmt{\to}}}\newcommand{\hpMoveStar}{\mathrel{\hpFmt{\hpMove{}^{\!\!\!*}}}}

\definecolor{tyColorCustom}{rgb}{0.0, 0.0, 0.85}\newcommand{\tyCol}[1]{{\color{tyColorCustom}{#1}}}\newcommand{\tyFont}[1]{{#1}}\newcommand{\tyFmt}[1]{\tyCol{\tyFont{#1}}}\newcommand{\tyFontC}[1]{\operatorname{#1}}\newcommand{\tyFmtC}[1]{\tyCol{\tyFontC{#1}}}

\newcommand{\tyGroundSet}{\stFmt{\mathcal{B}}}  
\newcommand{\tyGround}[1][]{\tyFmt{\ifempty{#1}{B}{B_{#1}}}}\newcommand{\tyGroundi}[1][]{\tyFmt{\ifempty{#1}{B'}{B'_{#1}}}}\newcommand{\tyGroundii}[1][]{\tyFmt{\ifempty{#1}{B''}{B''_{#1}}}}\newcommand{\tyBool}{\tyFmtC{bool}}\newcommand{\tyUnit}{\tyFmtC{unit}}\newcommand{\tyInt}{\tyFmtC{int}}\newcommand{\tyReal}{\tyFmtC{real}}

\newcommand{\tyTi}[1][]{\tyCol{\ifempty{#1}{\tyFont{T}'}{\tyFont{T}'_{#1}}}}

\newcommand{\quH}[1][]{\tyCol{\ifempty{#1}{\tyFont{\sigma}}{\tyFont{\sigma}_{#1}}}}\newcommand{\quHi}[1][]{\tyCol{\ifempty{#1}{\tyFont{\sigma'}}{\tyFont{\sigma'}_{#1}}}}\newcommand{\quHii}[1][]{\tyCol{\ifempty{#1}{\tyFont{\sigma''}}{\tyFont{\sigma''}_{#1}}}}\newcommand{\quHiii}[1][]{\tyCol{\ifempty{#1}{\tyFont{\sigma'''}}{\tyFont{\sigma''}_{#1}}}}\newcommand{\quEmpty}[1][]{\tyCol{\tyFont{\varnothing}}}\newcommand{\quMsg}[3]{\tyCol{\left({#1},\tyFont{{#2}\ifempty{#3}{}{({#3})}}\right)}}\newcommand{\quCons}[2]{\tyCol{\tyFont{{#1}\mathbin{\!\cdot\!}{#2}}}}

\newcommand{\tySub}{\mathrel{\tyCol{\leqslant}_{\text{a}}}}

\newcommand{\muCol}[1]{{\color{red}#1}}

\newcommand{\muWordEmpty}[1][]{\muCol{\epsilon}}

\newtcolorbox{cross}{blank,breakable,parbox=false,
  overlay={\draw[red,line width=5pt] (interior.south west)--(interior.north east);
    \draw[red,line width=5pt] (interior.north west)--(interior.south east);}}

\definecolor{roleColor}{rgb}{0.5, 0.0, 0.0}\newcommand{\roleCol}[1]{{\color{roleColor}#1}}\newcommand{\roleSet}{\roleCol{\mathcal{R}}}\newcommand{\roleFmt}[1]{\ensuremath{{\boldsymbol{\roleCol{\mathtt{#1}}}}}\xspace}

\newcommand{\roleP}[1][]{\ifempty{#1}{{\color{roleColor}\roleFmt{p}}}{{\color{roleColor}\roleFmt{p}_{#1}}}}\newcommand{\rolePi}[1][]{\ifempty{#1}{{\color{roleColor}\roleFmt{p}'}}{{\color{roleColor}\roleFmt{p}'_{#1}}}}\newcommand{\roleQ}[1][]{\ifempty{#1}{{\color{roleColor}\roleFmt{q}}}{{\color{roleColor}\roleFmt{q}_{#1}}}}\newcommand{\roleQi}[1][]{\ifempty{#1}{{\color{roleColor}\roleFmt{q}'}}{{\color{roleColor}\roleFmt{q}'_{#1}}}}\newcommand{\roleR}[1][]{\ifempty{#1}{{\color{roleColor}\roleFmt{r}}}{{\color{roleColor}\roleFmt{r}_{\!#1}}}}\newcommand{\roleS}[1][]{\ifempty{#1}{{\color{roleColor}\roleFmt{s}}}{{\color{roleColor}\roleFmt{s}_{\!#1}}}}\newcommand{\roleT}[1][]{\ifempty{#1}{{\color{roleColor}\roleFmt{t}}}{{\color{roleColor}\roleFmt{t}_{\!#1}}}}

\definecolor{gtColor}{rgb}{0.43, 0.21, 0.1}\newcommand{\gtFmt}[1]{\ensuremath{{\color{gtColor}#1}}\xspace}\newcommand{\gtMsgFmt}[1]{\gtFmt{\labFmt{#1}}}\newcommand{\gtLab}[1][]{\ifempty{#1}{\gtMsgFmt{m}}{{\color{gtColor}\gtMsgFmt{m}_{#1}}}}\newcommand{\gtLabi}[1][]{\ifempty{#1}{\gtMsgFmt{m}'}{{\color{gtColor}\gtMsgFmt{m}'_{#1}}}}\newcommand{\gtLabii}[1][]{\ifempty{#1}{\gtMsgFmt{m}''}{{\color{gtColor}\gtMsgFmt{m}''_{#1}}}}

\newcommand{\gtFv}[1]{\gtFmt{\fv{#1}}}

\newcommand{\gtG}[1][]{\gtFmt{\ifempty{#1}{G}{G_{#1}}}}\newcommand{\gtGi}[1][]{\gtFmt{\ifempty{#1}{G'}{G'_{#1}}}}\newcommand{\gtGii}[1][]{\gtFmt{\ifempty{#1}{G''}{G''_{#1}}}}\newcommand{\gtGiii}[1][]{\gtFmt{\ifempty{#1}{G'''}{G'''_{#1}}}}

\newcommand{\gtGone}[1][]{\gtFmt{\ifempty{#1}{G^{(1)}}{G^{(1)}_{#1}}}}\newcommand{\gtGtwo}[1][]{\gtFmt{\ifempty{#1}{G^{(2)}}{G^{(2)}_{#1}}}}\newcommand{\gtGthree}[1][]{\gtFmt{\ifempty{#1}{G^{(3)}}{G^{(3)}_{#1}}}}\newcommand{\gtGfour}[1][]{\gtFmt{\ifempty{#1}{G^{(4)}}{G^{(4)}_{#1}}}}\newcommand{\gtGfive}[1][]{\gtFmt{\ifempty{#1}{G^{(5)}}{G^{(5)}_{#1}}}}

\newcommand{\gtSeq}{\mathbin{\gtFmt{.}}}

\newcommand{\gtCommRaw}[3]{\gtFmt{{#1} {\to} {#2}{:}\left\{{#3}\right\}}}\newcommand{\gtComm}[6]{\gtFmt{\gtCommRaw{#1}{#2}{\gtCommChoice{#4}{#5}{#6}}_{#3}}}\newcommand{\gtCommSmall}[6]{\gtFmt{\gtCommRaw{#1}{#2}{\gtCommChoiceSmall{#4}{#5}{#6}}_{#3}}}\newcommand{\gtCommSingle}[5]{\gtFmt{{#1} {\to} {#2}{:}\gtCommChoice{#3}{#4}{#5}}}

\newcommand{\gtCommRawSquig}[4]{\gtFmt{{#1} {\overset{#3}{\rightsquigarrow}} {#2}{:}\left\{{#4}\right\}}}

\newcommand{\gtCommSquig}[7]{\gtFmt{\gtCommRawSquig{#1}{#2}{#3}{\gtCommChoice{#5}{#6}{#7}}_{#4}}}

\newcommand{\gtCommSquigSmall}[7]{\gtFmt{\gtCommRawSquig{#1}{#2}{#3}{\gtCommChoiceSmall{#5}{#6}{#7}}_{#4}}}

\newcommand{\gtCommSquigSingle}[6]{\gtFmt{{#1} {\overset{#3}{\rightsquigarrow}} {#2}{:}\gtCommChoice{#4}{#5}{#6}}}

\newcommand{\gtCommChoice}[3]{\gtFmt{\gtMsgFmt{#1}\ifempty{#2}{}{({#2})}\ifempty{#3}{}{\vphantom{x} \!\gtSeq\! {#3}}}}\newcommand{\gtCommChoiceSmall}[3]{\gtFmt{\gtMsgFmt{#1}\ifempty{#2}{}{({#2})}\ifempty{#3}{}{\vphantom{x} \!\gtSeq\! {#3}}}}

\newcommand{\gtEnd}{\gtFmt{\mathbf{end}}}

\newcommand{\gtRec}[2]{\gtFmt{\mu{#1}.{#2}}}\newcommand{\gtRecVarBase}{\gtFmt{\mathbf{t}}}\newcommand{\gtRecVar}[1][]{\gtFmt{\ifempty{#1}{\gtRecVarBase}{\gtRecVarBase_{#1}}}}

\newcommand{\gtRoles}[1]{{\color{roleColor} \operatorname{roles}(\gtFmt{#1})}}\newcommand{\gtARoles}[1]{{\color{roleColor} \operatorname{aRoles}(\gtFmt{#1})}}\newcommand{\gtSRoles}[1]{{\color{roleColor} \operatorname{sRoles}(\gtFmt{#1})}}\newcommand{\gtMRoles}[1]{{\color{roleColor} \operatorname{mRoles}(\gtFmt{#1})}}\newcommand{\gtDepth}[2]{{\color{roleColor} \operatorname{depth}_{#2}(\gtFmt{#1})}}\newcommand{\gtMDepth}[3]{
\ifempty{#1}{\color{roleColor} \operatorname{mDepth}}
{\color{roleColor} \operatorname{mDepth}_{#2,#3}(\gtFmt{#1})}
}\newcommand{\gtMCount}[3]{{\color{roleColor} \operatorname{count}_{#2,#3}(\gtFmt{#1})}}

\newcommand{\gtProj}[3][]{{\color{stColor}\gtFmt{#2}\!\ifempty{#1}{\upharpoonright}{\upharpoonright_{#1}}\!\roleFmt{#3}}}

\newcommand{\gtProjRel}[3]{{\color{stColor}\gtFmt{#1}{\upharpoonright_{\roleFmt{#2}}}\stFmt{#3}}}

\newcommand{\gtMove}[1][\phantom{\stEnvAnnotGenericSym}]{\gtFmt{\xrightarrow{#1}}} \newcommand{\gtMoveStar}[1][]{\ifempty{#1}{\gtMove[]{}^{\!\gtFmt{*}}}{\gtMove[]{#1}^{\!\gtFmt{*}}}} \newcommand{\gtNotMove}[1]{{#1}\!\not\!\!\!\gtMove{}}

\newcommand{\iruleGtMove}[1]{GR-{#1}}

\newcommand{\iruleGtMoveBra}[0]{\iruleGtMove{$\&$}}
\newcommand{\iruleGtMoveSel}[0]{\iruleGtMove{$\oplus$}}
\newcommand{\iruleGtMoveRec}[0]{\iruleGtMove{$\mu$}}

\newcommand{\iruleGtMoveCtx}[0]{\iruleGtMove{Ctx-I}}
\newcommand{\iruleGtMoveCtxII}[0]{\iruleGtMove{Ctx-II}}

\newcommand{\labFmt}[2][]{\ensuremath{\ifempty{#1}{\mathtt{#2}}{\mathtt{#2}\textsubscript{#1}}}\xspace}

\definecolor{stColor}{rgb}{0, 0, 0.9}\newcommand{\stFmt}[1]{\ensuremath{{\color{stColor}#1}}\xspace}

\newcommand{\stChoice}[2]{\stLabFmt{#1}\ifempty{#2}{}{\stFmt{({#2})}}}

\newcommand{\stSeq}{\mathbin{\!\stFmt{.}\!}}\newcommand{\stIntSum}[3]{\roleFmt{#1}\stFmt{\oplus\!\left\{#3\right\}_{#2}}}\newcommand{\stIntSumRaw}[2]{\stFmt{\roleFmt{#1} \oplus\!
    \left\{{#2}\right\}}}\newcommand{\stExtSum}[3]{\roleFmt{#1}\stFmt{\&\!\left\{#3\right\}_{#2}}}\newcommand{\stExtSumRaw}[2]{\stFmt{\roleFmt{#1} \&\!
    \left\{{#2}\right\}}}\newcommand{\stRec}[2]{\stFmt{\mu{#1}.{#2}}}\newcommand{\stEnd}{\stFmt{\mathbf{end}}}

\newcommand{\stLabFmt}[1]{\stFmt{\labFmt{#1}}}\newcommand{\stLab}[1][]{\ifempty{#1}{\stLabFmt{m}}{\stLabFmt{m}_{{\color{stColor}#1}}}
}\newcommand{\stLabi}[1][]{\ifempty{#1}{\stLabFmt{m}'}{\stLabFmt{m}'_{{\color{stColor}#1}}}
}\newcommand{\stLabii}[1][]{\ifempty{#1}{\stLabFmt{m}''}{\stLabFmt{m}''_{{\color{stColor}#1}}}
}

\newcommand{\stS}[1][]{\stFmt{\ifempty{#1}{S}{S_{#1}}}}

\newcommand{\stT}[1][]{\stFmt{\ifempty{#1}{T}{T_{#1}}}}\newcommand{\stTi}[1][]{\stFmt{\ifempty{#1}{T'}{T'_{#1}}}}\newcommand{\stTii}[1][]{\stFmt{\ifempty{#1}{T''}{T''_{#1}}}}

\newcommand{\stTopt}[1][]{\stFmt{\ifempty{#1}{T^{\text{\tiny opt}}}{T_{#1}^{\text{\tiny opt}}}}}
\newcommand{\stToptrole}[1]{\stFmt{T_{#1}^{\text{\tiny opt}}}}

\newcommand{\stW}[1][]{\stFmt{\ifempty{#1}{\mathbb{W}}{\mathbb{W}_{#1}}}}\newcommand{\stWi}[1][]{\stFmt{\ifempty{#1}{\mathbb{W}'}{\mathbb{W}'_{#1}}}}\newcommand{\stWii}[1][]{\stFmt{\ifempty{#1}{\mathbb{W}''}{\mathbb{W}''_{#1}}}}

\newcommand{\stU}[1][]{\stFmt{\ifempty{#1}{\mathbb{U}}{\mathbb{U}_{#1}}}}\newcommand{\stUi}[1][]{\stFmt{\ifempty{#1}{\mathbb{U}'}{\mathbb{U}'_{#1}}}}\newcommand{\stUii}[1][]{\stFmt{\ifempty{#1}{\mathbb{U}''}{\mathbb{U}''_{#1}}}}

\newcommand{\stV}[1][]{\stFmt{\ifempty{#1}{\mathbb{V}}{\mathbb{V}_{#1}}}}\newcommand{\stVi}[1][]{\stFmt{\ifempty{#1}{\mathbb{V}'}{\mathbb{V}'_{#1}}}}

\newcommand{\trT}[1][]{\stFmt{\ifempty{#1}{\mathbb{T}}{\mathbb{T}_{#1}}}}\newcommand{\trTi}[1][]{\stFmt{\ifempty{#1}{\mathbb{T}'}{\mathbb{T}'_{#1}}}}\newcommand{\trTii}[1][]{\stFmt{\ifempty{#1}{\mathbb{T}''}{\mathbb{T}''_{#1}}}}

\newcommand{\stRecv}[4]{\roleFmt{#1}\mathord{?}\labFmt{#2}\stFmt{({#3})}; {#4}}
\newcommand{\stRecvOne}[3]{\roleFmt{#1}\mathord{?}\labFmt{#2}\stFmt{({#3})}}
\newcommand{\stSend}[4]{\roleFmt{#1}\mathord{!}\labFmt{#2}\stFmt{({#3})}; {#4}}
\newcommand{\stSendOne}[3]{\roleFmt{#1}\mathord{!}\labFmt{#2}\stFmt{({#3})}}

\newcommand{\stRecVarBase}{\stFmt{\mathbf{t}}}\newcommand{\stRecVar}[1][]{\stFmt{\ifempty{#1}{\stRecVarBase}{\stRecVarBase_{#1}}}}

\newcommand{\stMerge}[2]{\stFmt{\bigsqcap_{#1}{#2}}}\newcommand{\stBinMerge}{\mathbin{\stFmt{\sqcap}}}\definecolor{mergeRelColor}{RGB}{60,90,140}\newcommand{\stMergeRel}[2]{
\stFmt{{\color{mergeRelColor}\operatorname{merge}}
\ifempty{#1}{}{\left\langle{#1},{#2}\right\rangle}
}
}

\newcommand{\stSub}{\mathrel{\stFmt{\leqslant}}}\newcommand{\tyGroundSub}{\mathrel{\stFmt{{<}{:}}}}

\newcommand{\ltsSel}[3]{{#1}\mathord{:}{#2}\mathord{\oplus}{#3}}
\newcommand{\ltsBra}[3]{{#1}\mathord{:}{#2}\&{#3}}

\newcommand{\ltsSubject}[1]{{\color{roleColor} \operatorname{subject}({#1})}}

\definecolor{mpColor}{rgb}{0, 0, 0}\newcommand{\mpFmt}[1]{{\color{mpColor}#1}}

\newcommand{\mpTrue}{\mpFmt{\text{\textit{\texttt{true}}}}}\newcommand{\mpFalse}{\mpFmt{\text{\textit{\texttt{false}}}}}

\newcommand{\mpNil}{\mpFmt{\mathbf{0}}}\newcommand{\mpSeq}{\mathbin{\mpFmt{\!.\!}}}\newcommand{\mpIf}[3]{\mpFmt{\mathbf{if}\,{#1}\,\mathbf{then}}\,{#2}\,\mathbf{else}\,{#3}}\newcommand{\mpPar}{\mathbin{\mpFmt{\mid}}}\newcommand{\mpBigPar}[2]{\mathbin{\mpFmt{\Pi_{#1}}{#2}}}\newcommand{\mpRec}[2]{\mpFmt{\mu{#1}.{#2}}}
\newcommand{\mpErr}{\mpFmt{\boldsymbol{\mathtt{err}}}}

\newcommand{\mpV}[1][]{\mpFmt{\ifempty{#1}{v}{v_{#1}}}}

\makeVars{mpE}{mpFmt}{e}

\newcommand{\mpC}[1][]{\mpFmt{\ifempty{#1}{c}{c_{#1}}}}

\newcommand{\mpX}[1][]{\mpFmt{\ifempty{#1}{X}{X_{#1}}}}

\newcommand{\mpP}[1][]{\mpFmt{\ifempty{#1}{P}{P_{#1}}}}\newcommand{\mpQ}[1][]{\mpFmt{\ifempty{#1}{Q}{Q_{#1}}}}

\newcommand{\mpMove}{\to}

\newcommand{\iruleStSubEnd}{Sub-$\stEnd$}

\newcommand{\iruleStSubRecL}{Sub-$\stFmt{\mu}$L}
\newcommand{\iruleStSubRecR}{Sub-$\stFmt{\mu}$R}

\newcommand{\iruleStSubIn}{Sub-$\stFmt{\&}$}
\newcommand{\iruleStSubOut}{Sub-$\stFmt{\oplus}$}

\newcommand{\iruleTCtxOut}{$\stEnv$-$\stFmt{\oplus}$}\newcommand{\iruleTCtxIn}{$\stEnv$-$\stFmt{\&}$}\newcommand{\iruleTCtxRec}{$\stEnv$-$\mu$}\newcommand{\iruleTCtxCong}{$\stEnv$-$\stEnvComp$}

\newcommand{\stEnv}[1][]{\stFmt{\ifempty{#1}{\Delta}{\Delta_{#1}}}}\newcommand{\stEnvi}[1][]{\stFmt{\ifempty{#1}{\Delta'}{\Delta'_{#1}}}}\newcommand{\stEnvii}[1][]{\stFmt{\ifempty{#1}{\Delta''}{\Delta''_{#1}}}}\newcommand{\stEnviii}[1][]{\stFmt{\ifempty{#1}{\Delta'''}{\Delta'''_{#1}}}}\newcommand{\stEnvEmpty}{\stFmt{\emptyset}}\newcommand{\stEnvMap}[2]{\stFmt{\mpFmt{#1}\mathbin{\!:\!}{#2}}}\newcommand{\stEnvComp}{\mathpunct{\stFmt{,}}}\newcommand{\stEnvApp}[2]{\stFmt{#1\!\left(\mpFmt{#2}\right)}}

\newcommand{\stEnvAssoc}[3]{\stFmt{{#1} \mathrel{\stFmt{\sqsubseteq}_{#3}} {#2}}}
\newcommand{\vertAssoc}{\rotatebox[origin=c]{90}{\stEnvAssoc{}{}{a}}}

\newcommand{\stEnvMove}{\mathrel{\stFmt{\to}}}\newcommand{\stEnvMoveStar}[1][]{\ifempty{#1}{\stEnvMove{}^{\!\gtFmt{*}}}{\stEnvMove{#1}^{\!\gtFmt{*}}}}

\newcommand{\stEnvAnnotOutSym}{\stFmt{\oplus}}\newcommand{\stEnvAnnotInSym}{\stFmt{\&}}\newcommand{\stEnvAnnotGenericSym}[1][]{\stFmt{\ifempty{#1}{\alpha}{\alpha_{#1}}}}\newcommand{\stEnvAnnotGenericSymi}[1][]{\stFmt{\ifempty{#1}{\alpha'}{\alpha'_{#1}}}}\newcommand{\stEnvAnnotGenericSymii}[1][]{\stFmt{\ifempty{#1}{\alpha''}{\alpha''_{#1}}}}\newcommand{\stEnvMoveAnnot}[1]{\mathrel{\stFmt{\xrightarrow{#1}}}}
\newcommand{\stEnvMoveGenAnnot}{\stEnvMoveAnnot{\stEnvAnnotGenericSym}}
\newcommand{\stEnvNotMoveP}[1]{{#1}\!\!\not\stEnvMove}

\newcommand{\mpEnv}[1][]{\stFmt{\ifempty{#1}{\Theta}{\Theta_{#1}}}}

\newcommand{\stJudge}[3]{\stFmt{#1} \mathrel{\mpFmt{\vdash}} \mpFmt{#2} \triangleright {#3}}

\newcommand{\stJudgeTop}[3]{\stFmt{#1} \mathrel{\mpFmt{\vdash^\text{\tiny top}}} \mpFmt{#2} \triangleright {#3}}

\newcommand{\stJudgeBot}[3]{\stFmt{#1} \mathrel{\mpFmt{\vdash^\text{\tiny bot}}} \mpFmt{#2} \triangleright {#3}}

\newcommand{\subtt}{\leqslant}

\newcommand{\asubt}{\mathrel{\stFmt{\leqslant_{\text{a}}}}}

\newcommand{\set}[1]{\{#1\}}

\newcommand{\N}{M}

\newcommand{\subttt}{\mathrel{\stFmt{\lesssim}}}
\newcommand{\ttreeSym}{\mathfrak{T}}\newcommand{\ttree}[1]{\operatorname{\ttreeSym}\!\left({#1}\right)}

\newcommand{\distinctInp}[1]{\mathcal{A}^{(#1)}}
\newcommand{\inpDistinctOut}[1]{\mathcal{B}^{(#1)}}

\newcommand{\Ap}{\distinctInp{\roleP}}
\newcommand{\Bp}{\inpDistinctOut{\roleP}}

\newcommand{\act}[1]{\texttt{act}(#1)}

\newcommand{\singleIn}[1]{{\llbracket #1 \rrbracket}_{\text{SI}}}
\newcommand{\singleOut}[1]{{\llbracket #1 \rrbracket}_{\text{SO}}}

\newcommand{\roleset}[1]{\operatorname{role}\left({#1}\right)}

\newcommand{\trCtxSel}{\mathbb{L}_{\oplus}}

\newcommand{\trCtxSelp}{\mathbb{L}^{(\roleP)}_{\oplus}}
\newcommand{\trCtxBrap}{\mathbb{L}^{(\roleP)}_{\&}}
\newcommand{\trCtxSelq}{\mathbb{L}^{(\roleQ)}_{\oplus}}
\newcommand{\trCtxBraq}{\mathbb{L}^{(\roleQ)}_{\&}}

\newcommand{\trCtxGSelpq}{\mathbb{G}^{(\roleP, \roleQ)}_{\oplus}}

\newcommand{\trCtxGBraqp}{\mathbb{G}^{(\roleQ, \roleP)}_{\&}}

\newcommand{\mergeRuleEnd}{\stMerge{}{}\text{-}\stEnd}
\newcommand{\mergeRuleOut}{\stMerge{}{}\text{-}\oplus}
\newcommand{\mergeRuleIn}{\stMerge{}{}\text{-}\&}

\newcommand{\forgetHole}[2]{\langle#1\rangle_{#2}}

\newcommand{\myparagraph}[1]{\paragraph*{\textbf{#1}}}

\newcommand{\ndownarrow}{\mathrel{\centernot\downarrow}}

\newcommand{\iruleLtRec}{LR-$\mu$}\newcommand{\iruleLtOut}{LR-$\stFmt{\oplus}$}\newcommand{\iruleLtIn}{LR-$\stFmt{\&}$}

\newcommand{\prestruct}{\Rrightarrow}

\begin{document}
\title{Asynchronous Global Protocols, Precisely}
\tnotemark[1]

\tnotetext[1]{Work supported by: EPSRC EP/T006544/2, EP/T014709/2,
EP/Z533749/1, ARIA and Horizon EU TaRDIS 101093006 (UKRI number
10066667).} 

\author{Kai Pischke}[orcid=0000-0002-6435-9456]
\ead{kai.pischke@cs.ox.ac.uk}
\author{Jake Masters}[orcid=0000-0002-3783-6149]
\ead{jake.masters@cs.ox.ac.uk}
\author{Nobuko Yoshida}[orcid=0000-0002-3925-8557]
\ead{nobuko.yoshida@cs.ox.ac.uk}

\affiliation{organization={University of Oxford},
city={Oxford},
country={UK}}

\shortauthors{K. Pischke et al.}
\shorttitle{Asynchronous Global Protocols, Precisely}

\begin{abstract}
Asynchronous multiparty session types are a type-based framework 
which ensure the compatibility of components in a distributed system 
by checking compliance against a specified global protocol. 
We propose a \emph{top-down} approach, starting
with the global protocol which is then projected into a set of
local specifications. Next, we use  
an asynchronous refinement relation, \emph{precise asynchronous
multiparty subtyping}, to enable local
specifications to be optimised by permuting actions 
within individual asynchronous components.
This supports local reasoning, as each component can be 
independently developed and \emph{refined} in isolation,  
before being integrated into a larger system.  
We show that this methodology guarantees both
\emph{type soundness} and \emph{liveness} of
the collection of optimised components. 

In this article, we first propose new operational semantics of global protocols 
which capture sound optimisations in the context of asynchronous message-passing.
Next we define an \emph{asynchronous association} between global protocols
and a set of optimised local types.
Thirdly, we prove, for the first time, the correctness of the most expressive 
endpoint projection in the literature, \emph{coinductive full merging projection}.
We then show the main theorems of this article: \emph{soundness} and \emph{completeness} 
of the operational correspondence of the asynchronous association.  
As a consequence, the association acts as an \emph{invariant} 
that can be used to transfer key theorems from the bottom-up system
to the top-down system. In particular, we used this to prove type soundness, session-fidelity, 
deadlock-freedom and liveness of the collection of optimised endpoints.
\end{abstract}


\maketitle

\section{Introduction}
\label{sec:intro}
\subsection{Background}
\label{sec:backgrounds}
\myparagraph{Multiparty Session Types (MPST).}
\emph{Session types} \cite{Takeuchi1994,honda.vasconcelos.kubo:language-primitives} are a type discipline for codifying concurrent 
components. They give an abstract view of \emph{sessions}, which 
are structured communications between distributed peers. 
Session types were first proposed in the context of \emph{binary sessions}, 
such as client-server protocols or dyadic interactions between two
communicating agents. 
\emph{Multiparty session types} \cite{HYC08,HYC2016} (MPST)
extend from two-party to multiparty communications,
allowing the programmer to specify a \emph{global protocol} 
to coordinate communicating components.
Using MPST, we can ensure that typed components
interact without type errors or deadlocks, entirely \emph{by construction}.
The MPST framework is \emph{language-agnostic}, and has been
adapted into over 30 programming languages \cite{Yoshida24}.

In MPST, we first start from specifying a global protocol (called
a \emph{global type}) which we then project  
into a set of local protocols (local types). 
Each local type can then be \emph{refined} to another, compatible local type.
Once each endpoint program is type-checked against its refined local type,
the set of distributed programs will be guaranteed to be interacting correctly, 
without type errors or deadlocks. 

There are a number of methods to refine or optimise local protocols.
In this article, we focus on 
the \textbf{precise multiparty asynchronous subtyping relation} $\asubt$
\cite{GPPSY2023}, as a refinement relation.

\begin{figure}[]
\centering
{\footnotesize
  \begin{tikzpicture}
  \node (Gtext) {A Global Type $\gtG$};
  \node[below= 0mm of Gtext, xshift=-37mm, align=center] (proj) {
    \small {\bf{projection}}\,($\upharpoonright$)};
  \node[below=6mm of Gtext, xshift=-40mm] (LA) {\footnotesize Local Type for $\roleP$
  \small \boxed{\stT_{\roleP}}};
  \node[below=6mm of Gtext] (LB) {\footnotesize Local Type for $\roleQ$  \small \boxed{\stT_{\roleQ}}};
  \node[below=6mm of Gtext, xshift=40mm] (LC) {\footnotesize Local Type for $\roleR$
   \small  \boxed{\stT_{\roleR}}};
  \node[below= 1mm of LA, xshift=-10.4mm, align=center] (typ) {
    \small {\bf{subtype}}\,($\asubt$)\ };
  \node[below=6mm of LA] (TA) {\footnotesize\qquad Subtype for $\roleP$
  \small \boxed{\stToptrole{\roleP}}};
  \node[below=6mm of LB] (TB) {\footnotesize\qquad Subtype for $\roleQ$  \small \boxed{\stToptrole{\roleQ}}};
  \node[below=6mm of LC] (TC) {\footnotesize\qquad Subtype for $\roleR$
   \small  \boxed{\stToptrole{\roleR}}};
  \draw[->] (Gtext) -- (LA);
  \draw[->] (Gtext) -- (LB);
  \draw[->] (Gtext) -- (LC);
   \node[below= 1mm of TA, xshift=-8.4mm, align=center] (typ) {
    \small {\bf{typing}}\,($\vdash$)\ };
   \node[below=6mm of TA, xshift=0mm] (PA) {\, \, \, \footnotesize Program for $\roleP$ \small
    \boxed{P_{\roleP}}};
   \node[below=6mm of TB, xshift=0mm] (PB) {\, \, \, \footnotesize Program for $\roleQ$ \small
    \boxed{P_{\roleQ}}};
     \node[below=6mm of TC, xshift=0mm] (PC) {\, \, \, \footnotesize Program for $\roleR$ \small
    \boxed{P_{\roleR}}};
   \draw[->] (LA) -- (TA);
   \draw[->] (LB) -- (TB);
   \draw[->] (LC) -- (TC);
   \draw[->] (TA) -- (PA);
   \draw[->] (TB) -- (PB);
   \draw[->] (TC) -- (PC);
\end{tikzpicture}
}
\caption{Top-down methodology of multiparty session types. 
$\gtG$ denotes a global type, which is projected into the three
  participants, $\roleP$, $\roleQ$ and $\roleR$, generating local
  types $\stT_{\roleP}$, $\stT_{\roleQ}$ and $\stT_{\roleR}$
  for each participant. 
  Local types are then refined to 
$\stToptrole{\roleP}$, $\stToptrole{\roleQ}$ and $\stToptrole{\roleR}$.
  Three distributed programs $P_{\roleP}$, $P_{\roleQ}$ and $P_{\roleR}$ 
follow.}
\label{fig:overview-of-topdown}
\end{figure}

\myparagraph{Precise Multiparty Asynchronous Subtyping.} 
\emph{Subtyping} enhances the expressiveness of typed programs.
The original synchronous subtyping \cite{Gay2005,DemangeonH11}
has been extended to 
accommodate \emph{asynchronous communications} 
\cite{Chen2017},
allowing the \emph{anticipation} 
of \emph{output actions} for message-passing optimisations.  
We say that a subtyping is 
\emph{precise} if it is both sound and complete--a subtyping relation 
is \emph{sound} if no typable program has a type or communication error. 
It is \emph{complete} if there is no strictly 
larger sound subtyping relation. It is proven that 
the synchronous subtyping relation is precise in both 
binary \cite{Chen2017} and multiparty sessions \cite{Ghilezan2019}. 

In this article, we consider \emph{asynchronous communication} 
which is modelled using infinite FIFO queues: output
processes put messages in queues, while input processes read messages 
from these queues. Hence messages are ordered but non-blocking. 
In asynchronous FIFO semantics, the types should
ensure not only deadlock-freedom, i.e., that input processes will
always find messages, but also \emph{liveness}, i.e., that
all messages in queues will eventually be consumed. 

Chen et al. \cite{Chen2017} have proposed the \emph{asynchronous subtyping relation}, 
where soundness is formulated as ensuring asynchronous liveness. 
The most technically involved but practically useful
subtyping system in this line of research
is proposed by Ghilezan et al.~\cite{GPPSY2023}, which presents the first formalisation of the precise 
subtyping relation for asynchronous MPST. The relation 
is defined based on a \emph{session decomposition} technique, from
full session types (including internal/external choices) into \emph{single
input/output session trees} (without choices). This session 
decomposition expresses the subtyping relation as a composition 
of refinement relations between single input/output trees and provides 
a reasoning principle for optimising the order between
asynchronous messages. See \cite{CDY2024} for a recent survey on
precise subtyping. 

Figure~\ref{fig:overview-of-topdown} describes the MPST workflow.
We assume a set of participants (or roles) $\roleSet$ in the distributed
system.
We specify programs for each participant
$\set{P_{\roleP}}_{\roleP\in\roleSet}$
and a \emph{global protocol (type)} $\gtG$.
To type-check the session
$\mpBigPar{\roleP\in\roleSet}{P_{\roleP}}$ (which denotes
a parallel composition of $P_{\roleP}$) 
against $\gtG$,
we project $\gtG$ onto a set of \emph{local protocols (types)}
$\set{\stT_{\roleP}}_{\roleP\in \roleSet}$ from the viewpoint of
each participant $\roleP$,
and synthesise \emph{local protocols (types)}
$\set{\stToptrole{\roleP}}_{\roleP\in \roleSet}$ from each
program $P_{\roleP}$.
We say that $\gtG$ types $\mpBigPar{\roleP\in\roleSet}{P_{\roleP}}$
iff $\stToptrole{\roleP}$ is a subtype of $\stT_{\roleP}$
\ie $\stT_{\roleP}$ can be refined or optimised to $\stToptrole{\roleP}$.

Asynchronous multiparty subtyping $\asubt$ allows 
for ``safe permutations'' of actions, 
enabling us to type a more optimised program $P_{\roleP}$ 
which conforms to $\stTopt[\roleP]$.
Once each program is typed,
we can automatically guarantee that a collection of distributed programs 
 $\{P_{\roleP}\}_{\roleP\in \roleSet}$ satisfy safety,
deadlock-freedom and liveness.

This top-down workflow 
is implemented by the MPST toolchains, {\sc{Scribble}} 
\cite{YHNN2013} and {$\nu$\sc{Scr}} \cite{YZF2021},  which
check whether a given global protocol is well-formed, and
if so, generate a corresponding set of local types. 
Building on {\sc{Scribble}} and {$\nu$\sc{Scr}}, 
several toolkits have been implemented. 
For instance, the Rust toolchain {\sc{Rumpsteak}} \cite{CYV2022} 
uses {$\nu$\sc{Scr}} to generate state machines,
from which \emph{optimised} APIs are
generated using a sound approximation of $\asubt$.  
We shall discuss the detailed related work on 
theories and implementations of multiparty asynchronous subtyping 
relations in \S~\ref{sec:related}. 
In the next subsection, we 
explain how we use the asynchronous subtyping relation for communication 
optimisations, introducing a running example. 

\subsection{Ring-Choice Example}
\label{sec:ring}
We explain our workflow by introducing 
a running example which will be referenced throughout this article, 
the ring-choice protocol $\gtG_\text{\tiny ring}$ 
from \cite{CutnerYoshida21}:
\begin{equation}
\gtG_\text{\tiny ring} = \gtRec{\gtRecVar}{
  \gtCommSingle{\roleP}{\roleQ}{\mathsf{add}}{\tyInt}
    {
      \gtCommRaw{\roleQ}{\roleR}{
        \begin{array}{@{}l@{}}
        \mathsf{add}(\tyInt)\gtSeq
        \gtCommRaw{\roleR}{\roleP}{
          \mathsf{add}(\tyInt)\gtSeq \gtRecVar
        }
        \\
        \mathsf{sub}(\tyInt)\gtSeq
        \gtCommRaw{\roleR}{\roleP}{
          \mathsf{sub}(\tyInt)\gtSeq \gtRecVar
        }
        \end{array}
      }
  }
}
\end{equation}
\begin{figure}[]
  \centering
  \resizebox{0.5\textwidth}{!}{
  \begin{tikzpicture}[
    x=1cm, y=-1cm, >=Stealth,
    msg/.style={->, thick, -{Stealth[length=3mm,width=2.2mm]}},
    lifeline/.style={densely dashed},
    frame/.style={draw, rounded corners, thick}
  ]
  
\coordinate (p) at (0,0);
  \coordinate (q) at (4,0);
  \coordinate (r) at (8,0);
  
  \node[font=\bfseries] at (p) {$\roleP$};
  \node[font=\bfseries] at (q) {$\roleQ$};
  \node[font=\bfseries] at (r) {$\roleR$};
  
\draw[lifeline] (p) ++(0,0.3) -- ++(0,7.6);
  \draw[lifeline] (q) ++(0,0.3) -- ++(0,7.6);
  \draw[lifeline] (r) ++(0,0.3) -- ++(0,7.6);
  
\draw[frame] (-1.2,0.6) rectangle (9.2,7.4);
  \node[anchor=north west, font=\scriptsize\bfseries] at (-1.15,0.6) {loop $\gtRecVar$};
  
\draw[msg] (0,1.3) -- node[above] {$\gtMsgFmt{add(\tyInt)}$} (4,1.3);
  
\draw[frame] (-0.9,2.0) rectangle (8.9,6.4);
  \node[anchor=north west, font=\scriptsize\bfseries] at (-0.85,2.0) {alt};
  
\draw[thick] (-0.9,4.2) -- (8.9,4.2);
  
\node[anchor=north west, font=\scriptsize] at (-0.75,2.2) {${\color{gray}[{\gtMsgFmt{add(\tyInt)}}]}$};
  \node[anchor=north west, font=\scriptsize] at (-0.75,4.4) {${\color{gray}[{\gtMsgFmt{sub(\tyInt)}}]}$};
  
\draw[msg] (4,2.8) -- node[above] {$\gtMsgFmt{add(\tyInt)}$} (8,2.8);
  \draw[msg] (8,3.6) -- node[above] {$\gtMsgFmt{add(\tyInt)}$} (0,3.6);
  
\draw[msg] (4,5.0) -- node[above] {$\gtMsgFmt{sub(\tyInt)}$} (8,5.0);
  \draw[msg] (8,5.8) -- node[above] {$\gtMsgFmt{sub(\tyInt)}$} (0,5.8);
  
\node[font=\itshape] at (4,6.95) {continue as $\gtRecVar$};
  
  \end{tikzpicture}
  }
  \caption{Message sequence chart for $\gtG[\mathrm{ring}]$ (one unfolding of $\mu \gtRecVar$).}
  \label{fig:ring-choice-mschart}
\end{figure}

The global type $\gtG[\text{\tiny ring}]$ specifies that:
\begin{enumerate}
\item $\roleP$ sends an integer $n$ to $\roleQ$ labelled by $\mathsf{add}$; 

\item $\roleQ$ sends an integer $m$ to $\roleR$ labelled by 
$\mathsf{add}$ or $\mathsf{sub}$; 
\begin{enumerate}
\item if $\mathsf{add}$ is selected, it sends the integer $m+k$ labelled by 
$\mathsf{add}$ to $\roleP$, and the protocol restarts from Step 1; and    
\item if $\mathsf{sub}$ is selected, it sends the integer $m-k$ labelled by 
$\mathsf{sub}$ to $\roleP$, and the protocol restarts from Step 1.   
\end{enumerate}
\end{enumerate}
Global types are semantically similar to more traditional 
protocol descriptions such as \emph{message sequence charts} (MSCs)~\cite{GraubmannRudolphGrabowski1993MSCStandardization,ITU-T-Z120-1996}.
In Figure~\ref{fig:ring-choice-mschart}
we provide a MSC for the ring protocol,
corresponding to $\gtG[\text{\tiny ring}]$,
showing the protocol visually.

If we assume \emph{synchronous interactions} as illustrated in 
Figure~\ref{fig:optimisation}(a),  
no data flow would occur from $\roleQ$ to $\roleR$ 
and from $\roleR$  to $\roleP$ \emph{before} 
$\roleQ$ receives data from $\roleP$. This synchronisation is 
captured by the local types which are projected from $\gtG$: 
\begin{align}
  \stT[\roleP] &= \stRec{\stRecVar}{
    \stIntSum{\roleQ}{}{
      \stChoice{add}{\tyInt} \stSeq
      \stExtSum{\roleR}{}{
        \stChoice{add}{\tyInt} \stSeq \stRecVar,
        \stChoice{sub}{\tyInt} \stSeq \stRecVar
      }
    }
  } \\
  \stT[\roleQ] &= \stRec{\stRecVar}{
    \stExtSum{\roleP}{}{
      \stChoice{add}{\tyInt} \stSeq
      \stIntSum{\roleR}{}{
        \stChoice{add}{\tyInt} \stSeq \stRecVar,
        \stChoice{sub}{\tyInt} \stSeq \stRecVar
      }
    }
  } \\
  \stT[\roleR] &= \stRec{\stRecVar}{
    \stExtSum{\roleQ}{}{
      \stChoice{add}{\tyInt} \stSeq
        \stIntSum{\roleP}{}{
          \stChoice{add}{\tyInt} \stSeq \stRecVar
        },
      \stChoice{sub}{\tyInt} \stSeq
        \stIntSum{\roleP}{}{
          \stChoice{sub}{\tyInt} \stSeq \stRecVar
        }
    }
  }
\end{align}
where the notation $\oplus$ is a \emph{selection type} which denotes 
an internal choice (followed by label and payload), while 
$\&$ denotes a \emph{branching type}, representing an external choice.

\begin{figure}[]
    \centering
    \begin{tabular}{c c}
\begin{tikzpicture}[x=1.1cm,y=1.1cm,->,process/.style={draw,minimum size=11pt}]
            \node[draw=none] (c_1_1) at (0,0) {};
            \node[process] (a_1) at (0.5,0) {$\roleP$};
            \node[process] (b_1) at (1,0) {$\roleQ$};
            \node[process] (c_1) at (1.5,0) {$\roleR$};
            \node[draw=none] (c_2_1) at (0,-0.5) {};
            \node[process] (a_2) at (0.5,-0.5) {$\roleP$};
            \node[process] (b_2) at (1,-0.5) {$\roleQ$};
            \node[process] (c_2) at (1.5,-0.5) {$\roleR$};
            \node[draw=none] (a_2_1) at (2,-0.5) {};
            \node[draw=none] (c_3_1) at (0,-1) {};
            \node[process] (a_3) at (0.5,-1) {$\roleP$};
            \node[process] (b_3) at (1,-1) {$\roleQ$};
            \node[process] (c_3) at (1.5,-1) {$\roleR$};
            \node[draw=none] (a_3_1) at (2,-1) {};
            \node[process] (a_4) at (0.5,-1.5) {$\roleP$};
            \node[process] (b_4) at (1,-1.5) {$\roleQ$};
            \node[process] (c_4) at (1.5,-1.5) {$\roleR$};
            \node[draw=none] (a_4_1) at (2,-1.5) {};
            \node[draw=none] (c_4_1) at (0,-1.5) {};
            \path (a_1) edge node {} (b_2);
            \path (b_2) edge node {} (c_3);
            \path[dashed] (c_3_1) edge node {} (a_4);
        \end{tikzpicture}                &
        \begin{tikzpicture}[x=1.1cm,y=1.1cm,->,process/.style={draw,minimum size=11pt}]
            \node[draw=none] (c_1_1) at (0,0) {};
            \node[process] (a_1) at (0.5,0) {$\roleP$};
            \node[process] (b_1) at (1,0) {$\roleQ$};
            \node[process] (c_1) at (1.5,0) {$\roleR$};
            \node[draw=none] (c_2_1) at (0,-0.5) {};
            \node[process] (a_2) at (0.5,-0.5) {$\roleP$};
            \node[process] (b_2) at (1,-0.5) {$\roleQ$};
            \node[process] (c_2) at (1.5,-0.5) {$\roleR$};
            \node[draw=none] (a_2_1) at (2,-0.5) {};
            \node[draw=none] (c_3_1) at (0,-1) {};
            \node[process] (a_3) at (0.5,-1) {$\roleP$};
            \node[process] (b_3) at (1,-1) {$\roleQ$};
            \node[process] (c_3) at (1.5,-1) {$\roleR$};
            \node[draw=none] (a_3_1) at (2,-1) {};
            \node[process] (a_4) at (0.5,-1.5) {$\roleP$};
            \node[process] (b_4) at (1,-1.5) {$\roleQ$};
            \node[process] (c_4) at (1.5,-1.5) {$\roleR$};
            \node[draw=none] (a_4_1) at (2,-1.5) {};
            \node[draw=none] (c_4_1) at (0,-1.5) {};
            \path (a_1) edge node {} (b_2);
            \path (b_1) edge node {} (c_2);
            \path[dashed] (c_2_1) edge node {} (a_3);
            \path (b_2) edge node {} (c_3);
            \path[dashed] (c_2) edge node {} (a_3_1);
            \path[dashed] (c_3_1) edge node {} (a_4);
            \path (a_3) edge node {} (b_4);
            \path[dashed] (c_3) edge node {} (a_4_1);
        \end{tikzpicture}                  \\
        {\footnotesize\begin{tabular}{c}
            (a) projection of $\gtG_\text{\tiny ring}$
\end{tabular}} &
        {\footnotesize\begin{tabular}{c}
            (b) optimised projection of $\gtG_\text{\tiny ring}$
                \end{tabular}}
    \end{tabular}
\caption{Ring protocol: Projected and optimised interactions 
{\small{(from \cite{CutnerYoshida21})}}}
\label{fig:optimisation}
\end{figure}

Under \emph{asynchronous interactions} illustrated in 
Figure~\ref{fig:optimisation}(b),   
assuming that each participant begins with its own initial value, 
$\roleQ$ can concurrently choose 
one of two labels to send the data 
to $\roleR$  \emph{before} receiving data from $\roleP$, 
letting $\roleR$ and $\roleP$ start the next action. 
By applying asynchronous subtyping ($\asubt$), we can optimise 
$\stT[\roleQ]$ to the following 
$\stTopt[\roleQ]$, pushing 
the external choice behind the internal one:
\begin{align}
  \stTopt[\roleQ] &= \stRec{\stRecVar}{
    \stIntSum{\roleR}{}{
      \stChoice{add}{\tyInt} \stSeq \stExtSum{\roleP}{}{\stChoice{add}{\tyInt} \stSeq \stRecVar},
      \stChoice{sub}{\tyInt} \stSeq \stExtSum{\roleP}{}{\stChoice{add}{\tyInt} \stSeq \stRecVar}
    }
  }
\end{align}
With process $P_{\roleQ}$ typed by $\stTopt[\roleQ]$, we can run
the ring protocol more efficiently (see \cite{CYV2022});
we make this precise in \cref{sec:properties}.
An overview of the history of asynchronous subtyping is given
by Chen et al.~\cite{CDY2024}, encompassing the theory and applications
of the relation.

\paragraph{\bf Associations and Type Soundness}
We say a set of local types $\set{\stTopt[\roleP]}_{\roleP\in \mathcal{P}}$ is
\emph{associated} to a global type $\gtG$ iff
they are refined or optimised by $\asubt$ from $\gtG$'s projection 
$\set{\stT[\roleP]}_{\roleP\in \mathcal{P}}$.
This article proves the \emph{sound} and \emph{complete operational correspondence}
between behaviours of $\gtG$
and $\set{\stTopt[\roleP]}_{\roleP\in \mathcal{P}}$.
This is the link between optimised
programs and global types in
the top-down workflow,
combining projection and optimisation into one step.

More formally, given a typing context 
$\stEnv=
\set{\roleP:(\quH[\roleP],\ \stT_{\roleP}^{\text{\tiny opt}})}_{\roleP\in \gtRoles{\gtG}}$ where $\gtRoles{\gtG}$ is a set of roles in $\gtG$, 
$\quH[\roleP]$ is the type of the queue for participant $\roleP$,   
then the \emph{association} between $\stEnv$
and a global type $\gtG$ is defined as follows:  
\begin{equation}\label{eq:intro-assoc}
  \stEnvAssoc{\stEnv}{\gtG}{a}\text{ if } 
\gtG \upharpoonright_{\roleP} \ (\quHi[\roleP], \stT[\roleP])\text{ and } 
\stToptrole{\roleP} \asubt \stT_{\roleP} \quad \text{ and }
\quH[\roleP]\tySub\quHi[\roleP]
\text{ for all} \ \roleP\in 
\gtRoles{\gtG} 
\end{equation}
where $\quH[\roleP]\tySub\quHi[\roleP]$ extends the asynchronous subtyping to the elements of each queue.
Here, $\gtG \upharpoonright_{\roleP}\ (\quHi[\roleP], \stT[\roleP])$ is the \emph{projection relation}, relating the global type $\gtG$ at participant $\roleP$ to a queue/local-type pair $(\quHi[\roleP], \stT[\roleP])$ that describes $\roleP$'s local behaviour (formally defined later in Def.\ 4). 
Once we obtain the soundness and
completeness of the association, 
we can derive the \emph{subject reduction}  
theorem and \emph{session fidelity} of the top-down approach from 
the corresponding results of the bottom-up system \cite[Theorems~4.11 and 4.13]{GPPSY2023}. The bottom-up system 
does \emph{not} use global types and their projections, but requires an additional
check that the collection of local types (i.e., a typing context) satisfies 
a \emph{safety} property \cite{POPL19LessIsMore}.  

More specifically, we divide the steps to derive these results as follows:
\begin{description}
\item[Step 1] We define the operational semantics of $\gtG$ 
(denoted by $\gtG\,\gtMove\,\gtGi$, which states $\gtG$ moves to
  $\gtGi$ after one communication between two participants) and 
a typing context $\stEnv$
(denoted by $\stEnv \stEnvMove \stEnvi$, which states $\stEnv$ moves to
  $\stEnvi$ after one communication). 
\item[Step 2] We prove \textbf{soundness}: if
$\stEnvAssoc{\stEnv}{\gtG}{a}$ and
$\gtG$ can take a step,
then there exist $\gtGi$ and $\stEnvi$ such
that
$\gtG\,\gtMove\,\gtGi$,
$\stEnv \stEnvMove \stEnvi$ and
$\stEnvAssoc{\stEnvi}{\gtGi}{a}$.
Here, $\gtG\,\gtMove\,\gtGi$ (resp.\ $\stEnv\stEnvMove\stEnvi$) denotes a single communication step of the global type (resp.\ typing context), formally introduced in Def.\ 7. The actual result we prove (Theorem~\ref{thm:gtype:proj-sound}) is in fact slightly stronger, but this extra precision is not needed for the subject reduction and session fidelity results.
  
\item[Step 3] We prove \textbf{completeness}: 
if $\stEnvAssoc{\stEnv}{\gtG}{a}$ and 
$\stEnv \,\stEnvMove\, \stEnvi$, 
then there exists $\gtGi$ such 
  that $\gtG\,\gtMove\,\gtGi$ and
  $\stEnvAssoc{\stEnvi}{\gtGi}{a}$. 

\item[Step 4]
We define the typing rule for multiparty session processes
using the association:
\[
\inference[SessTop]{\forall \roleP\in \dom{\stEnv}\quad  
\stJudge{}{P_{\roleP}}{\stT[\roleP]} \quad 
\stJudge{}{h_{\roleP}}{\quH[\roleP]} \quad 
\stEnv(\roleP)=(\quH[\roleP],\stT[\roleP])\quad 
\stEnvAssoc{\stEnv}{\gtG}{a}}{\stJudgeTop{}{{\Pi_{\roleP\in \dom{\stEnv}}{\ (\roleP \triangleleft 
P_{\roleP} \ | \ \roleP \triangleleft h_{\roleP})\ }}}{\stEnv}}
\]
where $\stJudge{}{P}{\stT}$ is a typing judgement to assign 
type $\stT$ to process $P$ and 
$\stJudge{}{h}{\quH}$ assigns type $\quH$ to a FIFO queue 
$h$ (defined in \cite[Figure~5]{GPPSY2023}). 
$\roleP \triangleleft P_{\roleP}$ means 
process $P_{\roleP}$ is acting as 
participant $\roleP$, buffering sent messages in
its queue $\roleP \triangleleft h_{\roleP}$. 

\item[Step 5] 
We prove \textbf{the subject reduction theorem of the top-down system} using
the completeness of the association with 
the subject reduction theorem of the bottom-up system 
\cite[Theorem~4.11]{GPPSY2023}; and  
\textbf{the session fidelity theorem of the top-down system}
using the soundness and completeness of the association with 
the session fidelity theorem of the bottom-up system 
\cite[Theorem~4.13]{GPPSY2023}.
We additionally prove that all projected typing contexts
are \textbf{safe}, \textbf{deadlock-free} and \textbf{live},
hence, by \cite[Theorem~4.12]{GPPSY2023},
the same is true for all associated typing contexts
(\Cref{thm:assoc-live}).
We give detailed explanations and proofs in \S~\ref{sec:properties}.  
\end{description}

\subsection{Outline and Contributions}
\label{subsec:outline}

We provide the first sound top-down system
of the asynchronous multiparty session types framework, from
global types right through to projected local types, their
asynchronous subtypes, and all the way down to the processes they type, 
including the most flexible form of projection. Importantly, we specify the \emph{operational correspondence} between these levels, including
in the presence of the additional semantic flexibility allowed by asynchrony. 
This enables us to prove subject reduction and session fidelity theorems for the top-down system.

Previous works have studied incomplete parts of this system, which are
disconnected from each other: the asynchronous subtyping relation \cite{GPPSY2023},
extensions to global type syntax to allow for en-route messages \cite{BHYZ2025,DBLP:conf/pldi/Castro-Perez0GY21},
the coinductive full projection relation in the synchronous setting
\cite{thien-nobuko-popl-25}, and  
the liveness in the bottom-up system \cite{POPL19LessIsMore}.

For completing the top-down system with
$\asubt$,
many critical components have been missing, up until now,
including (1) the operational correspondence between global types and
associated typing contexts (\Cref{sec:comp_sound}),
(2) correctness of coinductive full projection
for asynchronous session types (\Cref{sec:projections}),
and (3)
liveness, deadlock-freedom and safety of local types projected from a
global protocol (\Cref{sec:live}).
We fill these gaps, and in doing so we introduce several new technical
contributions,
including a non-deterministic operational semantics for global types
(\Cref{sec:gtype:lts-gt})
that captures the additional behaviours permitted by asynchronous reorderings;
the $\text{balanced}^+$ well-formedness condition and associated
depth functions (\Cref{sec:bal});
forgetful contexts (\Cref{def:global-ctx,def:local-tree-ctx})
which handle the non-determinism arising from asynchronous subtyping
in the soundness proof;
and context projection (\Cref{def:context-proj})
which lifts the projection relation to contexts.

Our major contributions are the soundness and completeness theorems
(\Cref{thm:gtype:proj-sound,thm:gtype:proj-comp})
which establish that the semantics of global types
(\Cref{sec:gtype:lts-gt}) can faithfully mirror the semantics of
any of its associated
typing contexts,
which includes the projected typing contexts
obtained via the coinductive projection
(\Cref{def:global-proj}).
We also prove for the first time that typing contexts
projected from a global type under asynchronous subtyping
are live (and hence safe and deadlock-free)
through the liveness theorem
(\Cref{thm:assoc-live}).
No previous work has established liveness
for projections in the presence of asynchronous subtyping
and coinductive full merging.
We use this to derive the subject reduction and
session fidelity results in \Cref{sec:properties}.

We provide an extensive exploration of global and local types
in~\Cref{sec:gtype:syntax}, 
including syntax, projection, and subtyping. 
We define operational semantics for both global
types (\Cref{sec:gtype:lts-gt}) 
and typing contexts (\Cref{sec:gtype:lts-context}). 
We establish the sound and complete operational relationship 
between these two semantics in~\Cref{sec:gtype:relating}.

This paper provides a significantly
updated and corrected version of the proofs 
of the theorems
from the extended abstract in \cite{PY2025},
as well as new examples and detailed explanations.
Compared to \cite{PY2025},
we have: (1) added full and corrected proofs for all lemmas and
theorems;
(2) corrected issues with the semantics of global types in \cite{PY2025}; 
(3) added new technical devices and syntactic functions to simplify
the proofs; and (4) added new examples to illustrate our technical results.
\section{Multiparty Session Types}
\label{sec:sessiontypes}
This section introduces \emph{global} and \emph{local} types, together with \emph{queue} types. 
As in the work of Barwell et al.~\cite{BHYZ2025}, our formulation of global types includes special runtime-specific constructs to allow 
global types to represent en-route messages which have been sent but not yet received, 
and we give a novel projection relation (\Cref{def:global-proj}) which extends the standard coinductive projection \cite[Definition 3.6]{Ghilezan2019} to 
asynchronous semantics by simultaneously projecting onto both local and queue types. 
We follow this by introducing the \emph{precise asynchronous} subtyping,
which will allow us to preemptively send messages.

\subsection{Global and Local Types}
\label{sec:gtype:syntax}
Multiparty Session Type (MPST) theory uses \emph{global types} to
provide a comprehensive overview of communications between 
\emph{roles}, such as $\roleP, \roleQ, \roleS, \roleT, \ldots$, belonging to a set $\roleSet$.
It employs \emph{local types}, which are obtained via \emph{projection} from a global type,  
to describe how an \emph{individual} role communicates with other roles from a local viewpoint. 
The syntax of global and local types is presented below
where constructs are mostly standard~\cite{POPL19LessIsMore}. \\[2mm]

\centerline{\(
    \begin{array}{r@{\quad}c@{\quad}l@{\quad}l}
     \tyGround & \bnfdef & \tyInt \bnfsep \tyBool \bnfsep \tyReal \bnfsep \tyUnit \bnfsep \ldots
        & \text{\footnotesize Basic types} 
        \\[.5mm]
      \gtG & \bnfdef &
        \gtComm{\roleP}{\roleQ}{i \in I}{\gtLab[i]}{\tyGround[i]}{\gtG[i]}
        &
        {\footnotesize\text{Transmission}}
         \\[.5mm]
         & \bnfsep & \gtCommSquig{\roleP}{\roleQ}{\gtLab}{}{\gtLab}{\tyGround}{\gtG}
         & \text{\footnotesize Transmission en route}  \\[1ex]
        & \bnfsep & \gtRec{\gtRecVar}{\gtG} \quad \bnfsep \quad \gtRecVar 
        \quad \bnfsep \quad \gtEnd
        &
        \text{\footnotesize Recursion, Type variable, Termination}  \\[1ex]
      \stT
        & \bnfdef & \stExtSum{\roleP}{i \in
    I}{\stChoice{\stLab[i]}{\tyGround[i]} \stSeq \stT[i]}
        \quad \bnfsep \quad
        \stIntSum{\roleP}{i \in
    I}{\stChoice{\stLab[i]}{\tyGround[i]} \stSeq \stT[i]} 
          & \text{\footnotesize External and internal choices} 
          \\[.5mm]
& \bnfsep & \stRec{\stRecVar}{\stT} \quad \bnfsep \quad \stRecVar 
        \quad \bnfsep \quad \stEnd
        &
        \text{\footnotesize Recursion, Type variable, Termination} \\[1ex]
     \quH 
       & \bnfdef & \quEmpty \! \bnfsep \! \quMsg{\roleQ}{\stLab}{\tyGround} \! \bnfsep \! \quCons{\quH}{\quH}
       & \text{\footnotesize Empty Queue, Message, Concatenation}
       \\[.5mm]
    \end{array}
  \)}

\smallskip
\noindent
{\bf Basic types}  are taken from a set $\tyGroundSet$, and describe types of values, ranging over integers, booleans, real numbers, units, \etc. 

\smallskip
\noindent
{\bf Global types} range over $\gtG, \gtGi, \gtG[i], \ldots$, 
and describe the high-level behaviour for all roles. 
The set of participants (or roles)
in a global
type $\gtG$ is denoted by $\gtRoles{\gtG}$.  
We explain each syntactic construct of global types. 

\begin{itemize}[left=0pt, topsep=0pt]
\item  
$\gtComm{\roleP}{\roleQ}{i \in
I}{\gtLab[i]}{\tyGround[i]}{\gtG[i]}$: a \emph{transmission}, denoting a
message from role $\roleP$ to role $\roleQ$, 
with a label $\gtLab[i]$, a payload of type $\tyGround[i]$, 
and a continuation $\gtG[i]$, where $i$ is taken from an index
set $I$. We require that the index set be non-empty ($I \neq \emptyset$), labels
$\gtLab[i]$ be pair-wise distinct, and self receptions be excluded (\ie
$\roleP \neq \roleQ$). 
\item  
$\gtCommSquig{\roleP}{\roleQ}{\gtLab}{}{\gtLab}{\tyGround}{\gtG}$: a \emph{transmission en route}, representing a \emph{transmission} 
of the message $\gtLab$ which has already been sent by role $\roleP$ but has not been received by role $\roleQ$. This type 
is only meaningful at runtime. 
We require that self receptions be excluded (\ie
$\roleP \neq \roleQ$). 
\item $\gtRec{\gtRecVar}{G}$: a \emph{recursive} global type, where 
contractive requirements apply~\cite[\S 21.8]{PierceTAPL}, \ie 
each recursion variable $\gtRecVar$ is bound within a $\gtRec{\gtRecVar}{\ldots}$ 
and is guarded. 
\item $\gtEnd$: a \emph{terminated} global type (omitted where unambiguous).
\end{itemize}

\noindent
{\bf Local types} 
(or \emph{session types}) range over $\stT, \stTi, \stT[i], \ldots$, and describe
the behaviour of a single role. We elucidate each syntactic construct of local types.

An \emph{internal choice}~(\emph{selection}), $\stIntSum{\roleP}{i \in I}{\stChoice{\stLab[i]}{\tyGround[i]} \stSeq \stT[i]}$, 
indicating that the \emph{current} role is expected to \emph{send} to role $\roleP$;
an \emph{external choice} (\emph{branching}), $\stExtSum{\roleP}{i \in I}{\stChoice{\stLab[i]}{\tyGround[i]} \stSeq \stT[i]}$, 
indicating that the \emph{current} role is expected to \emph{receive} from role $\roleP$; 
a \emph{recursive} local type $\stRec{\stRecVar}{\stT}$, following a pattern analogous to $\gtRec{\gtRecVar}{G}$; 
a \emph{termination} $\stEnd$ (omitted where
unambiguous). Similar to global types, local types also need  pairwise distinct, non-empty labels.

\noindent
{\bf Queue types} range over $\quH, \quHi, \quH[i], \ldots$, and describe the type of queues storing buffered asynchronous messages:  
$\quEmpty[]$ is the empty queue;  
$\quMsg{\roleP}{\stLab}{\tyGround}$ is the type of a queued message being sent to participant $\roleP$ 
with a message label $\stLab$ and a payload of type $\tyGround$; and 
$\quCons{\quH}{\quHi}$ is the concatenation of two queues. 
We consider queue types up-to associativity and we allow queue elements with distinct destinations to commute~\cref{def:queue-equiv}, this simulates the existence of a queue for each potential recipient. We take the convention that
$\quCons{\quMsg{\roleP}{\stLab}{\tyGround}}{\quH}$ is a queue with
head $\quMsg{\roleP}{\stLab}{\tyGround}$,
which can be dequeued, and tail $\quH$,
and that enqueueing $\quMsg{\roleP}{\stLabi}{\tyGroundi}$ to $\quHi$
gives us $\quCons{\quHi}{\quMsg{\roleP}{\stLabi}{\tyGroundi}}$
\ie queues are read left-to-right.

\begin{definition}[Queue Equivalence~\cite{GPPSY2023}]
	\label{def:queue-equiv}
	We define the relation $\equiv$ inductively on syntactic queues by:
\[
\begin{array}{c}
\inference[]{}{\quH\equiv\quH}
\qquad
\inference[]{}{\quCons{\quEmpty}{\quH}\equiv\quH}
\qquad
\inference[]{}{\quCons{\quH}{\quEmpty}\equiv\quH}
\qquad
\inference[]
{}
{\quCons{(\quCons{\quH}{\quHi})}{\quHiii}
\equiv \quCons{\quH}{(\quCons{\quHi}{\quHiii})}}
\\[2.5ex]
\inference[]
{\quH[L]\equiv\quHi[L] \quad \quH[R]\equiv\quHi[R]}
{\quCons{\quH[L]}{\quH[R]}\equiv\quCons{\quHi[L]}{\quHi[R]}}
\qquad
\inference[]
{\roleP\ne\roleQ}
{\quCons{\quMsg{\roleP}{\stLab}{\tyGround}}{\quMsg{\roleQ}{\stLabi}{\tyGroundi}}
\equiv 
\quCons{\quMsg{\roleQ}{\stLabi}{\tyGroundi}}{\quMsg{\roleP}{\stLab}{\tyGround}}}
\end{array}
\]
\end{definition}
  
\begin{definition}[Unfolding]
	\label{def:unfolding}
	We define the unfolding operator on global types and local types recursively to unfold recursive types.
	We define $\unfoldOne{\gtRec{\gtRecVar}{\gtG}}=\unfoldOne{\gtG{}[\gtRec{\gtRecVar}{\gtG}/\gtRecVar]}$
	and $\unfoldOne{\gtG}=\gtG$ otherwise.
	We define $\unfoldOne{\stRec{\stRecVar}{\stT}}=\unfoldOne{\stT{}[\stRec{\stRecVar}{\stT}/\stRecVar]}$
	and $\unfoldOne{\stT}=\stT$ otherwise.
\end{definition}

Below we define sets of roles for a given global type. 

\begin{definition}[Role Functions]
\label{def:roles}\ \\
\begin{itemize}
\item
We define $\gtRoles{\gtG}$
to be \emph{the set of participants} in $\gtG$
by induction on the structure of $\gtG$: 
$\gtRoles{\gtEnd} = \emptyset$; 
$\gtRoles{\gtRecVar} = \emptyset$; 
$\gtRoles{\gtRec{\gtRecVar}{\gtG}} = \gtRoles{\gtG}$; 
$\gtRoles{\gtComm{\roleP}{\roleQ}{i \in
I}{\gtLab[i]}{\tyGround[i]}{\gtG[i]}}
= \{\roleP,\roleQ\}\cup\bigcup_{i\in I}\gtRoles{\gtG[i]}$; and 
$\gtRoles{\gtCommSquig{\roleP}{\roleQ}{\gtLab}{}{\gtLab}{\tyGround}{\gtG}}
= \{\roleP,\roleQ\}\cup\gtRoles{\gtG}$. 
\item
We define $\gtSRoles{\gtG}$
to be the set of participants in $\gtG$
that have sent an en-route transmission
by induction on the structure of $\gtG$:
$\gtSRoles{\gtEnd} = \emptyset$; 
$\gtSRoles{\gtRecVar} = \emptyset$; 
$\gtSRoles{\gtRec{\gtRecVar}{\gtG}} = \gtSRoles{\gtG}$; 
$\gtSRoles{\gtComm{\roleP}{\roleQ}{i \in
I}{\gtLab[i]}{\tyGround[i]}{\gtG[i]}} = \bigcup_{i\in
I}\gtSRoles{\gtG[i]}$; and 
$\gtSRoles{\gtCommSquig{\roleP}{\roleQ}{\gtLab}{}{\gtLab}{\tyGround}{\gtG}} = \{\roleP\}\cup\gtSRoles{\gtG}$
\item
We define $\gtARoles{\gtG}$
to be the set of participants in $\gtG$
that may perform actions,
\ie appears either as the receiver of a transmission
or as the sender of a non-en-route transmission,
by induction on the structure of $\gtG$:
$\gtARoles{\gtEnd} = \emptyset$; 
$\gtARoles{\gtRecVar} = \emptyset$; 
$\gtARoles{\gtRec{\gtRecVar}{\gtG}} = \gtARoles{\gtG}$; 
$\gtARoles{\gtComm{\roleP}{\roleQ}{i \in
I}{\gtLab[i]}{\tyGround[i]}{\gtG[i]}}
= \{\roleP,\roleQ\}\cup\bigcup_{i\in I}\gtARoles{\gtG[i]}$;
and $\gtARoles{\gtCommSquig{\roleP}{\roleQ}{\gtLab}{}{\gtLab}{\tyGround}{\gtG}} = \{\roleQ\}\cup\gtARoles{\gtG}$
\end{itemize}
\end{definition}

$\gtSRoles{\gtG}$ are the participants that have non-empty queues
once projected, and
$\gtARoles{\gtG}$ are the participants that have non-terminated local types
once projected.

We will now use the ring protocol global type,
and an intermediate global type,
to demonstrate unfolding
and the role functions.
\begin{example}[Ring Protocol]
Recall the ring protocol (\cref{sec:ring}):
\begin{equation}
\gtG_\text{\tiny ring} = \gtRec{\gtRecVar}{
  \gtCommSingle{\roleP}{\roleQ}{\mathsf{add}}{\tyInt}
    {
      \gtCommRaw{\roleQ}{\roleR}{
        \begin{array}{@{}l@{}}
        \mathsf{add}(\tyInt)\gtSeq
        \gtCommRaw{\roleR}{\roleP}{
          \mathsf{add}(\tyInt)\gtSeq \gtRecVar
        }
        \\
        \mathsf{sub}(\tyInt)\gtSeq
        \gtCommRaw{\roleR}{\roleP}{
          \mathsf{sub}(\tyInt)\gtSeq \gtRecVar
        }
        \end{array}
      }
  }
}
\end{equation}
We have that $\gtRoles{\gtG_{\text{\tiny ring}}}=
\gtARoles{\gtG_{\text{\tiny ring}}}=\{\roleP,\roleQ,\roleR\}$
and $\gtSRoles{\gtG_{\text{\tiny ring}}}=\emptyset$.
We also have that
\begin{equation}
\begin{adjustbox}{width=\columnwidth,center}\(
\unfoldOne{\gtG_\text{\tiny ring}} = 
{
  \gtCommSingle{\roleP}{\roleQ}{\mathsf{add}}{\tyInt}
    {
      \gtCommRaw{\roleQ}{\roleR}{
        \begin{array}{@{}l@{}}
        \mathsf{add}(\tyInt)\gtSeq
        \gtCommRaw{\roleR}{\roleP}{
          \mathsf{add}(\tyInt)\gtSeq \gtRec{\gtRecVar}{
  \gtCommSingle{\roleP}{\roleQ}{\mathsf{add}}{\tyInt}
    {
      \gtCommRaw{\roleQ}{\roleR}{
        \begin{array}{@{}l@{}}
        \mathsf{add}(\tyInt)\gtSeq
        \gtCommRaw{\roleR}{\roleP}{
          \mathsf{add}(\tyInt)\gtSeq \gtRecVar
        }
        \\
        \mathsf{sub}(\tyInt)\gtSeq
        \gtCommRaw{\roleR}{\roleP}{
          \mathsf{sub}(\tyInt)\gtSeq \gtRecVar
        }
        \end{array}
      }
  }
}
        }
        \\
        \mathsf{sub}(\tyInt)\gtSeq
        \gtCommRaw{\roleR}{\roleP}{
          \mathsf{sub}(\tyInt)\gtSeq \gtRec{\gtRecVar}{
  \gtCommSingle{\roleP}{\roleQ}{\mathsf{add}}{\tyInt}
    {
      \gtCommRaw{\roleQ}{\roleR}{
        \begin{array}{@{}l@{}}
        \mathsf{add}(\tyInt)\gtSeq
        \gtCommRaw{\roleR}{\roleP}{
          \mathsf{add}(\tyInt)\gtSeq \gtRecVar
        }
        \\
        \mathsf{sub}(\tyInt)\gtSeq
        \gtCommRaw{\roleR}{\roleP}{
          \mathsf{sub}(\tyInt)\gtSeq \gtRecVar
        }
        \end{array}
      }
  }
}
        }
        \end{array}
      }
  }
}\)
\end{adjustbox}
\end{equation}
\end{example}

\begin{example}[Ring Protocol II]
Consider the global type that occurs during the ring protocol:
\begin{equation}
\gtG=\gtCommSquigSingle{\roleP}{\roleQ}{\gtMsgFmt{add}}{\mathsf{add}}{\tyInt}
      {
        \gtCommRaw{\roleQ}{\roleR}{
          \begin{array}{@{}l@{}}
          \mathsf{add}(\tyInt)\gtSeq
          \gtCommRaw{\roleR}{\roleP}{
            \mathsf{add}(\tyInt)\gtSeq  \gtG_\text{ring}
          }
          \\
          \mathsf{sub}(\tyInt)\gtSeq
          \gtCommRaw{\roleR}{\roleP}{
            \mathsf{sub}(\tyInt)\gtSeq  \gtG_\text{ring}
          }
          \end{array}
        }
    }
\end{equation}
We have that $\gtRoles{\gtG}=\gtARoles{\gtG}=\{\roleP,\roleQ,\roleR\}$,
$\gtSRoles{\gtG}=\{\roleP\}$,
and $\unfoldOne{\gtG}=\gtG$.
\end{example}

\subsection{Projections}
\label{sec:projections}
\noindent In the top-down approach of MPST, local types are obtained
by projecting a global type onto roles. Whereas projection is traditionally presented as a \emph{partial function} from a global type and a participant to a local type~\cite{HYC2016}, our \Cref{def:global-proj} below instead defines it as a \emph{coinductive ternary relation} $\gtProjRel{\gtG}{\roleP}{(\quH, \stT)}$ between a global type $\gtG$, a participant $\roleP$, and a (queue type, local type) pair. It should be read as a \emph{predicate} which holds when $(\quH, \stT)$ is a valid projection of $\gtG$ at $\roleP$. The queues $\quH$ allow us to capture buffered messages from en-route transmissions at the local level.

The rules are organised by the relationship between the projected role $\roleP$ and the outermost interaction of $\gtG$. When $\roleP$ is the sender, its local type begins with an internal choice (\RULE{P-$\oplus$}); when $\roleP$ is the receiver, with an external choice (\RULE{P-$\&$}); when $\roleP$ is not involved, the continuations from each branch must be reconciled into a single local type using merge (\RULE{P-$\stBinMerge$}). Each of these three cases has a corresponding rule (suffixed~-II) that handles the additional presence of an en-route transmission in $\gtG$, extracting the relevant queue prefix before proceeding. Finally, \RULE{P-End} and \RULE{P-$\quEmpty$} cover termination.

\begin{definition}[Global Type Projection]\label{def:global-proj}\label{def:local-type-merge}The \emph{projection relation} $\gtProjRel{\gtG}{\roleP}{(\quH, \stT)}$ (read ``$(\quH, \stT)$ is a projection of the global type $\gtG$ onto the participant $\roleP$'') is the largest relation such that whenever $\gtProjRel{\gtG}{\roleP}{(\quH, \stT)}$ holds:
	\begin{itemize}[leftmargin=0.5in,labelindent=-\leftmargin]
		\item[\inferrule{P-End}] If $\roleP\notin\gtARoles{\gtG}$, then $\unfoldOne{\stT}=\stEnd$
		\item[\inferrule{P-$\quEmpty$}] If $\roleP\notin\gtSRoles{\gtG}$, then $\quH=\quEmpty$ \item[\inferrule{P-$\stBinMerge$}] If $\unfoldOne{\gtG}=\gtCommSmall{\roleQ}{\roleR}{i \in I}{\gtLab[i]}{\tyGround[i]}{\gtG[i]}$, then for all $i\in I$
		there exist $\stT[i]$ such that $\gtProjRel{\gtG[i]}{\roleP}{\left(\quH, \stT[i]\right)}$ and
		$\stMerge{}{\{\stT[i]\}_{i\in I}} \ni \stT$
		\item[\inferrule{P-$\stBinMerge$-II}] If $\unfoldOne{\gtG}=\gtCommSquigSmall{\roleQ}{\roleR}{\gtLab}{}{\gtLab}{\tyGround}{\gtGi}$, then
		$\gtProjRel{\gtGi}{\roleP}{\left(\quH, \stT\right)}$
		\item[\inferrule{P-$\oplus$}] If $\unfoldOne{\gtG}=\gtCommSmall{\roleP}{\roleQ}{i \in I}{\gtLab[i]}{\tyGround[i]}{\gtG[i]}$, then
		for all $i\in I$
		there exist $\stT[i]$ such that $\gtProjRel{\gtG[i]}{\roleP}{\left(\quH, \stT[i]\right)}$ and
		$\unfoldOne{\stT}=\stIntSum{\roleQ}{i \in I}{\stChoice{\stLab[i]}{\tyGround[i]} \stSeq \stT[i]}$
		\item[\inferrule{P-$\oplus$-II}] If $\unfoldOne{\gtG}=\gtCommSquigSmall{\roleP}{\roleQ}{\gtLab}{}{\gtLab}{\tyGround}{\gtGi}$, then
		there exists $\quHi$ such that $\gtProjRel{\gtGi}{\roleP}{\left(\quHi, \stT\right)}$, and
		$\quH=\quCons{\quMsg{\roleQ}{\gtLab[j]}{\tyGround[j]}}{\quHi}$
		\item[\inferrule{P-$\&$}] If $\unfoldOne{\gtG}=\gtCommSmall{\roleQ}{\roleP}{i \in I}{\gtLab[i]}{\tyGround[i]}{\gtG[i]}$, then
		for all $i\in I$
		there exist $\stT[i]$ such that $\gtProjRel{\gtG[i]}{\roleP}{\left(\quH, \stT[i]\right)}$ and
		$\unfoldOne{\stT}=\stExtSum{\roleQ}{i \in I}{\stChoice{\stLab[i]}{\tyGround[i]} \stSeq \stT[i]}$
		\item[\inferrule{P-$\&$-II}] If $\unfoldOne{\gtG}=\gtCommSquigSmall{\roleQ}{\roleP}{\gtLab}{}{\gtLab}{\tyGround}{\gtGi}$, then
		there exists $\stTi$ such that $\gtProjRel{\gtGi}{\roleP}{\left(\quH, \stTi\right)}$ and
		$\unfoldOne{\stT}=\stExtSum{\roleQ}{}{\stChoice{\stLab}{\tyGround} \stSeq \stTi}$
	\end{itemize}

\noindent
In rule \RULE{P-$\stBinMerge$}, a role not involved in a choice cannot observe which branch was taken. The projection must therefore reconcile the continuations from each branch into a single local type that is compatible with all of them. This is the role of $\stMerge{}{}$, the \emph{merge operation for session types} (\emph{full merging}).
$\stMerge{}{}$
is a partial function that takes
a finite and non-empty set of local types, $\mathcal{T}$,
and returns
a set of local types, $\stMerge{}{\mathcal{T}}$.
We say that $\stT$ is a merge of $\mathcal{T}$
(equivalently, $\stT$ is a merge of the types in $\mathcal{T}$)
to mean that $\stT\in\stMerge{}{\mathcal{T}}$.
We say that $\mathcal{T}$ is mergeable
(equivalently, the types in $\mathcal{T}$ are mergeable)
if $\stMerge{}{\mathcal{T}}\neq\emptyset$.
In the case of indexed sets,
we may write $\stMerge{i\in I}{\stT[i]}$
for $\stMerge{}{\{\stT[i]\suchthat i\in I\}}$.

\noindent
We define the merge operation using the coinductive relation $\stMergeRel{}{}$ below,
by defining
$\stMerge{}{\mathcal{T}}=\{\stT\suchthat\stMergeRel{\mathcal{T}}{\stT}\}$.
We define $\stMergeRel{}{}$ to be
the largest relation between non-empty finite sets of local types
and local types such that if $\stMergeRel{\{\stT[i]\suchthat i\in I\}}{\stT}$, then:
\begin{itemize}[leftmargin=0.5in,labelindent=-\leftmargin]
    \item[\inferrule{$\mergeRuleEnd$}]
        If $\unfoldOne{\stT}=\stEnd$, then
        $\unfoldOne{\stT[i]}=\stEnd$ for all $i\in I$
    \item[\inferrule{$\mergeRuleOut$}]
        If $\unfoldOne{\stT}=\stIntSum{\roleQ}{j \in J}{\stChoice{\stLab[j]}{\stS[j]} \stSeq \stTi[j]}$, then for all $i\in I$
        $\unfoldOne{\stT[i]}=\stIntSum{\roleQ}{j \in J}{\stChoice{\stLab[j]}{\stS[j]} \stSeq \stT[i,j]}$
        where for all $j\in J$ $\stMergeRel{\left\{\stT[i,j]\suchthat i\in I\right\}}{\stTi[j]}$
    \item[\inferrule{$\mergeRuleIn$}]
        If $\unfoldOne{\stT} = \stExtSum{\roleP}{j \in J}{\stChoice{\stLab[j]}{\stS[j]}\stSeq \stTi[j]}$, then for all $i\in I$
        $\unfoldOne{\stT[i]}=\stExtSum{\roleP}{j \in J_{i}}{\stChoice{\stLab[j]}{\stS[j]} \stSeq \stT[i,j]}$
        where $J=\bigcup_{i\in I}J_{i}$ and for all $j\in J$
        $\stMergeRel{\left\{\stT[i,j]\suchthat j\in J_{i} \right\}}{\stTi[j]}$
\end{itemize}
\noindent
 Here $\unfoldOne{\cdot}$ is the unfolding function of \cref{def:unfolding}.
\\[2ex]
When the queue is empty,
we may leave it implicit \ie
if $\gtProjRel{\gtG}{\roleP}{(\quEmpty,\stT)}$ then
we may instead write $\gtProjRel{\gtG}{\roleP}{\stT}$.
\end{definition}

The rule \RULE{P-End} says that if a global type doesn't specify behaviour for some role, then its projection has no behaviour.
The rule \RULE{P-$\quEmpty$} is similar, it says that if a role doesn't add anything to the queue,
then its queue must be empty.
The rule \RULE{P-$\stBinMerge$} says that if the head of the global type does not involve
a participant and
the projections of the continuations are compatible, then
the projection of the global type is their merge.
The rule \RULE{P-$\stBinMerge$-II} says that
if the head of a global type is an en-route transmission not involving $\roleP$, then
the projection onto $\roleP$ is the projection of the continuation.
The rule \RULE{P-$\oplus$} says that if the head of a global type is a communication with sender $\roleP$ and destination $\roleQ$,
then the head of the projection onto $\roleP$ is a send to $\roleQ$ with the same messages and payloads,
and the local continuations are the projections of the global continuations.
The rule \RULE{P-$\oplus$-II}  
allows an en-route message $\quMsg{\roleQ}{\stLab[j]}{\tyGround[j]}$ to be included in the projected queue of outgoing messages.  
If a global type $\gtG$ starts with a transmission from role $\roleP$ to role $\roleQ$, projecting it onto 
role $\roleP$~(resp. $\roleQ$) results in an internal (resp. external) choice, provided that the continuation 
$\gtG$ is also projectable.
The rules \RULE{P-\&} and \RULE{P-\&-II} say that if the head of a global type is a communication with sender $\roleQ$ and destination $\roleP$,
then the head of the projection onto $\roleP$ is a receive from $\roleQ$ with the same messages and payloads,
and the local continuations are the projections of the global continuations.
When projecting $\gtG$ onto other participants $\roleR$~($\roleR \neq \roleP$ and $\roleR \neq \roleQ$), a merge operator, as defined in~\Cref{def:global-proj},
is used to ensure that the projections of all continuations are ``compatible''. It is noteworthy that there are global types that cannot be projected 
onto all of their participants as shown by Udomsrirungruang and Yoshida~\cite[\S4.4]{thien-nobuko-popl-25}.
\\[2ex]
We now return to our running example
to demonstrate a
derivation of
an instance of a projection and a merge.
\begin{example}[Ring Protocol Projection]
\label{ex:ring-proj}
Recall $\stT[\roleP]$ from the ring-choice example~(\cref{sec:ring}):
\begin{equation}
\stT[\roleP]= \stRec{\stRecVar}{
    \stIntSum{\roleQ}{}{
      \stChoice{add}{\tyInt} \stSeq
      \stExtSum{\roleR}{}{
        \stChoice{add}{\tyInt} \stSeq \stRecVar,
        \stChoice{sub}{\tyInt} \stSeq \stRecVar
      }
    }
  }
\end{equation}
We can derive $\gtProjRel{\gtG_\text{ring}}{\roleP}{\stT[\roleP]}$
by applying
${\color{ruleColor}\textsc{\scriptsize [P-\ensuremath{\oplus}]}}$, 
${\color{ruleColor}\textsc{\scriptsize [P-\ensuremath{\stBinMerge}]}}$, 
and merging the results of applying ${\color{ruleColor}\textsc{\scriptsize [P-\ensuremath{\&}]}}$, to the coinductive hypothesis for each branch.
\begin{equation}
	\cinference[P-$\oplus$]
	{
		\cinference[P-$\stBinMerge$]
		{
			\cinference[P-$\&$]{\gtProjRel{\gtG_\text{ring}}{\roleP}{\stT[\roleP]}}
			{
				\gtProjRel
				{
					\gtCommRaw{\roleR}{\roleP}{
					\mathsf{add}(\tyInt)\gtSeq \gtG_\text{ring}
					}
				}{\roleP}
				{
					\stExtSum{\roleR}{}{\stChoice{add}{\tyInt} \stSeq \stT[\roleP]}
				}
			}
			\cinference[P-$\&$]{\gtProjRel{\gtG_\text{ring}}{\roleP}{\stT[\roleP]}}
			{
				\gtProjRel
				{
					\gtCommRaw{\roleR}{\roleP}{
					\mathsf{sub}(\tyInt)\gtSeq \gtG_\text{ring}
					}
				}{\roleP}
				{
					\stExtSum{\roleR}{}{\stChoice{sub}{\tyInt} \stSeq \stT[\roleP]}
				}
			}
		}
		{
		\gtProjRel{
			\gtCommRaw{\roleQ}{\roleR}{
				\begin{array}{@{}l@{}}
					\mathsf{add}(\tyInt)\gtSeq
					\gtCommRaw{\roleR}{\roleP}{
					\mathsf{add}(\tyInt)\gtSeq \gtG_\text{ring}
					}
					\\
					\mathsf{sub}(\tyInt)\gtSeq
					\gtCommRaw{\roleR}{\roleP}{
					\mathsf{sub}(\tyInt)\gtSeq \gtG_\text{ring}
					}
				\end{array}
			}
		}
		{\roleP}
		{
			\stExtSum{\roleR}{}{
				\stChoice{add}{\tyInt} \stSeq \stT[\roleP],
				\stChoice{sub}{\tyInt} \stSeq \stT[\roleP]
			}
		}
		}
	}
	{\gtProjRel{\gtG_\text{ring}}{\roleP}{\stT[\roleP]}}
\end{equation}
Noting that:
\begin{equation}
\stMerge{}{\{
\stExtSum{\roleR}{}{\stChoice{add}{\tyInt} \stSeq \stT[\roleP]},
\stExtSum{\roleR}{}{\stChoice{sub}{\tyInt} \stSeq \stT[\roleP]}\}}\ni
\stExtSum{\roleR}{}{
	\stChoice{add}{\tyInt} \stSeq \stT[\roleP],
	\stChoice{sub}{\tyInt} \stSeq \stT[\roleP]
}
\end{equation}
\begin{equation}
	\cinference[$\stBinMerge$-$\&$]
	{
		\stMergeRel{\{\stT[\roleP]\}}{\stT[\roleP]}
		&
		\stMergeRel{\{\stT[\roleP]\}}{\stT[\roleP]}
	}
	{
		\stMergeRel{\{
		\stExtSum{\roleR}{}{\stChoice{add}{\tyInt} \stSeq \stT[\roleP]},
		\stExtSum{\roleR}{}{\stChoice{sub}{\tyInt} \stSeq \stT[\roleP]}\}}{
		\stExtSum{\roleR}{}{
			\stChoice{add}{\tyInt} \stSeq \stT[\roleP],
			\stChoice{sub}{\tyInt} \stSeq \stT[\roleP]
		}}
	}
\end{equation}
\end{example}
Of course, this projection
is derivable by the full inductive projection,
which is a subrelation of the full coinductive projection;
this is not always the case.
We now provide a more involved
example,
which
makes use of en-route transitions
and cannot be projected inductively.
\begin{example}[Coinductive Projection]
\label{ex:coind-proj}
Consider the global type:
\begin{equation}
	\gtG =
	\gtCommRaw{\roleP}{\roleR}{
		\begin{array}{l}
			\gtLab[1]\gtSeq \gtCommSquigSingle{\roleQ}{\roleP}{\gtLab}{\gtLab}{}{
				\gtRec{\gtRecVar}{
					\gtCommSingle{\roleP}{\roleQ}{\gtLab[1]}{}{
						\gtRecVar
					}
				}
			}\\
			\gtLab[2]\gtSeq \gtCommSquigSingle{\roleQ}{\roleP}{\gtLab}{\gtLab}{}{
				\gtCommSingle{\roleP}{\roleQ}{\gtLab[2]}{}{
					\gtRec{\gtRecVar}{
						\gtCommSingle{\roleP}{\roleQ}{\gtLab[1]}{}{
							\gtRecVar
						}
					}
				}
			}
		\end{array}
	}
\end{equation}
$\gtG$ projects onto all three roles. The projections onto $\roleR$ and $\roleP$ are obtained directly from the coinductive rules without any merging.
\begin{equation}
	\cinference[P-$\&$]
	{
		\cinference[P-End]{}{
			\gtProjRel{
			\gtCommSquigSingle{\roleQ}{\roleP}{\gtLab}{\gtLab}{}{
				\gtRec{\gtRecVar}{
					\gtCommSingle{\roleP}{\roleQ}{\gtLab[1]}{}{
						\gtRecVar
					}
				}
			}
		}{\roleR}{\stEnd}
		} &
		\cinference[P-End]{}{
			\gtProjRel{
			\gtCommSquigSingle{\roleQ}{\roleP}{\gtLab}{\gtLab}{}{
				\gtCommSingle{\roleP}{\roleQ}{\gtLab[2]}{}{
					\gtRec{\gtRecVar}{
						\gtCommSingle{\roleP}{\roleQ}{\gtLab[1]}{}{
							\gtRecVar
						}
					}
				}
			}
		}{\roleR}{\stEnd}
		}
	}
	{\gtProjRel{\gtG}{\roleR}
	{
		\stExtSum{\roleP}{}{\stLab[1],\stLab[2]}
	}
	}
\end{equation}
\begin{equation}
\begin{adjustbox}{width=\columnwidth,center}\(
	\cinference[P-$\oplus$]
	{
		\cinference[P-$\&$-II]
		{
			\cinference[P-$\oplus$]
			{
				\gtProjRel
				{\gtRec{\gtRecVar}{
					\gtCommSingle{\roleP}{\roleQ}{\gtLab[1]}{}{
						\gtRecVar
					}
				}}
				{\roleP}
				{\stRec{\stRecVar}{\stIntSum{\roleQ}{}{\stLab[1]\stSeq\stRecVar}}}
			}
			{
				\gtProjRel
				{\gtRec{\gtRecVar}{
					\gtCommSingle{\roleP}{\roleQ}{\gtLab[1]}{}{
						\gtRecVar
					}
				}}
				{\roleP}
				{\stRec{\stRecVar}{\stIntSum{\roleQ}{}{\stLab[1]\stSeq\stRecVar}}}
			}
		}
		{
			\gtProjRel
			{
				\gtCommSquigSingle{\roleQ}{\roleP}{\gtLab}{\gtLab}{}{
				\gtRec{\gtRecVar}{
					\gtCommSingle{\roleP}{\roleQ}{\gtLab[1]}{}{
						\gtRecVar
					}
				}
			}
			}
			{\roleP}
			{
				\stExtSum{\roleQ}{}{
						\stLab\stSeq
						\stRec{\stRecVar}{\stIntSum{\roleQ}{}{\stLab[1]\stSeq\stRecVar}}
					}
			}
		}
		&
		\cinference[P-$\&$-II]
		{
			\cinference[P-$\oplus$]
			{
				\cinference[P-$\oplus$]
				{
					\gtProjRel{
					\gtRec{\gtRecVar}{
						\gtCommSingle{\roleP}{\roleQ}{\gtLab[1]}{}{
							\gtRecVar
						}
					}
				}
				{\roleP}
				{
					\stRec{\stRecVar}{\stIntSum{\roleQ}{}{\stLab[1]\stSeq\stRecVar}}
				}
				}
				{\gtProjRel{
					\gtRec{\gtRecVar}{
						\gtCommSingle{\roleP}{\roleQ}{\gtLab[1]}{}{
							\gtRecVar
						}
					}
				}
				{\roleP}
				{
					\stRec{\stRecVar}{\stIntSum{\roleQ}{}{\stLab[1]\stSeq\stRecVar}}
				}}
			}
			{
				\gtProjRel
				{
					\gtCommSingle{\roleP}{\roleQ}{\gtLab[2]}{}{
					\gtRec{\gtRecVar}{
						\gtCommSingle{\roleP}{\roleQ}{\gtLab[1]}{}{
							\gtRecVar
						}
					}
				}
				}
				{\roleP}
				{
					\stIntSum{\roleQ}{}{\stLab[2]\stSeq
						\stRec{\stRecVar}{\stIntSum{\roleQ}{}{\stLab[1]\stSeq\stRecVar}}}
				}
			}
		}
		{
			\gtProjRel{
			\gtCommSquigSingle{\roleQ}{\roleP}{\gtLab}{\gtLab}{}{
				\gtCommSingle{\roleP}{\roleQ}{\gtLab[2]}{}{
					\gtRec{\gtRecVar}{
						\gtCommSingle{\roleP}{\roleQ}{\gtLab[1]}{}{
							\gtRecVar
						}
					}
				}
			}
			}{\roleP}
			{
				\stExtSum{\roleQ}{}{
						\stLab\stSeq
						\stIntSum{\roleQ}{}{\stLab[2]\stSeq
						\stRec{\stRecVar}{\stIntSum{\roleQ}{}{\stLab[1]\stSeq\stRecVar}}}
					}
			}
		}
	}
	{
		\gtProjRel{\gtG}{\roleP}
		{
			\stIntSum{\roleR}{}{
				\begin{array}{l}
					\stLab[1]\stSeq
					\stExtSum{\roleQ}{}{
						\stLab\stSeq
						\stRec{\stRecVar}{\stIntSum{\roleQ}{}{\stLab[1]\stSeq\stRecVar}}
					}\\
					\stLab[2]\stSeq
					\stExtSum{\roleQ}{}{
						\stLab\stSeq
						\stIntSum{\roleQ}{}{\stLab[2]\stSeq
						\stRec{\stRecVar}{\stIntSum{\roleQ}{}{\stLab[1]\stSeq\stRecVar}}}
					}
				\end{array}
			}
		}
	}
\)\end{adjustbox}
\end{equation}
Projecting onto $\roleQ$ is more delicate, because $\roleQ$ is not involved in the initial $\roleP\to\roleR$ choice and its continuations across the two branches differ. We isolate this case in the following example, in order to explain the merge step that reconciles the two continuations.
\end{example}

\begin{example}[Merging in Projection]
\label{ex:merge-proj}
Continuing from Example~\ref{ex:coind-proj}, we project $\gtG$ onto $\roleQ$. Since $\roleQ$ is not a participant in the top-level $\roleP\to\roleR$ communication, the applicable rule is \RULE{P-$\stBinMerge$}: the projection is obtained by merging the projections of the two branches onto $\roleQ$. In each branch, $\roleQ$ is the sender of the en-route message from $\roleQ$ to $\roleP$ carrying label $\stLab$, so \RULE{P-$\oplus$-II} applies and yields a projection of the form $(\quMsg{\roleP}{\stLab}{},\, \stT[i])$ with common queue-entry prefix $\quMsg{\roleP}{\stLab}{}$ and residuals
\begin{equation*}
\stT[1] \,=\, \stRec{\stRecVar}{\stExtSum{\roleP}{}{\stLab[1]\stSeq\stRecVar}}
\qquad
\stT[2] \,=\, \stExtSum{\roleP}{}{\stLab[2]\stSeq \stRec{\stRecVar}{\stExtSum{\roleP}{}{\stLab[1]\stSeq\stRecVar}}}
\end{equation*}
from the $\stLab[1]$- and $\stLab[2]$-branches respectively.
The derivation of the projection onto $\roleQ$ is:
\begin{equation}
\begin{adjustbox}{width=\columnwidth,center}\(
	\cinference[P-$\stBinMerge$]
	{
		\cinference[P-$\oplus$-II]
		{
			\cinference[P-$\&$]
			{
				\gtProjRel{
					\gtRec{\gtRecVar}{
					\gtCommSingle{\roleP}{\roleQ}{\gtLab[1]}{}{
						\gtRecVar
					}
				}
				}{\roleQ}
				{\stRec{\stRecVar}{\stExtSum{\roleP}{}{\stLab[1]\stSeq\stRecVar}}}
			}
			{
				\gtProjRel{
					\gtRec{\gtRecVar}{
					\gtCommSingle{\roleP}{\roleQ}{\gtLab[1]}{}{
						\gtRecVar
					}
				}
				}{\roleQ}
				{\stRec{\stRecVar}{\stExtSum{\roleP}{}{\stLab[1]\stSeq\stRecVar}}}
			}
		}
		{
			\gtProjRel
			{
			\gtCommSquigSingle{\roleQ}{\roleP}{\gtLab}{\gtLab}{}{
				\gtRec{\gtRecVar}{
					\gtCommSingle{\roleP}{\roleQ}{\gtLab[1]}{}{
						\gtRecVar
					}
				}
			}
			}
			{\roleQ}
			{(\quMsg{\roleP}{\stLab}{},\stRec{\stRecVar}{\stExtSum{\roleP}{}{\stLab[1]\stSeq\stRecVar}})}
		}
		&
		\cinference[P-$\oplus$-II]
		{
			\cinference[P-$\&$]
			{
				\cinference[P-$\&$]
				{
					\gtProjRel{
					\gtRec{\gtRecVar}{
						\gtCommSingle{\roleP}{\roleQ}{\gtLab[1]}{}{
							\gtRecVar
						}
					}	
					}{\roleQ}
					{ \stRec{\stRecVar}{\stExtSum{\roleP}{}{\stLab[1]\stSeq\stRecVar}}
					}
				}
				{
					\gtProjRel{
					\gtRec{\gtRecVar}{
						\gtCommSingle{\roleP}{\roleQ}{\gtLab[1]}{}{
							\gtRecVar
						}
					}	
					}{\roleQ}
					{ \stRec{\stRecVar}{\stExtSum{\roleP}{}{\stLab[1]\stSeq\stRecVar}}
					}
				}
			}
			{
				\gtProjRel
			{
				\gtCommSingle{\roleP}{\roleQ}{\gtLab[2]}{}{
					\gtRec{\gtRecVar}{
						\gtCommSingle{\roleP}{\roleQ}{\gtLab[1]}{}{
							\gtRecVar
						}
					}
				}
			}
			{\roleQ}
			{
			\stExtSum{\roleP}{}{
				\stLab[2]\stSeq
				\stRec{\stRecVar}{\stExtSum{\roleP}{}{\stLab[1]\stSeq\stRecVar}}
			}}
			}
		}
		{
			\gtProjRel
			{
				\gtCommSquigSingle{\roleQ}{\roleP}{\gtLab}{\gtLab}{}{
				\gtCommSingle{\roleP}{\roleQ}{\gtLab[2]}{}{
					\gtRec{\gtRecVar}{
						\gtCommSingle{\roleP}{\roleQ}{\gtLab[1]}{}{
							\gtRecVar
						}
					}
				}
			}
			}
			{\roleQ}
			{(\quMsg{\roleP}{\stLab}{},
			\stExtSum{\roleP}{}{
				\stLab[2]\stSeq
				\stRec{\stRecVar}{\stExtSum{\roleP}{}{\stLab[1]\stSeq\stRecVar}}
			})}
		}
	}
	{
		\gtProjRel{\gtG}{\roleQ}{
			(\quMsg{\roleP}{\stLab}{},
			\stExtSum{\roleP}{}{
				\begin{array}{l}
					\stLab[1]\stSeq\stRec{\stRecVar}{\stExtSum{\roleP}{}{\stLab[1]\stSeq\stRecVar}}\\
					\stLab[2]\stSeq\stRec{\stRecVar}{\stExtSum{\roleP}{}{\stLab[1]\stSeq\stRecVar}}
				\end{array}
			})
		}
	}
\)\end{adjustbox}
\end{equation}
whose top \RULE{P-$\stBinMerge$} step relies on the merge
\begin{equation}
\stMerge{}{\{\stRec{\stRecVar}{\stExtSum{\roleP}{}{\stLab[1]\stSeq\stRecVar}},
\stExtSum{\roleP}{}{
				\stLab[2]\stSeq
				\stRec{\stRecVar}{\stExtSum{\roleP}{}{\stLab[1]\stSeq\stRecVar}}
			}\}}\ni\stExtSum{\roleP}{}{
				\begin{array}{l}
					\stLab[1]\stSeq\stRec{\stRecVar}{\stExtSum{\roleP}{}{\stLab[1]\stSeq\stRecVar}}\\
					\stLab[2]\stSeq\stRec{\stRecVar}{\stExtSum{\roleP}{}{\stLab[1]\stSeq\stRecVar}}
				\end{array}
			}
\end{equation}
which is derivable by a single application of \RULE{$\stBinMerge$-$\&$}, with one sub-obligation per label of the resulting external choice:
\begin{equation}
	\cinference[$\stBinMerge$-$\&$]{
		\stMergeRel{\{\stRec{\stRecVar}{\stExtSum{\roleP}{}{\stLab[1]\stSeq\stRecVar}}\}}{\stRec{\stRecVar}{\stExtSum{\roleP}{}{\stLab[1]\stSeq\stRecVar}}}
		&
		\stMergeRel{\{\stRec{\stRecVar}{\stExtSum{\roleP}{}{\stLab[1]\stSeq\stRecVar}}\}}{\stRec{\stRecVar}{\stExtSum{\roleP}{}{\stLab[1]\stSeq\stRecVar}}}
	}
	{
		\stMergeRel{\{\stRec{\stRecVar}{\stExtSum{\roleP}{}{\stLab[1]\stSeq\stRecVar}},
\stExtSum{\roleP}{}{
				\stLab[2]\stSeq
				\stRec{\stRecVar}{\stExtSum{\roleP}{}{\stLab[1]\stSeq\stRecVar}}
			}\}}{\stExtSum{\roleP}{}{
				\begin{array}{l}
					\stLab[1]\stSeq\stRec{\stRecVar}{\stExtSum{\roleP}{}{\stLab[1]\stSeq\stRecVar}}\\
					\stLab[2]\stSeq\stRec{\stRecVar}{\stExtSum{\roleP}{}{\stLab[1]\stSeq\stRecVar}}
				\end{array}
			}}
	}
\end{equation}
Intuitively, the merged type offers $\roleQ$ both labels $\stLab[1]$ and $\stLab[2]$ as an external choice from $\roleP$: under $\stLab[1]$ it continues as $\stT[1]$ does after unfolding, immediately re-entering the loop; under $\stLab[2]$ it continues as $\stT[2]$ does after its $\stLab[2]$-receive, which in this example also lands in the same loop $\stRec{\stRecVar}{\stExtSum{\roleP}{}{\stLab[1]\stSeq\stRecVar}}$. Whichever label $\roleP$ eventually sends to $\roleQ$, the merged local type at $\roleQ$ is ready for it, witnessing the compatibility of $\stT[1]$ and $\stT[2]$.
\end{example}

\subsection{Asynchronous Multiparty Subtyping}

Given a standard subtyping $\tyGroundSub$ for basic types
(\eg including $\tyInt\tyGroundSub\tyInt$ and $\tyInt \tyGroundSub \tyReal$),
we give a summary of \emph{asynchronous subtyping} $\asubt$ introduced in~\cite{GPPSY2023}. We first consider the tree representation 
of local type $\stT$ (denoted by $\ttree{\stT}$). 

We write $\trT$ for generic trees and additionally define three specific types of tree. 
\emph{Single-input} trees (denoted by $\stV$) are those which have only a singleton choice in all branchings, 
while \emph{single-output} trees (denoted by $\stU$) are those which have only a singleton choice in all selections. 
Trees which are both single-input and single-output are called \emph{single-input-single-output (SISO) trees} (denoted by $\stW$). 
These can all be defined coinductively by the following equations. 
\begin{align}
    \trT
    &=~
    \stExtSum{\roleP}{i \in I}{\stChoice{\stLab[i]}{\tyGround[i]} \stSeq \trT[i]} ~\mid~ \stIntSum{\roleP}{i \in I}{\stChoice{\stLab[i]}{\tyGround[i]} \stSeq \trT[i]}  ~\mid~ \stEnd
    \\[1mm]\stU
    &=~
    \stExtSum{\roleP}{i \in I}{\stChoice{\stLab[i]}{\tyGround[i]} \stSeq \stU[i]} ~\mid~ \stSend{\roleP}{\stLab}{\tyGround}{\stU}  ~\mid~ \stEnd
    \\[1mm]\stV
    &=~
    \stRecv{\roleP}{\stLab}{\tyGround}{\stV} ~\mid~ \stIntSum{\roleP}{i \in I}{\stChoice{\stLab[i]}{\tyGround[i]} \stSeq \stV[i]}  ~\mid~ \stEnd  
    \\[1mm]\stW
    &=~
    \stRecv{\roleP}{\stLab}{\tyGround}{\stW} ~\mid~ \stSend{\roleP}{\stLab}{\tyGround}{\stW}  ~\mid~ \stEnd
\end{align}
We will define reorderings of SISO trees, and to do so, we consider non-empty sequences $\Ap$ of receives not including $\roleP$ 
and $\Bp$ of sends not including $\roleP$ together with receives from any participant. 
These sequences are inductively defined (where $\roleP \ne \roleQ$) by:
\begin{equation}
  \begin{array}{ll}
    \Ap
    &=~
    \stRecvOne{\roleQ}{\stLab}{\tyGround} ~\mid~ \stRecv{\roleQ}{\stLab}{\tyGround}{\Ap}\quad 
\Bp
~=~
\stRecvOne{\roleR}{\stLab}{\tyGround} ~\mid~ \stSendOne{\roleQ}{\stLab}{\tyGround}  ~\mid~ \stRecv{\roleR}{\stLab}{\tyGround}{\Bp} ~\mid~ \stSend{\roleQ}{\stLab}{\tyGround}{\Bp}
  \end{array}
\end{equation}
\noindent We define the set $\act{\stW}$ of actions of a SISO tree: 
$\act{\stEnd} = \emptyset$; 
$\act{\stRecv{\roleP}{\stLab}{\tyGround}{\stW}} = \{\roleP?\} ~\cup~ \act{\stW}$; and 
$\act{\stSend{\roleP}{\stLab}{\tyGround}{\stW}} = \{\roleP!\} ~\cup~ \act{\stW}$. 
Using these definitions, we introduce a \emph{refinement relation} ($\subttt$) defined coinductively by the following rules:
\\[1mm]
\centerline{\(
\begin{array}{c}

\cinference[Ref-$\mathcal{A}$]{
  \tyGroundi \tyGroundSub \tyGround
  &
  \stW \subttt \Ap ; \stW'
  & 
  \act{\stW} = \act{\Ap ; \stW'}
}{
  \stRecv{\roleP}{\stLab}{\tyGround}{\stW} \subttt \Ap ; \stRecv{\roleP}{\stLab}{\tyGroundi}{\stWi}
}

\\[2ex]

\cinference[Ref-$\mathcal{B}$]{
  \tyGround \tyGroundSub \tyGroundi
  &
  \stW \subttt \Bp ; \stW'
  & 
  \act{\stW} = \act{\Bp ; \stW'}
}{
  \stSend{\roleP}{\stLab}{\tyGround}{\stW} \subttt \Bp ; \stSend{\roleP}{\stLab}{\tyGroundi}{\stWi}
}

\\[2ex]

\cinference[Ref-In]{
  \tyGroundi \tyGroundSub \tyGround  & \stW \subttt \stWi
}{
  \stRecv{\roleP}{\stLab}{\tyGround}{\stW} \subttt \stRecv{\roleP}{\stLab}{\tyGround'}{\stW'}
}

\qquad 

\cinference[Ref-Out]{
  \tyGround \tyGroundSub \tyGroundi  & \stW \subttt \stWi
}{
  \stSend{\roleP}{\stLab}{\tyGround}{\stW} \subttt \stSend{\roleP}{\stLab}{\tyGround'}{\stW'}
}

\qquad 

\cinference[Ref-End]{}{
  \stEnd \subttt \stEnd
}
\end{array}
\)}

\smallskip 

We can extract sets of single-input and single-output trees from a given tree using the functions $\singleIn{\cdot}$ and $\singleOut{\cdot}$.
\\[1mm]
\centerline{\(
\begin{array}{c}
\singleIn{\stExtSum{\roleP}{i \in I}{\stChoice{\stLab[i]}{\tyGround[i]} \stSeq \trT[i]}} = \bigcup_{i \in I}\{ {\stRecv{\roleP}{\stLab[i]}{\tyGround[i]}{\stV[i]}} \mid \stV[i] \in \singleIn{\trT[i]} \} 
\\[2ex]
\singleIn{\stIntSum{\roleP}{i \in I}{\stChoice{\stLab[i]}{\tyGround[i]} \stSeq \trT[i]}} = \{ {\stIntSum{\roleP}{i \in I}{\stChoice{\stLab[i]}{\tyGround[i]} \stSeq \stV[i]}} \mid \forall i \in I : \stV[i] \in \singleIn{\trT[i]} \} 
\quad
\singleIn{\stEnd} = \{ \stEnd \} 
\\[2ex]
\singleOut{\stIntSum{\roleP}{i \in I}{\stChoice{\stLab[i]}{\tyGround[i]} \stSeq \trT[i]}} = \bigcup_{i \in I}\{ {\stSend{\roleP}{\stLab[i]}{\tyGround[i]}{\stU[i]}} \mid \stU[i] \in \singleOut{\trT[i]} \} 
\\[2ex]
\singleOut{\stExtSum{\roleP}{i \in I}{\stChoice{\stLab[i]}{\tyGround[i]} \stSeq \trT[i]}} = \{ {\stExtSum{\roleP}{i \in I}{\stChoice{\stLab[i]}{\tyGround[i]} \stSeq \stU[i]}} \mid \forall i \in I : \stU[i] \in \singleOut{\trT[i]} \} 
\quad
\singleOut{\stEnd} = \{ \stEnd \} 
\\[2ex]
\end{array}
\)}
\begin{definition}[Subtyping]
\label{def:subtyping}
We consider trees that have only singleton choices in branchings (called \emph{single-input (SI) trees}), or in selections (\emph{single-output (SO) trees}), and we define the session subtyping $\asubt$ over all session types by considering their decomposition into SI, SO, and SISO trees.\\[1mm]
\begin{equation}
\inference[Sub]{
    \forall \stU \in \singleOut{\ttree{\stT}} & \forall \stVi \in \singleIn{\ttree{\stTi}} & \exists \stW \in \singleIn{\stU} & \exists \stWi \in \singleOut{\stVi} & \stW \subttt \stWi
}
{
    \stT \asubt \stTi
}
\end{equation}
The refinement $\subttt$ captures \emph{safe permutations} of input/output messages,
that never cause deadlocks or communication errors under asynchrony;
and the subtyping relation $\asubt$ focuses on
reconciling refinement $\subttt$ with the branching
structures in session types. 
\end{definition}

\begin{example}[Subtyping the Ring Protocol Projection~(\cref{sec:ring})]
Recall the local types:
\begin{equation}
\stT[\roleQ]=\stRec{\stRecVar}{
    \stExtSum{\roleP}{}{
      \stChoice{add}{\tyInt} \stSeq
      \stIntSum{\roleR}{}{
        \stChoice{add}{\tyInt} \stSeq \stRecVar,
        \stChoice{sub}{\tyInt} \stSeq \stRecVar
      }
    }
  }
\end{equation}
\begin{equation}
\stTopt[\roleQ] = \stRec{\stRecVar}{
    \stIntSum{\roleR}{}{
      \stChoice{add}{\tyInt} \stSeq \stExtSum{\roleP}{}{\stChoice{add}{\tyInt} \stSeq \stRecVar},
      \stChoice{sub}{\tyInt} \stSeq \stExtSum{\roleP}{}{\stChoice{add}{\tyInt} \stSeq \stRecVar}
    }
  }
\end{equation}
  To demonstrate that $\stTopt[\roleQ] \asubt \stT[\roleQ]$, we must show that for all $\stU \in \singleOut{\ttree{\stTopt[\roleQ]}}$ and $\stVi \in \singleIn{\ttree{\stT[\roleQ]}}$, there exist $\stW \in \singleIn{\stU}$ and $\stWi \in \singleOut{\stVi}$ such that $\stW \subttt \stWi$.
  Consider the following sets:
  \begin{align}
  \singleOut{\ttree{\stTopt[\roleQ]}} 
    &= 
    \left\{
      \stSend{\roleR}{\stFmt{\mathsf{add}}}{\tyInt}{
        \stRecv{\roleP}{\stFmt{\mathsf{add}}}{\tyInt}{\ldots}
      }, 
      \stSend{\roleR}{\stFmt{\mathsf{sub}}}{\tyInt}{
        \stRecv{\roleP}{\stFmt{\mathsf{add}}}{\tyInt}{\ldots}
      }, \ldots
    \right\} 
    \\[1ex]
  \singleIn{\ttree{\stT[\roleQ]}}  
    &= 
    \left\{
      \stRecv{\roleP}{\stFmt{\mathsf{add}}}{\tyInt}{
        \stIntSumRaw{\roleR}{\begin{array}{l}
          \stChoice{add}{\tyInt} \stSeq \ldots \\
          \stChoice{sub}{\tyInt} \stSeq \ldots
        \end{array}}
      }
    \right\}
  \end{align}
  Now, we must find for each $\stU$ in the first set and $\stVi$ in the second, a pair of SISO trees $(\stW, \stWi)$ such that $\stW \subttt \stWi$. 
For instance, if the second $\stU$ is chosen, we have
$\stW = \stSend{\roleR}{\stFmt{\mathsf{sub}}}{\tyInt}{
              \stRecv{\roleP}{\stFmt{\mathsf{add}}}{\tyInt}{\ldots}
}$ and we can pick
$\stWi = \stRecv{\roleP}{\stFmt{\mathsf{add}}}{\tyInt}{
               \stSend{\roleR}{\stFmt{\mathsf{sub}}}{\tyInt}{\ldots}
             }$

  \noindent
  Then we can apply rule \RULE{Ref-$\mathcal{B}$} to validate that it is safe to reorder the send ahead of the receive in the optimised type. 
We could form a similar argument in the other cases.
Thus  we conclude that:
  \(
    \stTopt[\roleQ] \asubt \stT[\roleQ]
  \).
  \end{example}

\myparagraph{Compactness of Subtypes.}
Beyond enabling message reordering optimisations, 
subtyping, in both synchronous and asynchronous settings, 
can also potentially yield \emph{arbitrarily more compact} type representations.
This compactness is particularly valuable in practical implementations,
where small types can also simplify readability and understanding.
\\
For example, consider a type that reads exactly $n$ messages labelled $\mathsf{a}$ followed by one $\mathsf{b}$:
\begin{equation}
\stT[n] = \underbrace{\stExtSum{\roleP}{}{\mathsf{a} \stSeq 
\stExtSum{\roleP}{}{\mathsf{a} \stSeq 
\cdots
\stExtSum{\roleP}{}{\mathsf{a} \stSeq 
\stExtSum{\roleP}{}{\mathsf{b} \stSeq \stEnd}
}}}}_{\text{$n$ nested external choices}}
\end{equation}
This type has size $O(n)$. However, a subtype that accepts \emph{any} number of $\mathsf{a}$'s before $\mathsf{b}$ can be written compactly as:
\begin{equation}
\stTi = \stRec{\stRecVar}{\stExtSum{\roleP}{}{\mathsf{a} \stSeq \stRecVar,\ \mathsf{b} \stSeq \stEnd}}
\end{equation}
We have $\stTi \asubt \stT[n]$ for all $n$, yet $\stTi$ has constant size $O(1)$.
This is in fact the case even for real projected local types, as coinductive projection 
can in general yield local types superpolynomial in the size of the global type~\cite{thien-nobuko-popl-25}.

\myparagraph{Properties of Subtyping.}
While asynchronous subtyping is very flexible,
allowing us to delay communications and 
simplify structure,
this can not be done indefinitely.
For example, asynchronous subtyping never
admits a non-trivial supertype of $\stEnd$.
  \begin{lemma}[End Subtyping]
    \label{lem:end-sub-end}
    If $\unfoldOne{\stTi} = \stEnd$ and $\stTi \asubt \stT$, then $\unfoldOne{\stT} = \stEnd$.
  \end{lemma}  
  \begin{proof}
    The type $\stEnd$ is conveniently already a SISO tree. If $\stTi \asubt \stT$ and $\unfoldOne{\stTi} = \stEnd$, by inversion we see that $\stEnd \subtt \stT$ can only be derived from \RULE{Ref-End}, 
    and so it must be that $\unfoldOne{\stT} = \stEnd$.
  \end{proof}
We also note that,
as one might expect from usual properties
of subtyping,
the precise asynchronous subtyping
is a preorder.
\begin{lemma}[Reflexivity and Transitivity of Subtyping, Lemma 3.8 in Ghilezan et al.~\cite{GPPSY2023}]
    For any closed, well-guarded local types $\stT$, $\stTi$ and $\stTii$: (1) $\stT \asubt \stT$ holds, and (2) if $\stT \asubt \stTi$ and $\stTi \asubt \stTii$ then it must be that $\stT \asubt \stTii$ holds.
\end{lemma}
However, unlike the synchronous subtyping~(\cref{def:subtyping}),
the precise asynchronous subtyping
does not give us a strong operational
correspondence between instances.
This is elaborated further in~\cref{ex:sub-not-corr}.
Regardless, we retain a weak operational correspondence,
witnessed by~\cref{lem:inv-subtyping}'s structural inversion of subtyping.

\subsection{Queue Subtyping}
As we have the ground subtyping $\tyGroundSub$
and ground types in our queues,
we can define a subtyping for queues.
This is important as
when we subtype local types we
may use $\tyGroundSub$
to alter the payload types;
we will need to subtype projected
queues to identify an
operational correspondence after
local actions.
The queue subtyping is defined to be covariant
with $\tyGround$ \ie
we define the queue subtyping
inductively by:\\[1mm]
\newline
\smallskip
\centerline{\(
	\inference[]{}{\quEmpty\asubt\quEmpty}
	\qquad
	\inference[]{\quH[L]\asubt\quHi[L]&\quH[R]\asubt\quHi[R]}
	{\quCons{\quH[L]}{\quH[R]}\asubt\quCons{\quHi[L]}{\quHi[R]}}
	\qquad
	\inference[]{\tyGround\tyGroundSub\tyGroundi}{\quMsg{\roleP}{\stLab}{\tyGround}\asubt\quMsg{\roleP}{\stLab}{\tyGroundi}}
\)}

The projection of a global type onto a participant produces a pairs containing
a queue and a local type.
Thus in order to use local type subtyping after projecting as in \cref{eq:intro-assoc} (\cref{sec:intro}),
we also require this notion of queue subtyping.
\section{Operational Semantics of Global and Local Types}
\label{sec:semantics}
In this section, we will introduce semantics for global types
and typing contexts.
While the typing context semantics are standard~\cite{POPL19LessIsMore},
our global type semantics~(\cref{def:gtype:lts-gt}) is new
and has been designed to reflect
the typing context semantics up-to
the precise asynchronous subtyping.

\subsection{Semantics of Global Types}
\label{sec:gtype:lts-gt}
We now present the Labelled Transition System (LTS) semantics for global types.
To begin, we introduce the transition labels in~\Cref{def:mpst-env-reduction-label},
which are also used in the LTS semantics of typing contexts~(discussed later in~\Cref{sec:gtype:lts-context}).

\begin{definition}[Transition Labels]
  \label{def:mpst-env-reduction-label}\label{def:mpst-label-subject}Let $\stEnvAnnotGenericSym$ be a transition label of the form:
\begin{equation}
  \stEnvAnnotGenericSym \ \bnfdef \
  \ltsBra{\roleP}{\roleQ}{\gtLab(\tyGround)} 
    \ \bnfsep \ \ltsSel{\roleP}{\roleQ}{\gtLab(\tyGround)} \quad
    \text{(receive or send a message)}
\end{equation}
\noindent
The subject of a transition label, written
\,$\ltsSubject{\stEnvAnnotGenericSym}$,\, is defined as follows:
\begin{equation}
    \ltsSubject{\ltsBra{\roleP}{\roleQ}{\gtLab(\tyGround)}} =
    \ltsSubject{\ltsSel{\roleP}{\roleQ}{\gtLab(\tyGround)}} =
    \roleP
\end{equation}
If $\tyGround\tyGroundSub\tyGroundi$ then we say that
$\ltsBra{\roleP}{\roleQ}{\gtLab(\tyGroundi)}\preccurlyeq\ltsBra{\roleP}{\roleQ}{\gtLab(\tyGround)}$
and
$\ltsSel{\roleP}{\roleQ}{\gtLab(\tyGround)}\preccurlyeq\ltsSel{\roleP}{\roleQ}{\gtLab(\tyGroundi)}$
\end{definition}
\noindent
The label 
$\ltsSel{\roleP}{\roleQ}{\gtLab(\tyGround)}$
denotes that $\roleP$ has enqueued a message
with label $\gtLab$ and payload of type $\tyGround$
onto its queue for $\roleQ$
and the label
$\ltsBra{\roleP}{\roleQ}{\gtLab(\tyGround)}$
denotes that $\roleP$ has dequeued/received a message
with label $\gtLab$ and payload of type $\tyGround$
from $\roleQ$'s queue.

\begin{definition}[Global Type Transitions]
  \label{def:gtype:lts-gt}
  The global type transition $\gtMove[\stEnvAnnotGenericSym]$ is inductively
  defined by the following rules. 

\centerline{\(
\begin{array}{c}
  \inference[\iruleGtMoveRec]{
     \gtG{}[\gtRec{\gtRecVar}{\gtG}/\gtRecVar] 
     \,\gtMove[\stEnvAnnotGenericSym]\, 
     \gtGi}
   {
     \gtRec{\gtRecVar}{\gtG} 
     \,\gtMove[\stEnvAnnotGenericSym]\, 
    \gtGi
  }
  \qquad 
   \inference[\iruleGtMoveBra]{}{
     \gtCommSquigSmall{\roleP}{\roleQ}{\gtLab}{}{\gtLab}{\tyGround}{\gtGi}
    \,\gtMove[\ltsBra{\roleQ}{\roleP}{\gtLab(\tyGround)}]\, 
    \gtGi
   }\\[2ex]
   \inference[\iruleGtMoveSel]{
     j \in I
   }{
     \gtCommSmall{\roleP}{\roleQ}{i \in I}{\gtLab[i]}{\tyGround[i]}{\gtGi[i]}
     \,\gtMove[\ltsSel{\roleP}{\roleQ}{\gtLab[j](\tyGround[j])}]\, 
     \gtCommSquigSmall{\roleP}{\roleQ}{\gtLab[j]}{}{\gtLab[j]}{\tyGround[j]}{\gtGi[j]}
   }
   \\[2ex]

   \inference[\iruleGtMoveCtx]{
     \forall i \in I: 
     \gtGi[i] 
     \,\gtMove[\stEnvAnnotGenericSym]\, 
     \gtGii[i] 
     &
 	\ltsSubject{\stEnvAnnotGenericSym}\ne \roleP
   }{
       \gtCommSmall{\roleP}{\roleQ}{i \in I}{\gtLab[i]}{\tyGround[i]}{\gtGi[i]}
    \,\gtMove[\stEnvAnnotGenericSym]\, 
       \gtCommSmall{\roleP}{\roleQ}{i \in I}{\gtLab[i]}{\tyGround[i]}{\gtGii[i]}   
   }
 \\[2ex]
   \inference[\iruleGtMoveCtx']{
     \forall i \in J\subseteq I: 
     \gtGi[i] 
     \,\gtMove[\stEnvAnnotGenericSym]\, 
     \gtGii[i] 
     &
     \stEnvAnnotGenericSym = \ltsSel{\roleP}{\roleR}{\gtLabi(\tyGroundi)}
     &
     \text{with}&\roleR\ne\roleQ
   }{
       \gtCommSmall{\roleP}{\roleQ}{i \in I}{\gtLab[i]}{\tyGround[i]}{\gtGi[i]}
    \,\gtMove[\stEnvAnnotGenericSym]\, 
       \gtCommSmall{\roleP}{\roleQ}{i \in J}{\gtLab[i]}{\tyGround[i]}{\gtGii[i]}   
   }
 \\[2ex]
   \inference[\iruleGtMoveCtxII]{
     \gtGi
     \,\gtMove[\stEnvAnnotGenericSym]\, 
     \gtGii
     &
     \stEnvAnnotGenericSym \ne \ltsBra{\roleQ}{\roleP}{\gtLabi(\tyGroundi)}
   }{
       \gtCommSquigSmall{\roleP}{\roleQ}{\gtLab}{}{\gtLab}{\tyGround}{\gtGi}
    \,\gtMove[\stEnvAnnotGenericSym]\, 
       \gtCommSquigSmall{\roleP}{\roleQ}{\gtLab}{}{\gtLab}{\tyGround}{\gtGii}   
   }
 \end{array}
 \)}
\smallskip 
 
  We use 
  $\gtG \,\gtMove\, 
  \gtGi$
  if there
  exists $\stEnvAnnotGenericSym$ such that
  $\gtG
  \,\gtMove[\stEnvAnnotGenericSym]\, 
  \gtGi$; 
  we write
  $\gtG \,\gtMove$
  if there
  exists $\gtGi$ such that
  $\gtG
  \,\gtMove\,  
  \gtGi$, 
  and $\gtNotMove{\gtG}$ for its negation (\ie there is no $\gtGi$ 
  such that  $\gtG
  \,\gtMove\,  
  \gtGi$). 
  Finally, 
 $\gtMoveStar$ denotes the transitive and reflexive closure of
  $\gtMove$.
\end{definition}
\noindent
The semantics of global types reflect the behaviours permitted by asynchronous subtyping by allowing specific behaviours to be executed with a type. 
\begin{itemize}
  \item \RULE{GR-$\mu$} permits a valid transition to take place under a recursion binder.
  \item \RULE{GR-$\&$} describes the receiving of asynchronous messages, allowing en-route message to be received.
  \item \RULE{GR-$\oplus$} describes the sending of asynchronous messages, resulting in a standard transmission becoming an en-route one.
  \item \RULE{GR-Ctx-I}
 allows a transition $\stEnvAnnotGenericSym$ to be anticipated inside the continuations of a communication $\roleP\to\roleQ$, provided the subject of $\stEnvAnnotGenericSym$ is not $\roleP$. That is, actions by participants other than the sender can proceed independently of $\roleP$'s pending communication.
  \item \RULE{GR-Ctx-I'}
  handles the case where $\roleP$ itself sends to a \emph{different} recipient $\roleR\ne\roleQ$ before completing its communication with $\roleQ$. This mirrors the $\Bp$ reordering in the refinement relation~(\cref{def:subtyping}), which permits sends to be moved ahead of sends to other participants. The indexing set may shrink from $I$ to $J\subseteq I$, reflecting that asynchronous subtyping allows some branches to be forgotten when a send is reordered.
  \item \RULE{GR-Ctx-II} applies when a message from $\roleP$ to $\roleQ$ is already en route. Here any transition is permitted except $\roleQ$ receiving a different message from $\roleP$, which would violate the ordering of $\roleP$'s messages to $\roleQ$.
\end{itemize}

Together, \RULE{GR-Ctx-I}, \RULE{GR-Ctx-I'} and \RULE{GR-Ctx-II} enable the global type semantics to capture the same safe reorderings that are present in the precise asynchronous subtyping relation.

\begin{example}[Non-Determinism in Global Semantics]
\label{ex:non-det}
	\RULE{GR-Ctx-I'} will result in non-determinism \ie
	there is $\gtG$, $\stEnvAnnotGenericSym$, and $\gtGi\ne\gtGii$ with
	$\gtG\,\gtMove[\stEnvAnnotGenericSym]\,\gtGi$ and
	$\gtG\,\gtMove[\stEnvAnnotGenericSym]\,\gtGii$.
	
	Consider the following global types:
\begin{equation}
		\gtG[\text{\tiny non-det}] = \gtRec{\gtRecVar}{
			\gtCommRaw{\roleP}{\roleQ}{
				\begin{array}{@{}l@{}}
					\gtLab[1]\gtSeq
					\gtCommSingle{\roleQ}{\roleR}{\gtLab[1]}{}{\gtRecVar}
					\\
					\gtLab[2]\gtSeq
					\gtCommSingle{\roleQ}{\roleR}{\gtLab[2]}{}{
						\gtCommSingle{\roleP}{\roleR}{\gtLab}{}{}
					}
				\end{array}
			}
		}
\end{equation}
\begin{equation}
		\gtGi =
			\gtCommSingle{\roleP}{\roleQ}{\gtLab[2]}{}{
					\gtCommSingle{\roleQ}{\roleR}{\gtLab[2]}{}{
						\gtCommSquigSingle{\roleP}{\roleR}{\gtLab}{\gtLab}{}{}
					}
			}
	\qquad
		\gtGii = \gtRec{\gtRecVar}{
			\gtCommRaw{\roleP}{\roleQ}{
				\begin{array}{@{}l@{}}
					\gtLab[1]\gtSeq
					\gtCommSingle{\roleQ}{\roleR}{\gtLab[1]}{}{\gtGi}
					\\
					\gtLab[2]\gtSeq
					\gtCommSingle{\roleQ}{\roleR}{\gtLab[2]}{}{
						\gtCommSquigSingle{\roleP}{\roleR}{\gtLab}{\gtLab}{}{}
					}
				\end{array}
			}
		}
\end{equation}
	We can derive that $\gtG[\text{\tiny non-det}]\,\gtMove[\ltsSel{\roleP}{\roleR}{\gtLab}]\,\gtGi$
	and that $\gtG[\text{\tiny non-det}]\,\gtMove[\ltsSel{\roleP}{\roleR}{\gtLab}]\,\gtGii$.
	$\gtGi\neq\gtGii$, so the transition relation is non-deterministic.
	
	This is necessary for $\gtG[\text{\tiny non-det}]$ to capture all behaviours
	up to asynchronous subtyping.
	Indeed, suppose we take some deterministic subrelation of the transition relation
We will now be able to associate $\gtG$
	with a typing context,
	which it does not operationally correspond to.

	Consider its projections:
\begin{equation}
		\stT[\roleP] = \stRec{\stRecVar}{
			\stIntSumRaw{\roleQ}{
				\begin{array}{l}
					\stChoice{\stLab[1]}{}\stSeq\stRecVar\\
					\stChoice{\stLab[2]}{}\stSeq \roleR\stFmt{\oplus}\stChoice{\stLab}{}
				\end{array}
			}
		}
		\qquad
		\stT[\roleQ] = \stRec{\stRecVar}{
			\stExtSumRaw{\roleP}{
				\begin{array}{l}
					\stChoice{\stLab[1]}{}\stSeq \roleR\stFmt{\oplus}\stChoice{\stLab[1]}{}\stSeq\stRecVar\\
					\stChoice{\stLab[2]}{}\stSeq \roleR\stFmt{\oplus}\stChoice{\stLab[2]}{}
				\end{array}
			}
		}
		\qquad
		\stT[\roleR] = \stRec{\stRecVar}{
			\stExtSumRaw{\roleQ}{
				\begin{array}{l}
					\stChoice{\stLab[1]}{}\stSeq\stRecVar\\
					\stChoice{\stLab[2]}{}\stSeq \roleP\stFmt{\&}\stChoice{\stLab}{}
				\end{array}
			}
		}
\end{equation}	 
	We have that
\begin{equation}
	\roleR\stFmt{\oplus}\stChoice{\stLab}{}
	\underbrace{
	\stSeq \roleQ\stFmt{\oplus}\stChoice{\stLab[1]}{}
	\dots
	\stSeq \roleQ\stFmt{\oplus}\stChoice{\stLab[1]}{}
	}_{n+1}
	\stSeq \roleQ\stFmt{\oplus}\stChoice{\stLab[2]}{} \asubt\stT[\roleP]
\end{equation}
	If $\gtG[\text{\tiny non-det}]\,\gtMove[\ltsSel{\roleP}{\roleR}{\gtLab}]\,\gtGiii$, then
	$\gtGiii$ can perform $\ltsSel{\roleP}{\roleQ}{\gtLab[1]}$
	at most $n$ times, for some $n$.
	This is because the transition has a finite derivation,
	beyond which we must truncate the possibility for
	$\ltsSel{\roleP}{\roleQ}{\gtLab[1]}$ to occur.
	If no such transition exists, then operational correspondence
	fails for this transition relation.
	As we are assuming that the transition is deterministic,
	we must take this transition in the simulation of the associated context.
	
	$\gtGiii$ cannnot perform the action $\ltsSel{\roleP}{\roleQ}{\gtLab[1]}$
	$n+1$ times and so
	$\gtG[\text{\tiny non-det}]$ does not correspond to the projected context,
	up to precise asynchronous subtyping.
	
	Therefore, non-determinism is necessary
	for global types to correspond to projected contexts up to
	the precise asynchronous subtyping.
	
	There are optimisations of our transition that can limit
	the non-determinism;
	we could
	enforce that maximum $J\subset I$ are taken
	when they can be shown to exist,
	but this is not the case in general.
	For example,
	$\gtG[\text{\tiny non-det}]$ has no maximum transition
	for $\gtG[\text{\tiny non-det}]\,\gtMove[\ltsSel{\roleP}{\roleR}{\stLab[1]}]$,
	as the redex can be taken to have arbitrary depth.
\end{example}

\begin{example}[Semantics of Global Type for Ring Protocol]
  Consider the global type for the ring-choice protocol (\Cref{sec:ring}). 
  The asynchronous semantics enable us to apply both $\RULE{GR-\ensuremath{\oplus}}$ transitions, 
  corresponding to sends from $\roleP$ to $\roleQ$ and from $\roleQ$ to $\roleR$, 
  before any receive transitions (using $\RULE{GR-\ensuremath{\&}}$) are applied. 
  As we will see later, this particular choice of global transition path corresponds to behaviour which 
  can only be captured by the optimised local type. 
  \\
  We begin by reducing $\gtG[\text{ring}]$ via a send action from $\roleP$ to $\roleQ$
using \RULE{\iruleGtMoveSel}:
\begin{align}
  \gtG[\text{ring}] & \gtMove[\ltsSel{\roleP}{\roleQ}{\gtMsgFmt{add}}] \ \gtGone[\text{ring}] = \
    \gtCommSquigSingle{\roleP}{\roleQ}{\gtMsgFmt{add}}{\mathsf{add}}{\tyInt}
      {
        \gtCommRaw{\roleQ}{\roleR}{
          \begin{array}{@{}l@{}}
          \mathsf{add}(\tyInt)\gtSeq
          \gtCommRaw{\roleR}{\roleP}{
            \mathsf{add}(\tyInt)\gtSeq  \gtG_\text{ring}
          }
          \\
          \mathsf{sub}(\tyInt)\gtSeq
          \gtCommRaw{\roleR}{\roleP}{
            \mathsf{sub}(\tyInt)\gtSeq  \gtG_\text{ring}
          }
          \end{array}
        }
    }
\end{align}
  At this point, a message from $\roleP$ to $\roleQ$ is in transit. 
  We then perform another $\RULE{GR-\ensuremath{\oplus}}$ transition, using 
  $\RULE{\iruleGtMoveCtxII}$ to apply the sending from $\roleQ$ to $\roleR$ 
  under the existing en-route type:
\begin{align}
  \gtGone[\text{ring}] & \gtMove[\ltsSel{\roleQ}{\roleR}{\gtMsgFmt{sub}}] \ \gtGtwo[\text{ring}] = 
    \gtCommSquigSingle{\roleP}{\roleQ}{\gtMsgFmt{add}}{\mathsf{add}}{\tyInt}
      {
        \gtCommSquigSingle{\roleQ}{\roleR}{\gtMsgFmt{sub}}{
          \mathsf{sub}}{\tyInt}{
          \gtCommSingle{\roleR}{\roleP}{
            \mathsf{sub}}{\tyInt}{\gtG_\text{ring}}
          }
        }
\end{align}
  The state $\gtGtwo[\text{ring}]$ reflects the two en-route messages: one from $\roleP$ to $\roleQ$ and one from $\roleQ$ to $\roleR$.
  We can then proceed with the corresponding receive actions using the $\RULE{GR-\&}$ rule. First, $\roleQ$ receives the message from $\roleP$:
\begin{align}
  \gtGtwo[\text{ring}] & \gtMove[\ltsBra{\roleQ}{\roleP}{\gtMsgFmt{add}}] \ \gtGthree[\text{ring}] =  \
      \gtCommSquigSingle{\roleQ}{\roleR}{\gtMsgFmt{sub}}{
        \mathsf{sub}}{\tyInt}{
        \gtCommSingle{\roleR}{\roleP}{
          \mathsf{sub}}{\tyInt}\gtG_\text{ring}
        }
\end{align}
Then, $\roleR$ receives the message from $\roleQ$ by \RULE{GR-\&}:
\begin{align}
  \gtGthree[\text{ring}] & \gtMove[\ltsBra{\roleR}{\roleQ}{\gtMsgFmt{sub}}] \ \gtGfour[\text{ring}] =  \
          \gtCommRaw{\roleR}{\roleP}{
            \mathsf{sub}(\tyInt)\gtSeq  \gtG_\text{ring}
          } 
\end{align}
Next, $\roleR$ sends to $\roleP$ by \RULE{GR-$\oplus$}:
\begin{align}
  \gtGfour[\text{ring}] & \gtMove[\ltsSel{\roleR}{\roleP}{\gtMsgFmt{sub}}] \ \gtGfive[\text{ring}] =  \
          \gtCommRawSquig{\roleR}{\roleP}{\gtMsgFmt{sub}}{
            \mathsf{sub}(\tyInt)\gtSeq  \gtG_\text{ring}
          } 
\end{align}
Finally, $\roleP$ receives this last message by \RULE{GR-\&}, returning us to the original state of the protocol:
\begin{align}
  \gtGfive[\text{ring}] & \gtMove[\ltsBra{\roleR}{\roleP}{\gtMsgFmt{sub}}] \ \gtG[\text{ring}]
\end{align}
In \Cref{sec:gtype:lts-context}, we will show that this transition sequence corresponds to a behaviour of the optimised
  local implementation for $\roleQ$.  
\end{example}

\subsection{Semantics of Typing Contexts}
\label{sec:gtype:lts-context}
After introducing the semantics of global types, 
we now present an LTS semantics for \emph{typing contexts}, which are 
 collections of local types. 
The formal definition of a typing context is provided in \Cref{def:mpst-env}, 
followed by its transition rules in \Cref{def:mpst-env-reduction}.
We then define safety, freedom, and liveness for contexts in~\cref{def:live}.

\begin{definition}[Typing Contexts]\label{def:mpst-env}\label{def:mpst-env-closed}\label{def:mpst-env-comp}\label{def:mpst-env-subtype}$\stEnv$ denotes a partial mapping
  from participants to queues and types. Their syntax is 
  defined as:
\begin{equation}
  \stEnv
  \,\coloncolonequals\,
  \stEnvEmpty
  \bnfsep
  \stEnv \stEnvComp \stEnvMap{\roleP}{(\quH, \stT)}
\end{equation}
  The \emph{context composition} $\stEnv[1] \stEnvComp \stEnv[2]$ 
  is defined iff $\dom{\stEnv[1]} \cap \dom{\stEnv[2]} = \emptyset$. \end{definition}

\begin{definition}[Typing Context Transition]\label{def:mpst-env-reduction}The \emph{typing context transition\; $\stEnvMoveGenAnnot$} \;is inductively defined by the following
  rules:

\smallskip
\centerline{\(\begin{array}{c}
    \inference[\iruleTCtxOut]{k \in I}{\stEnvMap{\roleP
      }{(\quH, \stIntSum{\roleQ}{i \in I}{\stChoice{\stLab[i]}{\tyGround[i]} \stSeq \stT[i]})}\,\stEnvMoveAnnot{\ltsSel{\roleP}{\roleQ}{\stChoice{\stLab[k]}{\tyGround[k]}}}\,\stEnvMap{\roleP }{
        (\quCons{\quH}{\quMsg{\roleQ}{\stLab[k]}{\tyGround[k]}}, \stT[k])}}\\[2mm]
    \inference[\iruleTCtxIn]{k \in I & \tyGroundi \tyGroundSub \tyGround[k]
    }{\stEnvMap{\roleP }{(\quH, \stExtSum{\roleQ}{i \in I}{\stChoice{\stLab[i]}{\tyGround[i]} \stSeq \stT[i]})}\!\stEnvComp
      \stEnvMap{\roleQ }{(\quCons{\quMsg{\roleP}{\stLab[k]}{\tyGroundi}}{\quHi}, \stT)}\,\stEnvMoveAnnot{\ltsBra{\roleP}{\roleQ}{\stChoice{\stLab[k]}{\tyGround[k]}}}\,\stEnvMap{\roleP }{(\quH, \stT[k])}  
      \!\stEnvComp
      \stEnvMap{\roleQ }{(\quHi, \stT)}  
    }\\[2mm]
    \inference[\iruleTCtxRec]{\stEnvMap{\roleP }{(\quH,\stT\subst{\stRecVar}{\stRec{\stRecVar}{\stT}})}\stEnvMoveGenAnnot \stEnvi }{\stEnvMap{\roleP }{(\quH,\stRec{\stRecVar}{\stT})}\stEnvMoveGenAnnot \stEnvi }\qquad
    \inference[\iruleTCtxCong]{\stEnv \stEnvMoveGenAnnot \stEnvi }{\stEnv \!\stEnvComp \stEnvMap{\mpC}{(\quH,\stT)}
      \stEnvMoveGenAnnot \stEnvi \!\stEnvComp \stEnvMap{\mpC}{(\quH,\stT)}
    }\end{array}
  \)}

\smallskip 

We write $\stEnv \!\stEnvMoveGenAnnot$ if there exists $\stEnvi$ 
such that $\stEnv \!\stEnvMoveGenAnnot\! \stEnvi$. We write $\stEnv \!\stEnvMove\! \stEnvi$ iff $\stEnv \!\stEnvMoveGenAnnot \stEnvi$ for some $\stEnvAnnotGenericSym$
    and $\stEnvNotMoveP{\stEnv}$ for its negation 
    (\ie  there is no $\stEnvi$ such that $\stEnv \!\stEnvMove\!  \stEnvi$),   and we denote $\stEnvMoveStar$ as the reflexive and transitive closure of $\stEnvMove$. 
\end{definition}

Rule \inferrule{\iruleTCtxOut} says that participant $\roleP$
with internal choices $I$ towards $\roleQ$
can make the choice $k\in I$
and enqueue the message $\stChoice{\stLab[k]}{\tyGround[k]}$
for $\roleQ$ into its queue, and continue with $\stT[k]$.
Dually, rule \inferrule{\iruleTCtxIn} says that participant $\roleP$
with external choices $I$ from $\roleQ$
can dequeue message $\stChoice{\stLab[k]}{\tyGround[k]}$
from $\roleP$'s queue if $k\in I$, and continue with $\stT[k]$.
Rule \inferrule{\iruleTCtxRec} unfolds recursive types
and rule \inferrule{\iruleTCtxCong}
handles enlarging typing contexts.

\begin{example}[Operational Semantics of Optimised Ring Context]
\label{ex:op-sem-opt}
As an example, consider the operational semantics of the optimised ring protocol. 
Each transition captures either a message send or receive, which either enqueues or dequeues a 
message in the queue of the sending participant. 
\\[2.5ex]
\begin{adjustbox}{width=\columnwidth,center}\(
\begin{array}{rrclc}
&
\stEnv[0] &=& \stEnvMap{\roleP }{
(\quEmpty[], \stT[\roleP])}, 
\stEnvMap{\roleQ }{
(\quEmpty[], \stTopt[\roleQ])}, 
\stEnvMap{\roleR }{
(\quEmpty[], \stT[\roleR])}
\\
\stEnvMoveAnnot{\ltsSel{\roleP}{\roleQ}{\stChoice{\stLabFmt{add}}{\tyInt}}} &
\stEnv[1] &=& 
\stEnvMap{\roleP }{
(\langle  \quMsg{\roleQ}{\stLabFmt{add}}{\tyInt} \rangle, 
    \stExtSumRaw{\roleR}{
      \begin{array}{l}
        \stChoice{add}{\tyInt} \stSeq \stT[\roleP] \\
        \stChoice{sub}{\tyInt} \stSeq \stT[\roleP]
      \end{array}
    }
)}, 
\stEnvMap{\roleQ }{
(\quEmpty[], \stTopt[\roleQ])}, 
\stEnvMap{\roleR }{
(\quEmpty[], \stT[\roleR])}&
\inferrule{\iruleTCtxOut,\iruleTCtxRec,\iruleTCtxCong} 
\\
\stEnvMoveAnnot{\ltsSel{\roleQ}{\roleR}{\stChoice{\stLabFmt{sub}}{\tyInt}}}  &
\stEnv[2]  &=&
\stEnvMap{\roleP }{
(\langle  \quMsg{\roleQ}{\stLabFmt{add}}{\tyInt} \rangle, 
    \stExtSumRaw{\roleR}{
      \begin{array}{l}
        \stChoice{add}{\tyInt} \stSeq \stT[\roleP] \\
        \stChoice{sub}{\tyInt} \stSeq \stT[\roleP]
      \end{array}
    }
)}, 
\stEnvMap{\roleQ }{
  ( \langle \quMsg{\roleR}{ \stLabFmt{sub}}{\tyInt} \rangle, 
  \stExtSum{\roleP}{}{\stChoice{add}{\tyInt} \stSeq \stTopt[\roleQ]}
  )}, 
\stEnvMap{\roleR }{
(\quEmpty, \stT[\roleR])}
&
\inferrule{\iruleTCtxOut,\iruleTCtxRec,\iruleTCtxCong}
\\
\stEnvMoveAnnot{\ltsBra{\roleQ}{\roleP}{\stChoice{\stLabFmt{add}}{\tyInt}}} &\stEnv[3]& =&
\stEnvMap{\roleP }{
(\quEmpty[], 
    \stExtSumRaw{\roleR}{
      \begin{array}{l}
        \stChoice{add}{\tyInt} \stSeq \stT[\roleP] \\
        \stChoice{sub}{\tyInt} \stSeq \stT[\roleP]
      \end{array}
    }
)}, 
\stEnvMap{\roleQ }{
  ( \langle \quMsg{\roleR}{ \stLabFmt{sub}}{\tyInt} \rangle
  , 
  \stTopt[\roleQ]
  )}, 
\stEnvMap{\roleR }{
(\quEmpty[], \stT[\roleR])}
&
\inferrule{\iruleTCtxIn,\iruleTCtxCong}
\\
\stEnvMoveAnnot{\ltsBra{\roleR}{\roleQ}{\stChoice{\stLabFmt{sub}}{\tyInt}}}  &\stEnv[4] &=&
\stEnvMap{\roleP }{
(\quEmpty[], 
    \stExtSumRaw{\roleR}{
      \begin{array}{l}
        \stChoice{add}{\tyInt} \stSeq \stT[\roleP] \\
        \stChoice{sub}{\tyInt} \stSeq \stT[\roleP]
      \end{array}
    }
)}, 
\stEnvMap{\roleQ }{
  ( \quEmpty[], 
  \stTopt[\roleQ]
  )}, 
\stEnvMap{\roleR }{
(\quEmpty[], 
\stIntSum{\roleP}{}{
         \stChoice{sub}{\tyInt} \stSeq \stT[\roleR]
       })}
&
\inferrule{\iruleTCtxIn,\iruleTCtxRec,\iruleTCtxCong}
\\
\stEnvMoveAnnot{\ltsSel{\roleR}{\roleP}{\stChoice{\stLabFmt{sub}}{\tyInt}}}&  \stEnv[5] &=& 
\stEnvMap{\roleP }{
(\quEmpty[], 
  \stExtSumRaw{\roleR}{
    \begin{array}{l}
      \stChoice{add}{\tyInt} \stSeq \stT[\roleP] \\
      \stChoice{sub}{\tyInt} \stSeq \stT[\roleP]
    \end{array}
  }
)}, 
\stEnvMap{\roleQ }{
  ( \quEmpty[], 
  \stTopt[\roleQ]
  )}, 
\stEnvMap{\roleR }{
(\langle \quMsg{\roleP}{\stLabFmt{sub}}{\tyInt} \rangle, \stT[\roleR])
}
&
\inferrule{\iruleTCtxOut,\iruleTCtxCong}
\\
\stEnvMoveAnnot{\ltsBra{\roleP}{\roleR}{\stChoice{\stLabFmt{sub}}{\tyInt}}}  &
\stEnv[0] &&
&
\inferrule{\iruleTCtxIn,\iruleTCtxCong}
\end{array}
\)
\end{adjustbox}
\end{example}
Not all contexts describe `correct' protocols.
The most obvious errors
are label mismatches;
a message on the queue may
possess a label for $\roleP$
that is excluded in the
possible branches of
$\roleP$'s type.
We call contexts,
where this can never occur,
\emph{safe}~(\cref{def:safe}).
A context could fail to be correct
by terminating
with non-empty queues
or non-$\stEnd$ types;
these are \emph{deadlocked} contexts~(\cref{def:df}).
In~\cref{sec:properties},
we shall see similar behaviours
for processes and
we can define corresponding
session correctness properties,
which will be implied by
the typing context properties.
Namely, session deadlock-freedom
and session safety~(\cref{def:safety-deadlock-free})
are implied by context deadlock-freedom
and context safety, respectively.
A more subtle failure is indefinitely delaying
participant behaviour, even under \emph{fair scheduling}.
Assuming that we reduce a context so that
whenever $\roleP$ can send a message or receive a message
on the queue it will do so,
we want to know that
every message on the queue will eventually
be dequeued and that
every local receiving type will eventually
dequeue a message.
This is liveness~(\cref{def:live})
and corresponds
with session liveness~(\cref{def:session-liveness}).
In this asynchronous setting,
context safety and deadlock-freedom
are implied by liveness,
so we will only need to prove liveness.
We establish in
\cref{sec:live}
that every typing context obtained via projection from a balanced global type is live (\cref{thm:assoc-live}), and hence safe and deadlock-free.
\begin{definition}[Context Safety]
\label{def:safe}
We say that a context $\stEnv$ is safe iff
for all $\stEnv\stEnvMoveStar\stEnvi$ if
$\stEnvApp{\stEnvi}{\roleP}=(\quCons{\quMsg{\roleQ}{\stLab}{\tyGround}}{\quH},\stT)$
and
$\stEnvApp{\stEnvi}{\roleQ}=(\quHi,\stExtSum{\roleP}{i \in I}{\stChoice{\stLab[i]}{\tyGround[i]} \stSeq \stT[i]})$
then
there is $j\in I$ such that
$\stLab=\stLab[j]$ and
$\tyGround\tyGroundSub\tyGround[j]$.
\end{definition}
\begin{definition}[Context Deadlock-freedom]
\label{def:df}
We say that a context $\stEnv$ is deadlock-free iff
for all $\stEnv\stEnvMoveStar\stEnvi$ if
$\stEnvNotMoveP{\stEnvi}$
then for all $\roleP\in\dom{\stEnvi}$
$\unfoldOne{\stEnvApp{\stEnvi}{\roleP}}=\stEnd$.
\end{definition}
\begin{definition}[Context Liveness~\cite{GPPSY2023}]
\label{def:live}
Let $I=\{0,1,2,\dots,n\}$ for some $n$ or $I=\{0,1,2,\dots\}$.
Let $\{\stEnv[i]\}_{i\in I}$ be a sequence of typing contexts
such that $\stEnv[i]\,\stEnvMove\,\stEnv[i+1]$ for $i,i+1\in I$.
We say that the path $\{\stEnv[i]\}_{i\in I}$ is fair iff
for all $i\in I$:
\begin{itemize}
	\item[F1] if $\stEnv[i]\stEnvMoveAnnot{\ltsSel{\roleP}{\roleQ}{\stLab(\tyGround)}}$,
	then there are $k$, $\stLabi$, and $\tyGroundi$ such that
	$I\ni k+1> i$ and $\stEnv[k]\stEnvMoveAnnot{\ltsSel{\roleP}{\roleQ}{\stLabi(\tyGroundi)}}\stEnv[k+1]$
	\item[F2] if $\stEnv[i]\stEnvMoveAnnot{\ltsBra{\roleP}{\roleQ}{\stLab(\tyGround)}}$,
	then there is $k$ such that
	$I\ni k+1> i$ and $\stEnv[k]\stEnvMoveAnnot{\ltsBra{\roleP}{\roleQ}{\stLab(\tyGround)}}\stEnv[k+1]$
\end{itemize}
We say that the path $\{\stEnv[i]\}_{i\in I}$ is live iff,
for all $i\in I$:
\begin{itemize}
	\item[L1] If $\stEnvApp{\stEnv[i]}{\roleP}
	=(\quCons{\quMsg{\roleQ}{\stLab}{\tyGround}}{\quH},\stT)$,
	then there exist $k$ and $\tyGroundi$ such that
	$I\ni k+1> i$ and
	$\stEnv[k]\stEnvMoveAnnot{\ltsBra{\roleQ}{\roleP}{\stLab(\tyGroundi)}}\stEnv[k+1]$
	\item[L2] If $\stEnvApp{\stEnv[i]}{\roleP}
	=(\quH,\stExtSum{\roleQ}{j \in J}{\stChoice{\stLab[j]}{\tyGround[j]} \stSeq \stT[j]})$,
	then there exist $k$, $\stLabi$, and $\tyGroundi$ such that
	$I\ni k+1 > i$ and
	$\stEnv[k]\stEnvMoveAnnot{\ltsBra{\roleP}{\roleQ}{\stLabi(\tyGroundi)}}\stEnv[k+1]$
\end{itemize}
We say that a context $\stEnv$ is live iff
all fair paths starting at $\stEnv$
are live.
\end{definition}
The context $\stEnv[0]$
from~\cref{ex:op-sem-opt} above is live,
which we will be able to show
by Theorem~\ref{thm:assoc-live}
later.
\\[2ex]
Finding non-live contexts
merely requires us to witness
failures of L1 or L2.
We now provide three non-live contexts:
one that is unsafe;
one that is safe but deadlocked;
and one that is safe and deadlock-free.
\begin{example}[Failures of Liveness]
\ \\
	\begin{enumerate}[leftmargin=0.8in,labelindent=-\leftmargin]
	\item[Unsafe]
	Consider the context
	$\stEnvMap{\roleP}{(\quMsg{\roleQ}{\stLab}{},\stEnd)},
	\stEnvMap{\roleQ}{(\quEmpty,\stExtSum{\roleP}{}{\stLabi})}$.
	This context is not live because the message
	$\quMsg{\roleQ}{\stLab}{}$
	does not match with
	$\stExtSum{\roleP}{}{\stLabi}$
	and so no communication can occur,
	violating L1 and L2.
	This is also an example of an unsafe context.
\item[Deadlocked]
	Consider the context
	$\stEnvMap{\roleP}{(\quEmpty, \stExtSum{\roleQ}{}{\stLab})}$.
	This context is not live because no
	communication can occur and hence
	violating L2.
	This is also an example of
	a safe but deadlocked context.
\item[Livelocked]
	Consider the context
	$\stEnvMap{\roleP}{(\quEmpty, \stRec{\stRecVar}{\stIntSum{\roleQ}{}{\stLab\stSeq\stRecVar}})},
	\stEnvMap{\roleQ}{(\quEmpty, \stRec{\stRecVar}{\stExtSum{\roleP}{}{\stLab\stSeq\stRecVar}})},
	\stEnvMap{\roleR}{(\quMsg{\roleP}{\stLabi}{}, \stEnd)}$.
	This context is safe and deadlock-free
	as it is non-terminating,
	but it is livelocked as
	the message on $\roleR$'s
	queue will never be dequeued,
	violating L1.
    \end{enumerate}
\end{example}

\section{Properties of Global Types and Local Types}
\label{sec:properties-of-types}
This section introduces the constraint
that well-formed global types must be
$\text{balanced}^+$~(\cref{def:balanced}).
This is an extension of the standard
condition, that well-formed global types
must be balanced,
to our syntax with en-route transitions.
We will show that this condition is
preserved by transitions~(\cref{thm:bal-closed}),
so the restriction to $\text{balanced}^+$ types
is respected by our semantics.
We additionally prove that
this allows us to define
an assortment of recursive
depth functions~(\cref{def:basic-depth-func,def:depth-func,def:actv-role})
and understand
their interactions with our semantics~(\cref{lem:depth-decr,lem:en-route-bounded}).
Following this,
we analyse the algebraic properties of the
full coinductive merge
and the operational correspondence
between sets of local types
and their merge~(\cref{sec:merge-prop}).

\subsection{Balancedness}
\label{sec:bal}
Balanced global types are well-known
to correspond to live protocols in
the synchronous setting,
but in the asynchronous setting we need additional assumptions.
Intuitively, a global type is balanced iff
every participant occurs at a bounded depth in every reachable global type.
To formally define balancedness, we need to define the depth function.
\begin{definition}[Depth Function]
\label{def:basic-depth-func}
For a participant $\roleR$, we define $\gtDepth{\gtG}{\roleR}$ to be a partial function on global types by the following recursive definition:
if $\roleR\notin\gtRoles{\gtG}$ and $\gtFv{\gtG}=\emptyset$ then
$\gtDepth{\gtG}{\roleR}=1$, otherwise
\\[2ex]
\centerline{\(
	\begin{array}{rcl}
		\gtDepth{\gtRec{\gtRecVar}{\gtG}}{\roleR} & = & \gtDepth{\gtG}{\roleR}\\[2ex]
		\gtDepth{\gtComm{\roleP}{\roleQ}{i \in I}{\gtLab[i]}{\tyGround[i]}{\gtG[i]}}{\roleR} & = & 
			\left\{\begin{array}{cc}
                1 & \roleR\in\{\roleP,\roleQ\}\\
				1+max\left\{\gtDepth{\gtG[i]}{\roleR}\suchthat i\in I\right\} & otherwise
			\end{array}\right.\\[2ex]
		\gtDepth{\gtCommSquig{\roleP}{\roleQ}{\gtLab}{}{\gtLab}{\tyGround}{\gtG}}{\roleR} & = & 
            \left\{\begin{array}{cc}
                1 & \roleR=\roleQ\\
				1+\gtDepth{\gtG}{\roleR}& otherwise
			\end{array}\right.
	\end{array}
\)}
\end{definition}

\begin{remark}
\label{rem:depth}
The depth of $\roleR$ in $\gtG$ is the bound on
how deep you can go down the type tree of $\gtG$
before either $\roleR$ occurs
or $\roleR$ is no longer a participant.
\end{remark}

\begin{definition}[Balanced Global Types {\cite[Def 3.3]{Ghilezan2019}\cite[Def.~4.17]{thien-nobuko-popl-25}}]
\label{def:balanced-original}
	A global type $\gtG$ is balanced
	iff for all $\gtG\,\gtMoveStar\,\gtGi$
	and $\roleP\in\gtRoles{\gtGi}$,
	$\gtDepth{\gtG}{\roleP}$ exists.
\end{definition}

\begin{example}[Balanced]
	\begin{equation}
		\gtG = \gtCommRaw{\roleP}{\roleQ}{
			\begin{array}{l}
				\gtCommChoice{\gtLab[0]}{}{
					\gtCommSingle{\roleQ}{\roleP}{\gtLabi}{}{}
				}\\
				\gtCommChoice{\gtLab[1]}{}{
					\gtCommSingle{\roleQ}{\roleP}{\gtLabi}{}{}
				}
			\end{array}
		}
		\qquad
		\gtGi = \gtCommRaw{\roleP}{\roleQ}{
			\begin{array}{l}
				\gtCommChoice{\gtLab[0]}{}{
					\gtCommSingle{\roleR}{\roleS}{\gtLabi}{}{}
				}\\
				\gtCommChoice{\gtLab[1]}{}{
					\gtCommSingle{\roleR}{\roleS}{\gtLabi}{}{}
				}
			\end{array}
		}
		\qquad
		\gtGii = \gtCommRaw{\roleP}{\roleQ}{
			\begin{array}{l}
				\gtCommChoice{\gtLab[0]}{}{
					\gtCommSingle{\roleR}{\roleS}{\gtLabi}{}{}
				}\\
				\gtCommChoice{\gtLab[1]}{}{
					\gtCommSingle{\roleR}{\roleS}{\gtLabii}{}{}
				}
			\end{array}
		}
	\end{equation}
	All three of the above types are balanced,
	their participants all have depth at most $2$
	for each reachable global type.
	Note that $\gtGii$ cannot be projected onto $\roleR$,
	since 
	$\stMerge{}{
	\{\stIntSum{\roleS}{}{\stLabi},\stIntSum{\roleS}{}{\stLabii}\}
	=
	\emptyset
	}$;
	balancedness does not imply projectability.
	
	\noindent
	The ring protocol, $\gtG_\text{\tiny ring}$,
	is balanced as
	$\roleP$, $\roleQ$, and $\roleR$
	are the receivers of messages
	no more than two communications
	from the head of $\gtGi$
	for all $\gtG_\text{\tiny ring}\,\gtMoveStar\,\gtGi$.
\end{example}

\begin{example}[Unbalanced]
	\begin{equation}
\gtG[1] = \gtRec{\gtRecVar}{
			\gtCommRaw{\roleP}{\roleQ}{
				\begin{array}{l}
					\gtCommChoice{\gtLab[0]}{}{
						\gtRecVar
					}
					\\
					\gtCommChoice{\gtLab[1]}{}{
						\gtCommSingle{\roleP}{\roleR}{\gtLab}{}{}
					}
				\end{array}
			}
		}
		\qquad
		\gtG[2] = \gtRec{\gtRecVar}{
			\gtCommRaw{\roleP}{\roleQ}{
				\begin{array}{l}
					\gtCommChoice{\gtLab[0]}{}{
						\gtRecVar
					}
					\\
					\gtCommChoice{\gtLab[1]}{}{
						\gtCommSingle{\roleS}{\roleR}{\gtLab}{}{}
					}
				\end{array}
			}
		}
		\end{equation}
		\begin{equation}
		\gtGi[i] =
			\gtCommSingle{\roleP}{\roleQ}{\gtLab}{}{
				\gtCommSingle{\roleS}{\roleR}{\gtLab}{}{\gtG[i]}
			}
		\quad
		\text{for }i\in\{0,1,2\}
	\end{equation}
	All of the above types are not balanced.
$\gtDepth{\gtG[1]}{\roleR}$ does not exist,
	because there is a path down the unfolding of $\gtG[1]$,
	avoiding $\roleR$, so $\gtG[1]$ is unbalanced;
	this corresponds to the specified protocol not
	being live.
	$\gtDepth{\gtG[2]}{\roleR}$ and $\gtDepth{\gtG[2]}{\roleS}$
	do not exist,
	because there is a path down the unfolding of $\gtG[1]$,
	avoiding $\roleR$ and $\roleS$, so $\gtG[2]$ is unbalanced;
	this corresponds to the specified protocol
	not allowing commutativity of independent actions
	\ie
	\begin{equation}
		\gtG[2]\,\gtMove[\ltsSel{\roleP}{\roleQ}{\gtLab[1]}]\,
		\cdot
		\,\gtMove[\ltsSel{\roleS}{\roleR}{\gtLab}]\,
		\cdot
		\text{ but }
		\gtG[2]\quad\not\!\!\!\!\!\!\!\gtMove[\ltsSel{\roleS}{\roleR}{\gtLab}]
	\end{equation}
	$\gtGi[i]$ is not balanced for $i\in\{
1,2\}$,
	since, although all depths exists at $\gtGi[i]$,
	$\gtGi[i]\,\gtMoveStar\,\gtG[i]$,
	which is not balanced.
	
	\noindent
$\gtG[1]$ is projectable and operationally corresponds to
	its projected context, but does not guarantee liveness.
	$\gtG[2]$ is projectable but does not correspond to its projected
	context, since contexts always have commutativity
	for independent actions.
	\noindent
	Projectability does not imply balancedness and
	balancedness is necessary to ensure correctness of the projection.
\end{example}

The positions of en-route transmissions in a global type
can cause unexpected behaviour.
For example,
consider the global type
\begin{equation}
	\gtG = \gtCommSingle{\roleP}{\roleQ}{\gtLab}{}{
		\gtCommSquigSingle{\roleP}{\roleQ}{\gtLabi}{\gtLabi}{}{}
	}
\end{equation}
The projection of $\gtG$ 
has $\quMsg{\roleQ}{\gtLabi}{}$ on $\roleP$'s
queue, and
$\roleQ$'s local type is
$\stExtSum{\roleP}{}{\stLab\stSeq\stExtSum{\roleP}{}{\stLabi}}$,
so that the projected context is \emph{unsafe}.
\\[2ex]
We will introduce an extension
of balancedness for asynchronous
global types,
which forbids this behaviour.
Suppose that $\gtG$
has no en-route messages
and $\gtG\,\gtMoveStar\,\gtGi$.
An important observation is that for all $\roleP$ and $\roleQ$, every path down $\gtGi$ encounters the same
number of en-route transmissions from $\roleP$ to $\roleQ$.
We introduce a recursively defined partial function, $\gtMCount{\gtG}{\roleP}{\roleQ}$, to formalise this.

\begin{definition}[Active Roles]
\label{def:actv-role}
We define $\gtMRoles{\gtG}$ on global types
to be the set of pairs of participants that have an en-route transmission
from the first to the second,
by the recursive definition:
$\gtMRoles{\gtEnd}= \emptyset$; 
$\gtMRoles{\gtRecVar}  =  \emptyset$; 
$\gtMRoles{\gtRec{\gtRecVar}{\gtG}}  =  \gtMRoles{\gtG}$; 
$\gtMRoles{\gtComm{\roleP}{\roleQ}{i \in
I}{\gtLab[i]}{\tyGround[i]}{\gtG[i]}} = \bigcup_{i\in
I}\gtMRoles{\gtG[i]}$; and
$\gtMRoles{\gtCommSquig{\roleP}{\roleQ}{\gtLab}{}{\gtLab}{\tyGround}{\gtG}} = \{(\roleP,\roleQ)\}\cup\gtMRoles{\gtG}$
\end{definition}

\begin{definition}[Counting En-Route Messages]
\label{def:count}
We define $\gtMCount{\gtG}{\roleP}{\roleQ}$ by the following recursive definition:
\newline
\smallskip
\centerline
{\(
	\begin{array}{rcl}
		\gtMCount{\gtEnd}{\roleP}{\roleQ} & = & 0\\[4ex]
		\gtMCount{\gtRecVar}{\roleP}{\roleQ} & = & 0\\[4ex]
		\gtMCount{\gtRec{\gtRecVar}{\gtG}}{\roleP}{\roleQ} & = & \left\{
			\begin{array}{cc}
			\gtMCount{\gtG}{\roleP}{\roleQ} & \gtRecVar\not\in\gtFv{\gtG}\text{ or }\gtMCount{\gtG}{\roleP}{\roleQ}=0\\
			\text{undefined} & \text{otherwise}
			\end{array}
		\right.\\[4ex]
		\gtMCount{\gtComm{\rolePi}{\roleQi}{i \in I}{\gtLab[i]}{\tyGround[i]}{\gtG[i]}}{\roleP}{\roleQ} & = & \left\{
			\begin{array}{cc}
			\gtMCount{\gtG[i]}{\roleP}{\roleQ} & \begin{array}{c}
				\text{for all }i,j\in I\\
                \gtMCount{\gtG[i]}{\roleP}{\roleQ}=\gtMCount{\gtG[j]}{\roleP}{\roleQ}\\
                (\roleP,\roleQ)\ne(\rolePi,\roleQi)
			\end{array}\\[4ex]
            0&(\roleP,\roleQ)=(\rolePi,\roleQi)\not\in\bigcup_{i\in I}\gtMRoles{\gtG[i]}\\[4ex]
			\text{undefined} & \text{otherwise}
			\end{array}
		\right.\\[4ex]
		\gtMCount{\gtCommSquig{\rolePi}{\roleQi}{\gtLab}{}{\gtLab}{\tyGround}{\gtG}}{\roleP}{\roleQ} & = & \left\{
			\begin{array}{cc}
			\gtMCount{\gtG}{\roleP}{\roleQ}+1 & \roleP=\rolePi\text{ and }\roleQ=\roleQi\\
			\gtMCount{\gtG}{\roleP}{\roleQ} & otherwise
			\end{array}
		\right.
	\end{array}
\)}
\end{definition}

En-route messages are projected to queue elements.
The existence of $\gtMCount{\gtG}{\roleP}{\roleQ}$
is necessary for this interpretation to be correct:
\begin{enumerate*}
	\item queues have a finite length,
	which is assured by the conditions on recursion;
	\item queues have a known length,
	which is assured by the conditions on independent communications; and
	\item elements cannot be enqueued prematurely,
	which is ensured by the condition that the count
	of en-route transmissions from $\roleP$ to $\roleQ$
	following a transmission from $\roleP$ to $\roleQ$
	must be $0$. 
\end{enumerate*}
The existence of this function
is the necessary
condition for global types to be
asynchronously
meaningful.

\begin{definition}[$\text{Balanced}^{+}$ Global Types]
\label{def:balanced}
A global type $\gtG$ is $\text{balanced}^{+}$ iff, $\gtG$ is balanced (\cref{def:balanced-original})
and
for all participants $\roleP$ and $\roleQ$,
$\gtMCount{\gtG}{\roleP}{\roleQ}$ exists.
Hereafter we assume well-formed
global types satisfy this condition. 
\end{definition}
Given en-route types are only runtime behaviour, we can impose this
additional restriction upon balanced types by
Theorem~\ref{thm:bal-closed}.

\begin{remark}
Checking balancedness is decidable \cite{thien-nobuko-popl-25},
by which the existence of $\gtMCount{\gtG}{\roleP}{\roleQ}$
is decidable.
Therefore, by checking all $\roleP,\roleQ\in\gtRoles{\gtG}$,
$\text{balanced}^+$ is a decidable property.
\end{remark}

We now return to our running example~(\cref{sec:ring})
to briefly check that it satisfies this condition.

\begin{example}[Ring Protocol is $\text{Balanced}^+$]
	$\gtG_\text{\tiny ring}$ has no en-route transmissions
	and is balanced,
	so $\gtG_\text{\tiny ring}$ is $\text{balanced}^+$.
	As $\gtG_\text{\tiny ring}\,\gtMoveStar\,
	\gtCommSquigSingle{\roleP}{\roleQ}{\gtMsgFmt{add}}{\mathsf{add}}{\tyInt}
      {
        \gtCommRaw{\roleQ}{\roleR}{
          \begin{array}{@{}l@{}}
          \mathsf{add}(\tyInt)\gtSeq
          \gtCommRaw{\roleR}{\roleP}{
            \mathsf{add}(\tyInt)\gtSeq  \gtG_\text{ring}
          }
          \\
          \mathsf{sub}(\tyInt)\gtSeq
          \gtCommRaw{\roleR}{\roleP}{
            \mathsf{sub}(\tyInt)\gtSeq  \gtG_\text{ring}
          }
          \end{array}
        }
    }$, this type is also $\text{balanced}^+$.
\end{example}
Of course, not all global types are $\text{balanced}^+$.
We now provide examples for each
way that $\text{balanced}^+$ can
fail to hold.
\begin{example}[Non-$\text{Balanced}^+$ Types]
	Any non-balanced type is not $\text{balanced}^+$.
	For example,
	\begin{equation}
		\gtG = \gtRec{\gtRecVar}{
			\gtCommRaw{\roleP}{\roleQ}{
				\begin{array}{@{}l@{}}
					\gtLab[1]\gtSeq\gtRecVar\\
					\gtLab[2]\gtSeq\gtCommSingle{\roleP}{\roleR}{\gtLab}{}{}
				\end{array}
			}
		}
	\end{equation}
	$\gtG$ is not balanced as $\roleR$ has unbounded depth,
	so $\gtG$ is not $\text{balanced}^+$.
	\\[2ex]
	The more interesting violations of $\text{balanced}^+$
	can be observed by occurrences of en-route transmissions.
	None can occur beneath non-trivial recursive binders
	and they must appear in types uniformly.
	Consider the following global types:
	\begin{equation}
		\begin{array}{c}
		\gtG[1] = \gtRec{\gtRecVar}{\gtCommSquigSingle{\roleP}{\roleQ}{\gtLab}{\gtLab}{}{\gtRecVar}}
		\qquad
		\gtG[2] = \gtRec{\gtRecVar}{\gtCommSquigSingle{\roleP}{\roleQ}{\gtLab}{\gtLab}{}{}}
		\qquad
		\gtG[3] = \gtCommSingle{\roleP}{\roleQ}{\gtLabi}{}{
			\gtCommSquigSingle{\roleP}{\roleQ}{\gtLab}{\gtLab}{}{}
		}
		\qquad
		\gtG[4] = \gtCommRaw{\rolePi}{\roleQi}{
			\begin{array}{@{}l@{}}
				\gtLab[1]\gtSeq\gtCommSingle{\roleP}{\roleQ}{\gtLab}{}{}\\
				\gtLab[2]\gtSeq\gtCommSquigSingle{\roleP}{\roleQ}{\gtLab}{\gtLab}{}{}
			\end{array}
		}
		\end{array}
	\end{equation}
	Only $\gtG[2]$ is $\text{balanced}^+$,
	and none of them can
	occur as a transition
	from a global type with no en-route transmissions.
	$\gtG[1]$ cannot occur as it
	has an en-route transmission under
	a recursive binder and followed
	by a recursive variable,
	so this global type would suggest
	that infinitely many messages
	are queued from $\roleP$ to $\roleQ$.
	$\gtG[2]$ cannot occur for the same reason,
	however this type is still meaningful
	as it is equivalent to
	$\gtCommSquigSingle{\roleP}{\roleQ}{\gtLab}{\gtLab}{}{}$.
	$\gtG[3]$ cannot occur
	as global type
	transitions require the top-most
	message from
	$\roleP$ to $\roleQ$
	to be sent
	before any later messages can be sent.
	$\gtG[4]$ cannot occur
	as the en-route transmission
	only occurs down one branch,
	whereas global transitions
result in en-route transmissions occurring
	down every path in the result.
\end{example}

\begin{theorem}[Transitions preserve En-Route Counting]
\label{thm:bal-closed}
If $\gtMCount{\gtG}{\roleP}{\roleQ}$ exists
and $\gtG\,\gtMove\,\gtGi$,
then $\gtMCount{\gtGi}{\roleP}{\roleQ}$ exists.
Therefore, if $\gtG\,\gtMoveStar\,\gtGi$
and $\gtG$ has no en-route transmissions,
then $\gtMCount{\gtGi}{\roleP}{\roleQ}$
exists for all
$\roleP$ and $\roleQ$.
\end{theorem}
\begin{proof}
Suppose that $\gtMCount{\gtG}{\roleP}{\roleQ}$ exists.
Suppose that $\gtG\,\gtMove[\stEnvAnnotGenericSym]\,\gtGi$.
We now show that: if $\stEnvAnnotGenericSym=\ltsSel{\roleP}{\roleQ}{\gtLab(\tyGround)}$, then
$\gtMCount{\gtGi}{\roleP}{\roleQ}=\gtMCount{\gtG}{\roleP}{\roleQ}+1$;
if $\stEnvAnnotGenericSym=\ltsBra{\roleQ}{\roleP}{\gtLab(\tyGround)}$, then
$\gtMCount{\gtGi}{\roleP}{\roleQ}=\gtMCount{\gtG}{\roleP}{\roleQ}-1$;
and $\gtMCount{\gtGi}{\roleP}{\roleQ}=\gtMCount{\gtG}{\roleP}{\roleQ}$ otherwise.
We proceed by induction on the derivation of $\gtG\,\gtMove[\stEnvAnnotGenericSym]\,\gtGi$.
\begin{itemize}[leftmargin=0.8in,labelindent=-\leftmargin]
	\item[\inferrule{\iruleGtMoveRec}] $\gtG=\gtRec{\gtRecVar}{\gtGii}$ and $\gtGii{}[\gtRec{\gtRecVar}{\gtGii}/\gtRecVar] 
    \,\gtMove[\stEnvAnnotGenericSym]\, 
    \gtGi$.
	$\gtRecVar\not\in\gtFv{\gtGii}$ so $\gtGii=\gtGii{}[\gtRec{\gtRecVar}{\gtGii}/\gtRecVar]$ and
	$\gtMCount{\gtG}{\roleP}{\roleQ}=\gtMCount{\gtGii}{\roleP}{\roleQ}$,
	so apply the I.H.

	\item[\inferrule{\iruleGtMoveBra}] $\gtG=\gtCommSquigSmall{\rolePi}{\roleQi}{\gtLab}{}{\gtLab}{\tyGround}{\gtGi}$, and $\stEnvAnnotGenericSym=\ltsBra{\roleQi}{\rolePi}{\gtLab(\tyGround)}$.
	If $(\rolePi,\roleQi)=(\roleP,\roleQ)$, $\gtMCount{\gtGi}{\roleP}{\roleQ}=\gtMCount{\gtG}{\roleP}{\roleQ}-1$ by definition.
	Otherwise, $\gtMCount{\gtGi}{\roleP}{\roleQ}=\gtMCount{\gtG}{\roleP}{\roleQ}$ by definition.

	\item[\inferrule{\iruleGtMoveSel}] $\gtG=\gtCommSmall{\rolePi}{\roleQi}{i \in I}{\gtLab[i]}{\tyGround[i]}{\gtGi[i]}$,
	$\stEnvAnnotGenericSym = \ltsSel{\rolePi}{\roleQi}{\gtLab[j](\tyGround)}$, and
	$\gtGi=\gtCommSquigSmall{\rolePi}{\roleQi}{\gtLab[j]}{}{\gtLab[j]}{\tyGround[j]}{\gtGi[j]}$.
	If $(\rolePi,\roleQi)=(\roleP,\roleQ)$, $\gtMCount{\gtGi}{\roleP}{\roleQ}=\gtMCount{\gtG}{\roleP}{\roleQ}+1$ by definition.
	Otherwise, $\gtMCount{\gtGi}{\roleP}{\roleQ}=\gtMCount{\gtG}{\roleP}{\roleQ}$ by definition.

	\item[\inferrule{\iruleGtMoveCtx(')}] $\gtG=\gtCommSmall{\rolePi}{\roleQi}{i \in I}{\gtLab[i]}{\tyGround[i]}{\gtGi[i]}$,
	$\gtGi = \gtCommSmall{\rolePi}{\roleQi}{i \in J}{\gtLab[i]}{\tyGround[i]}{\gtGii[i]}$ and
	$\gtGi[i]\,\gtMove[\stEnvAnnotGenericSym]\,\gtGii[i]$ for $i\in J\subseteq I$.
	If $\stEnvAnnotGenericSym=\ltsSel{\roleP}{\roleQ}{\gtLab(\tyGround)}$, then
	$\gtMCount{\gtGi[i]}{\roleP}{\roleQ}+1=\gtMCount{\gtGii[i]}{\roleP}{\roleQ}$ for $i\in J$ by the I.H.
	so $\gtMCount{\gtG}{\roleP}{\roleQ}+1=\gtMCount{\gtGi}{\roleP}{\roleQ}$.
	If $\stEnvAnnotGenericSym=\ltsBra{\roleQ}{\roleP}{\gtLab(\tyGround)}$, then
	$\gtMCount{\gtGi[i]}{\roleP}{\roleQ}-1=\gtMCount{\gtGii[i]}{\roleP}{\roleQ}$ for $i\in J$ by the I.H.
	so $\gtMCount{\gtG}{\roleP}{\roleQ}-1=\gtMCount{\gtGi}{\roleP}{\roleQ}$.
	Otherwise,
	$\gtMCount{\gtGi[i]}{\roleP}{\roleQ}=\gtMCount{\gtGii[i]}{\roleP}{\roleQ}$ for $i\in J$ by the I.H.
	so $\gtMCount{\gtG}{\roleP}{\roleQ}=\gtMCount{\gtGi}{\roleP}{\roleQ}$.

	\item[\inferrule{\iruleGtMoveCtxII}] $\gtG=\gtCommSquigSmall{\rolePi}{\roleQi}{\gtLabi}{}{\gtLabi}{\tyGroundi}{\gtGii}$,
	$\gtGi = \gtCommSquigSmall{\rolePi}{\roleQi}{\gtLabi}{}{\gtLabi}{\tyGroundi}{\gtGiii}$ and
	$\gtGii\,\gtMove[\stEnvAnnotGenericSym]\,\gtGiii$.
	If $\stEnvAnnotGenericSym=\ltsSel{\roleP}{\roleQ}{\gtLab(\tyGround)}$, then
	$\gtMCount{\gtGii}{\roleP}{\roleQ}+1=\gtMCount{\gtGiii}{\roleP}{\roleQ}$ by the I.H.
	so $\gtMCount{\gtG}{\roleP}{\roleQ}+1=\gtMCount{\gtGi}{\roleP}{\roleQ}$.
	If $\stEnvAnnotGenericSym=\ltsBra{\roleQ}{\roleP}{\gtLab(\tyGround)}$, then
	$\gtMCount{\gtGii}{\roleP}{\roleQ}-1=\gtMCount{\gtGiii}{\roleP}{\roleQ}$ by the I.H.
	so $\gtMCount{\gtG}{\roleP}{\roleQ}-1=\gtMCount{\gtGi}{\roleP}{\roleQ}$.
	Otherwise,
	$\gtMCount{\gtGii}{\roleP}{\roleQ}=\gtMCount{\gtGiii}{\roleP}{\roleQ}$ by the I.H.
	so $\gtMCount{\gtG}{\roleP}{\roleQ}=\gtMCount{\gtGi}{\roleP}{\roleQ}$.
\end{itemize}
We now note that for $\gtG$ without en-route transmissions,
$\gtMCount{\gtG}{\roleP}{\roleQ}=0$ for all $\roleP$,$\roleQ$.
Therefore, if $\gtG\,\gtMoveStar\,\gtGi$, then $\gtMCount{\gtGi}{\roleP}{\roleQ}$ exists for all $\roleP$,$\roleQ$.
\end{proof}
\noindent
Protocols begin with empty queues,
thus \cref{thm:bal-closed} says that
$\text{balanced}^+$ is
no more restrictive than the
standard balanced requirement.

We can now use $\text{balanced}^+$
to recursively define a
helper function in~\cref{def:depth-func}. 
Later on, we make use of this function to structure 
inductive proofs using~\cref{lem:depth-decr}. 

\begin{definition}[Message Depth Function]
\label{def:depth-func}
For participants $\roleP$ and $\roleQ$, we define $\gtMDepth{\gtG}{\roleP}{\roleQ}$,
read as `en-route message depth of $(\roleP,\roleQ)$',
to be a partial function on global types by the following recursive definition:
\\[2ex]
\centerline{\(
	\begin{array}{rcl}
		\gtMDepth{\gtRec{\gtRecVar}{\gtG}}{\roleP}{\roleQ} & = & \gtMDepth{\gtG}{\roleP}{\roleQ}\\[2ex]
		\gtMDepth{\gtComm{\rolePi}{\roleQi}{i \in I}{\gtLab[i]}{\tyGround[i]}{\gtG[i]}}{\roleP}{\roleQ} & = & 
				1+max\left\{\gtMDepth{\gtG[i]}{\roleP}{\roleQ}\suchthat i\in I\right\}\\[2ex]
		\gtMDepth{\gtCommSquig{\rolePi}{\roleQi}{\gtLab}{}{\gtLab}{\tyGround}{\gtG}}{\roleP}{\roleQ} & = & 
            \left\{\begin{array}{cc}
                1 & \roleP=\rolePi\text{ and }\roleQ=\roleQi\\
				1+\gtMDepth{\gtG}{\roleP}{\roleQ} & otherwise
			\end{array}\right.
	\end{array}
\)}
\end{definition}

\begin{remark}
Similarly to \cref{rem:depth}, $\gtMDepth{\gtG}{\roleP}{\roleQ}\leq n$ iff every path of length $n$ down $\gtG$ contains an en-route transmission from $\roleP$ to $\roleQ$.
\end{remark}

It is useful to understand how depths interact with global type semantics.
Transitions excluding $\roleP$ either remove branches
or remove type constructors excluding $\roleP$,
thus decreasing its depth.

\begin{lemma}[Depth is decreasing]
\label{lem:depth-decr}
Suppose that $\gtG\,\gtMove[\stEnvAnnotGenericSym]\,\gtGi$, then:
\begin{itemize}
    \item If $\roleP\ne\ltsSubject{\stEnvAnnotGenericSym}$ and $\gtDepth{\gtG}{\roleP}$ exists, then
    $\gtDepth{\gtGi}{\roleP}\leq \gtDepth{\gtG}{\roleP}$.
    \item If $\stEnvAnnotGenericSym$ is not of the form $\ltsBra{\roleP}{\roleQ}{\gtLab(\tyGround)}$ and $\gtMDepth{\gtG}{\roleP}{\roleQ}$ exists, then
    $\gtMDepth{\gtGi}{\roleP}{\roleQ}\leq \gtMDepth{\gtG}{\roleP}{\roleQ}$.
\end{itemize}
\end{lemma}

\begin{proof}
Proceed by induction on the definition of $\gtG\,\gtMove[\stEnvAnnotGenericSym]\,\gtGi$.
\newline
Suppose that $\stEnvAnnotGenericSym\neq\ltsSel{\roleP}{\roleQ}{\gtLab}$
or $\ltsBra{\roleP}{\roleQ}{\gtLab}$ for any $\roleQ$,
and that $\gtDepth{\gtG}{\roleP}$ exists.
Suppose that $\gtG\,\gtMove[\stEnvAnnotGenericSym]\,\gtGi$ and consider the cases for the last transition rule:
\begin{itemize}[leftmargin=0.8in,labelindent=-\leftmargin]
	\item[\inferrule{\iruleGtMoveRec}] $\gtG=\gtRec{\gtRecVar}{\gtGii}$ and $\gtGii{}[\gtRec{\gtRecVar}{\gtGii}/\gtRecVar]
    \,\gtMove[\stEnvAnnotGenericSym]\,
    \gtGi$.
	$\gtDepth{\gtG}{\roleP}=\gtDepth{\gtGii}{\roleP}=\gtDepth{\gtGii{}[\gtRec{\gtRecVar}{\gtGii}/\gtRecVar]}{\roleP}$
	so by the I.H, $\gtDepth{\gtGi}{\roleP}\leq \gtDepth{\gtGii{}[\gtRec{\gtRecVar}{\gtGii}/\gtRecVar]}{\roleP}=
	\gtDepth{\gtG}{\roleP}$.
	\item[\inferrule{\iruleGtMoveBra}] $\gtG=\gtCommSquigSmall{\rolePi}{\roleQ}{\gtLab}{}{\gtLab}{\tyGround}{\gtGi}$, and $\stEnvAnnotGenericSym=\ltsBra{\roleQ}{\rolePi}{\gtLab(\tyGround)}$.
	Now, $\gtDepth{\gtG}{\roleP}=\gtDepth{\gtGi}{\roleP}+1>\gtDepth{\gtGi}{\roleP}$.
	\item[\inferrule{\iruleGtMoveSel}] $\gtG=\gtCommSmall{\rolePi}{\roleQ}{i \in I}{\gtLab[i]}{\tyGround[i]}{\gtGi[i]}$,
	$\stEnvAnnotGenericSym = \ltsSel{\rolePi}{\roleQ}{\gtLab[j](\tyGround[j])}$, and
	$\gtGi=\gtCommSquigSmall{\rolePi}{\roleQ}{\gtLab[j]}{}{\gtLab[j]}{\tyGround[j]}{\gtGi[j]}$.
	Now, $\gtDepth{\gtG}{\roleP}=max_{i\in I}\gtDepth{\gtGi[i]}{\roleP}+1\geq\gtDepth{\gtGi}{\roleP}$ if $\roleP\neq\roleQ$
	and $\gtDepth{\gtG}{\roleP}=\gtDepth{\gtGi}{\roleP}=1$ if $\roleP=\roleQ$.
	\item[\inferrule{\iruleGtMoveCtx}] $\gtG=\gtCommSmall{\rolePi}{\roleQi}{i \in I}{\gtLab[i]}{\tyGround[i]}{\gtGi[i]}$,
	$\gtGi = \gtCommSmall{\rolePi}{\roleQi}{i \in I}{\gtLab[i]}{\tyGround[i]}{\gtGii[i]}$ and
	$\gtGi[i]\,\gtMove[\stEnvAnnotGenericSym]\,\gtGii[i]$ for $i\in I$.
	If $\roleP\in\{\rolePi,\roleQi\}$, $\gtDepth{\gtG}{\roleP}=\gtDepth{\gtGi}{\roleP}=1$.
	Otherwise, $\gtDepth{\gtGi[i]}{\roleP}<\gtDepth{\gtG}{\roleP}$ for $i\in I$ so by the I.H.,
	$\gtDepth{\gtGii[i]}{\roleP}\leq \gtDepth{\gtGi[i]}{\roleP}< \gtDepth{\gtG}{\roleP}$ for $i\in I$, so
	$\gtDepth{\gtGi}{\roleP}\leq \gtDepth{\gtG}{\roleP}$.
	\item[\inferrule{\iruleGtMoveCtx'}] $\gtG=\gtCommSmall{\rolePi}{\roleQi}{i \in I}{\gtLab[i]}{\tyGround[i]}{\gtGi[i]}$,
	$\gtGi = \gtCommSmall{\rolePi}{\roleQi}{i \in J}{\gtLab[i]}{\tyGround[i]}{\gtGii[i]}$ and
	$\gtGi[i]\,\gtMove[\stEnvAnnotGenericSym]\,\gtGii[i]$ for $i\in J\subseteq I$.
	If $\roleP\in\{\rolePi,\roleQi\}$, $\gtDepth{\gtG}{\roleP}=\gtDepth{\gtGi}{\roleP}=1$.
	Otherwise, $\gtDepth{\gtGi[i]}{\roleP}<\gtDepth{\gtG}{\roleP}$ for $i\in J$ so by the I.H.,
	$\gtDepth{\gtGii[i]}{\roleP}\leq \gtDepth{\gtGi[i]}{\roleP}< \gtDepth{\gtG}{\roleP}$ for $i\in J$, so
	$\gtDepth{\gtGi}{\roleP}\leq \gtDepth{\gtG}{\roleP}$.
    \item[\inferrule{\iruleGtMoveCtxII}] $\gtG=\gtCommSquigSmall{\rolePi}{\roleQi}{\gtLab}{}{\gtLab}{\tyGround}{\gtGii}$,
	$\gtGi = \gtCommSquigSmall{\rolePi}{\roleQi}{\gtLabi}{}{\gtLabi}{\tyGroundi}{\gtGiii}$ and
	$\gtGii\,\gtMove[\stEnvAnnotGenericSym]\,\gtGiii$.
	If $\roleP=\roleQi$, $\gtDepth{\gtG}{\roleP}=\gtDepth{\gtGi}{\roleP}=1$.
	Otherwise, $\gtDepth{\gtGii}{\roleP}<\gtDepth{\gtG}{\roleP}$ so by the I.H.,
	$\gtDepth{\gtGiii}{\roleP}\leq \gtDepth{\gtGii}{\roleP}< \gtDepth{\gtG}{\roleP}$, so
	$\gtDepth{\gtGi}{\roleP}\leq \gtDepth{\gtG}{\roleP}$.
\end{itemize}
Suppose that $\stEnvAnnotGenericSym$ is not of the form $\ltsBra{\roleP}{\roleQ}{\gtLab(\tyGround)}$
and $\gtMDepth{\gtG}{\roleP}{\roleQ}$ exists.
Suppose that $\gtG\,\gtMove[\stEnvAnnotGenericSym]\,\gtGi$ and consider the last case of the transition rule:
\begin{itemize}[leftmargin=0.8in,labelindent=-\leftmargin]
	\item[\inferrule{\iruleGtMoveRec}] $\gtG=\gtRec{\gtRecVar}{\gtGii}$ and $\gtGii{}[\gtRec{\gtRecVar}{\gtGii}/\gtRecVar]
    \,\gtMove[\stEnvAnnotGenericSym]\,
    \gtGi$.
	$\gtMDepth{\gtG}{\roleP}{\roleQ}=\gtMDepth{\gtGii}{\roleP}{\roleQ}=\gtMDepth{\gtGii{}[\gtRec{\gtRecVar}{\gtGii}/\gtRecVar]}{\roleP}{\roleQ}$
	so by the I.H, $\gtMDepth{\gtGi}{\roleP}{\roleQ}\leq \gtMDepth{\gtGii{}[\gtRec{\gtRecVar}{\gtGii}/\gtRecVar]}{\roleP}{\roleQ}=
	\gtMDepth{\gtG}{\roleP}{\roleQ}$.
	\item[\inferrule{\iruleGtMoveBra}] $\gtG=\gtCommSquigSmall{\rolePi}{\roleQi}{\gtLab}{}{\gtLab}{\tyGround}{\gtGi}$,  and $\stEnvAnnotGenericSym=\ltsBra{\roleQi}{\rolePi}{\gtLab(\tyGround)}$.
	So, $\roleP\neq \roleQi$ or $\roleQ\neq\rolePi$.
	Now, $\gtMDepth{\gtG}{\roleP}{\roleQ}=\gtMDepth{\gtGi}{\roleP}{\roleQ}+1>\gtMDepth{\gtGi}{\roleP}{\roleQ}$.
	\item[\inferrule{\iruleGtMoveSel}] $\gtG=\gtCommSmall{\rolePi}{\roleQi}{i \in I}{\gtLab[i]}{\tyGround[i]}{\gtGi[i]}$,
	$\stEnvAnnotGenericSym = \ltsSel{\rolePi}{\roleQi}{\gtLab[j](\tyGround[j])}$, and
	$\gtGi=\gtCommSquigSmall{\rolePi}{\roleQi}{\gtLab[j]}{}{\gtLab[j]}{\tyGround[j]}{\gtGi[j]}$.
	Now, $\gtMDepth{\gtG}{\roleP}{\roleQ}=\gtMDepth{\gtGi}{\roleP}{\roleQ}$
	if $(\rolePi,\roleQi)\neq(\roleP,\roleQ)$ and
	$\gtMDepth{\gtGi}{\roleP}{\roleQ}=1\leq\gtMDepth{\gtG}{\roleP}{\roleQ}$ otherwise.
	\item[\inferrule{\iruleGtMoveCtx}] $\gtG=\gtCommSmall{\rolePi}{\roleQi}{i \in I}{\gtLab[i]}{\tyGround[i]}{\gtGi[i]}$,
	$\gtGi = \gtCommSmall{\rolePi}{\roleQi}{i \in I}{\gtLab[i]}{\tyGround[i]}{\gtGii[i]}$ and
	$\gtGi[i]\,\gtMove[\stEnvAnnotGenericSym]\,\gtGii[i]$ for $i\in I$.
	$\gtMDepth{\gtGi[i]}{\roleP}{\roleQ}<\gtMDepth{\gtG}{\roleP}{\roleQ}$ for $i\in I$ so by the I.H.,
	$\gtMDepth{\gtGii[i]}{\roleP}{\roleQ}\leq \gtMDepth{\gtGi[i]}{\roleP}{\roleQ}< \gtMDepth{\gtG}{\roleP}{\roleQ}$ for $i\in I$, so
	$\gtMDepth{\gtGi}{\roleP}{\roleQ}\leq \gtMDepth{\gtG}{\roleP}{\roleQ}$.
	\item[\inferrule{\iruleGtMoveCtx'}] $\gtG=\gtCommSmall{\rolePi}{\roleQi}{i \in I}{\gtLab[i]}{\tyGround[i]}{\gtGi[i]}$,
	$\gtGi = \gtCommSmall{\rolePi}{\roleQi}{i \in J}{\gtLab[i]}{\tyGround[i]}{\gtGii[i]}$ and
	$\gtGi[i]\,\gtMove[\stEnvAnnotGenericSym]\,\gtGii[i]$ for $i\in J\subseteq I$.
	$\gtMDepth{\gtGi[i]}{\roleP}{\roleQ}<\gtMDepth{\gtG}{\roleP}{\roleQ}$ for $i\in J$ so by the I.H.,
	$\gtMDepth{\gtGii[i]}{\roleP}{\roleQ}\leq \gtMDepth{\gtGi[i]}{\roleP}{\roleQ}< \gtMDepth{\gtG}{\roleP}{\roleQ}$ for $i\in J$, so
	$\gtMDepth{\gtGi}{\roleP}{\roleQ}\leq \gtMDepth{\gtG}{\roleP}{\roleQ}$.
	\item[\inferrule{\iruleGtMoveCtxII}] $\gtG=\gtCommSquigSmall{\rolePi}{\roleQi}{\gtLabi}{}{\gtLabi}{\tyGroundi}{\gtGii}$,
	$\gtGi = \gtCommSquigSmall{\rolePi}{\roleQi}{\gtLabi}{}{\gtLabi}{\tyGroundi}{\gtGiii}$ and
	$\gtGii\,\gtMove[\stEnvAnnotGenericSym]\,\gtGiii$.
	If $\roleP=\rolePi$ and $\roleQ=\roleQi$, $\gtMDepth{\gtG}{\roleP}{\roleQ}=\gtMDepth{\gtGi}{\roleP}{\roleQ}=1$.
	Otherwise, $\gtMDepth{\gtGii}{\roleP}{\roleQ}<\gtMDepth{\gtG}{\roleP}{\roleQ}$ so by the I.H.,
	$\gtMDepth{\gtGiii}{\roleP}{\roleQ}\leq \gtMDepth{\gtGii}{\roleP}{\roleQ}< \gtMDepth{\gtG}{\roleP}{\roleQ}$, so
	$\gtMDepth{\gtGi}{\roleP}{\roleQ}\leq \gtMDepth{\gtG}{\roleP}{\roleQ}$.
\end{itemize}
\end{proof}

We need to proceed by induction on $\gtMDepth{\gtG}{\roleP}{\roleQ}$
when trying to invert a projection onto $\roleQ$
with a receive from $\roleP$ at the head
and a message from $\roleP$ to $\roleQ$ on the queue.
We know that the $\gtDepth{\gtG}{\roleP}$ exists for all $\roleP\in\gtRoles{\gtG}$
by balancedness;
we would like a similar statement for $\gtMDepth{}{}{}$
and $\text{balanced}^+$.
$\text{Balanced}^+$ ensures the existence of
$\gtMDepth{\gtG}{\roleP}{\roleQ}$ for all $(\roleP,\roleQ)\in\gtMRoles{\gtG}$.

\begin{lemma}[En Route Transmissions are Bounded]
\label{lem:en-route-bounded}
If $\gtG$ is $\text{balanced}^+$, then
for $(\roleP,\roleQ)\in\gtMRoles{\gtG}$, $\gtMDepth{\gtG}{\roleP}{\roleQ}$ exists.
\end{lemma}
\begin{proof}
We show that if $\gtMCount{\gtG}{\roleP}{\roleQ}>0$, then $\gtMDepth{\gtG}{\roleP}{\roleQ}$ exists by
induction on the structure of $\gtG$.
\begin{itemize}
\item If $\gtG=\gtEnd$, then $\gtMCount{\gtG}{\roleP}{\roleQ}=0$.
\item If $\gtG=\gtRecVar$, then $\gtMCount{\gtG}{\roleP}{\roleQ}=0$.
If $\gtG=\gtRec{\gtRecVar}{\gtGi}$ and $\gtMCount{\gtG}{\roleP}{\roleQ}>0$, then
$\gtMCount{\gtGi}{\roleP}{\roleQ}>0$ so $\gtMDepth{\gtGi}{\roleP}{\roleQ}$ exists by I.H.,
so $\gtMDepth{\gtG}{\roleP}{\roleQ}=\gtMDepth{\gtGi}{\roleP}{\roleQ}$.
\item If $\gtG=\gtComm{\rolePi}{\roleQi}{i \in I}{\gtLab[i]}{\tyGround[i]}{\gtG[i]}$, then
$\gtMCount{\gtG[i]}{\roleP}{\roleQ}>0$ for $i\in I$,
so $\gtMDepth{\gtG[i]}{\roleP}{\roleQ}$ exists for $i\in I$ by the I.H.,
so $\gtMDepth{\gtG}{\roleP}{\roleQ}=1+max_{i\in I}\gtMDepth{\gtG[i]}{\roleP}{\roleQ}$.
\item If $\gtG=\gtCommSquig{\rolePi}{\roleQi}{\gtLab}{}{\gtLab}{\tyGround}{\gtGi}$
and $(\roleP,\roleQ)\neq(\rolePi,\roleQi)$, then
$\gtMCount{\gtGi}{\roleP}{\roleQ}>0$,
so $\gtMDepth{\gtGi}{\roleP}{\roleQ}$ exists by the I.H.,
so $\gtMDepth{\gtG}{\roleP}{\roleQ}=1+\gtMDepth{\gtGi}{\roleP}{\roleQ}$.
\item If $\gtG=\gtCommSquig{\roleP}{\roleQ}{\gtLab}{}{\gtLab}{\tyGround}{\gtGi}$,
then $\gtMDepth{\gtG}{\roleP}{\roleQ}=1$.
\end{itemize}
Next we show that if $\gtMCount{\gtG}{\roleP}{\roleQ}$ exists and
$(\roleP,\roleQ)\in\gtMRoles{\gtG}$, then
$\gtMCount{\gtG}{\roleP}{\roleQ}>0$.
We proceed by induction on the structure of $\gtG$.
\begin{itemize}
\item If $\gtG=\gtEnd$, then $(\roleP,\roleQ)\not\in\gtMRoles{\gtG}$.
\item If $\gtG=\gtRecVar$, then $(\roleP,\roleQ)\not\in\gtMRoles{\gtG}$.
\item If $\gtG=\gtRec{\gtRecVar}{\gtGi}$ then $(\roleP,\roleQ)\in\gtMRoles{\gtGi}$
and $\gtMCount{\gtGi}{\roleP}{\roleQ}$ exists, so by I.H.
$\gtMCount{\gtGi}{\roleP}{\roleQ}>0$,
so $\gtMCount{\gtG}{\roleP}{\roleQ}=\gtMCount{\gtGi}{\roleP}{\roleQ}>0$.
\item If $\gtG=\gtComm{\rolePi}{\roleQi}{i \in I}{\gtLab[i]}{\tyGround[i]}{\gtG[i]}$, then
$(\roleP,\roleQ)\in\gtMRoles{\gtG[i]}$
and $\gtMCount{\gtG[i]}{\roleP}{\roleQ}$ exists for $i\in I$, so by I.H.
$\gtMCount{\gtG[i]}{\roleP}{\roleQ}>0$ for $i\in I$, so
$\gtMCount{\gtG}{\roleP}{\roleQ}>0$.
\item If $\gtG=\gtCommSquig{\rolePi}{\roleQi}{\gtLab}{}{\gtLab}{\tyGround}{\gtGi}$
and $(\roleP,\roleQ)\neq(\rolePi,\roleQi)$, then
$(\roleP,\roleQ)\in\gtMRoles{\gtGi}$
and $\gtMCount{\gtGi}{\roleP}{\roleQ}$ exists, so by I.H.
$\gtMCount{\gtGi}{\roleP}{\roleQ}>0$, so
$\gtMCount{\gtG}{\roleP}{\roleQ}>0$.
\item If $\gtG=\gtCommSquig{\roleP}{\roleQ}{\gtLab}{}{\gtLab}{\tyGround}{\gtGi}$,
then $\gtMCount{\gtG}{\roleP}{\roleQ}\geq 1$.
\end{itemize}
Now, if $\gtG$ is $\text{balanced}^+$ then,
for $(\roleP,\roleQ)\in\gtMRoles{\gtG}$,
$\gtMCount{\gtG}{\roleP}{\roleQ}$ exists, so
$\gtMCount{\gtG}{\roleP}{\roleQ}>1$, so
$\gtMDepth{\gtG}{\roleP}{\roleQ}$ exists.
\end{proof}

\subsection{Properties of Full Coinductive Merging}
\label{sec:merge-prop}
In order to use the full coinductive projection,
we first study basic properties of the full coinductive merge.
Following from previous work on projections,
we prove an operational
correspondence between a merged type and its components
(Lemma~\ref{lem:good-merge-corr})
and we prove that mergeability respects
set inclusions (Lemma~\ref{lem:mer:setinclusion}).
Additionally,
previous work on projections makes use
of the standard subtyping~(\cref{def:syncsubtyping})
in their operational correspondence results.
The standard subtyping provides an
operational correspondence,
which the precise asynchronous subtyping lacks;
this makes the standard subtyping more useful
in determining typing context properties.
Once we define the standard subtyping,
we shall show that merging produces a subtype (Lemma~\ref{lem:mer:subtype})
and preserves subtyping (Lemma~\ref{lem:mer:prevsubtype}).
\\[2ex]
We will encounter merges
being done iteratively.
Under the view that merges produce witnesses
to compatibility,
surely we can
collapse these repeated applications
to a single one.
This property is
associativity.

\begin{lemma}[Associativity of Merge]
	\label{lem:good-merge-assoc}
	Full coinductive merge is associative \ie
	if $\stMerge{i\in I_j}{\stT[i]}\ni \stTi[j]$ for all $j\in J$, then
	$\stMerge{j\in J, i\in I_j}{\stT[i]}=\stMerge{j\in J}{\stTi[j]}$.
\end{lemma}
\begin{proof}
Consider
\begin{equation}
	\Re \ =
		\stMerge{}{}
		\cup
		\left\{
			\begin{array}{ccc}
				(\mathcal{T},\stT)
				&\suchthat&
				\begin{array}{c}
					\bigcup_{i\in I} \mathcal{T}_i=\mathcal{T}\\
					\forall i\in I, \stMergeRel{\mathcal{T}_i}{\stTi[i]}\\
					\stMergeRel{\{\stTi[i]\suchthat i\in I\}}{\stT}
				\end{array}
			\end{array}
		\right\}
	\cup
		\left\{
			\begin{array}{ccc}
				(\{\stTi[i]\suchthat i\in I\},\stT)
				&\suchthat&
				\begin{array}{c}
					\bigcup_{i\in I}\mathcal{T}_i=\mathcal{T}\\
					\forall i\in I, \stMergeRel{\mathcal{T}_i}{\stTi[i]}\\
					\stMergeRel{\mathcal{T}}{\stT}
				\end{array}
			\end{array}
		\right\}
\end{equation}
This satisfies the coinductive rules defining full merge, and so $\Re\,
=\,\stMergeRel{}{}$.
Suppose that $\{\stT[i]\}_{i\in I} \Re\, \stT$.
If $\stMergeRel{\{\stT[i]\}_{i\in I}}{ \stT}$, then:
\begin{itemize}[leftmargin=0.5in,labelindent=-\leftmargin]
    \item[\inferrule{$\mergeRuleEnd$}]
        If $\unfoldOne{\stT}=\stEnd$, then
        $\unfoldOne{\stT[i]}=\stEnd$ for all $i\in I$
    \item[\inferrule{$\mergeRuleOut$}]
        If $\unfoldOne{\stT}=\stIntSum{\roleQ}{j \in J}{\stChoice{\stLab[j]}{\tyGround[j]} \stSeq \stTi[j]}$, then
        $\unfoldOne{\stT[i]}=\stIntSum{\roleQ}{j \in J}{\stChoice{\stLab[j]}{\tyGround[j]} \stSeq \stT[i,j]}$ for all $i\in I$
        where for all $j\in J$ $\stMergeRel{\left\{\stT[i,j]\suchthat i\in I\right\}}{\stTi[j]}$ so
	$\left\{\stT[i,j]\suchthat i\in I\right\}\Re\,\stTi[j]$
    \item[\inferrule{$\mergeRuleIn$}]
        If $\unfoldOne{\stT} = \stExtSum{\roleP}{j \in J}{\stChoice{\stLab[j]}{\tyGround[j]}\stSeq \stTi[j]}$, then
        $\unfoldOne{\stT[i]}=\stExtSum{\roleP}{j \in J_{i}}{\stChoice{\stLab[j]}{\tyGround[j]} \stSeq \stT[i,j]}$ for all $i\in I$
        where $J=\bigcup_{i\in I}J_{i}$ and for all $j\in J$
        $\stMergeRel{\left\{\stT[i,j]\suchthat j\in J_{i} \right\}}{ \stTi[j]}$ so
	$\left\{\stT[i,j]\suchthat j\in J_{i} \right\}\Re\,\stTi[j]$
\end{itemize}
If $\{I_j\}_{j\in J}$ is a cover of $I$ of non-empty sets and for $j\in J$
there is $\stTi[j]$ with $\stMergeRel{\{\stT[i]\}_{i\in I_j}}{\stTi[j]}$
and $\stMergeRel{\{\stTi[j]\}_{j\in J}}{\stT}$, then:
\begin{itemize}[leftmargin=0.5in,labelindent=-\leftmargin]
    \item[\inferrule{$\mergeRuleEnd$}]
        If $\unfoldOne{\stT}=\stEnd$, then
        $\unfoldOne{\stTi[j]}=\stEnd$ for all $j\in J$, then
	$\unfoldOne{\stT[i]}=\stEnd$ for all $i\in I$.
    \item[\inferrule{$\mergeRuleOut$}]
        If $\unfoldOne{\stT}=\stIntSum{\roleQ}{k \in K}{\stChoice{\stLab[k]}{\tyGround[k]} \stSeq \stTii[k]}$, then
        $\unfoldOne{\stTi[j]}=\stIntSum{\roleQ}{k \in K}{\stChoice{\stLab[k]}{\tyGround[k]} \stSeq \stTii[{j,k}]}$ for all $j\in J$
        where for all $k\in K$ $\stMergeRel{\left\{\stTii[{j,k}]\suchthat j\in J\right\}}{\stTii[k]}$.
	So for all $j\in J$,
	$\unfoldOne{\stT[i]}=\stIntSum{\roleQ}{k \in K}{\stChoice{\stLab[k]}{\tyGround[k]} \stSeq \stTii[{i,k}]}$ for all $i\in I_j$
	where for all $k\in K$ $\stMergeRel{\left\{\stTii[i,k]\suchthat i\in I_j\right\}}{\stTii[j,k]}$.
	Therefore, $\left\{\stTii[i,k]\right\}_{i\in I} \Re\, \stTii[k]$.
    \item[\inferrule{$\mergeRuleIn$}]
        If $\unfoldOne{\stT} = \stExtSum{\roleP}{k \in K}{\stChoice{\stLab[k]}{\tyGround[k]}\stSeq \stTii[k]}$, then
        $\unfoldOne{\stTi[j]}=\stExtSum{\roleP}{k \in K_{j}}{\stChoice{\stLab[k]}{\tyGround[k]} \stSeq \stTii[j,k]}$ for all $j\in J$
        where $K=\bigcup_{j\in J}K_{j}$ and for all $k\in K$
        $\stMergeRel{\left\{\stTii[j,k]\suchthat k\in K_{j} \right\}}{ \stTii[k]}$.
	So for all $j\in J$,
	$\unfoldOne{\stT[i]}=\stExtSum{\roleP}{k \in K_{i}}{\stChoice{\stLab[k]}{\tyGround[k]} \stSeq \stTii[i,k]}$ for all $i\in I_j$
	where $K_j=\bigcup_{i\in I_j}K_i$ and for all $k\in K_j$
	$\stMergeRel{\left\{\stTii[i,k]\suchthat k\in K_{i}, i\in I_j \right\}}{ \stTii[j,k]}$.
	Therefore, $\left\{\stTii[i,k]\suchthat k\in K_i \right\} \Re\, \stTii[k]$.
\end{itemize}
If $\stMergeRel{\{\stTi[j]\}_{j\in J_i}}{ \stT[i]}$ for $i\in I$
and $\stMergeRel{\{\stTi[j]\}_{j\in J_i, i\in I}}{\stT}$, then:
\begin{itemize}[leftmargin=0.5in,labelindent=-\leftmargin]
    \item[\inferrule{$\mergeRuleEnd$}]
        If $\unfoldOne{\stT}=\stEnd$, then
        $\unfoldOne{\stTi[j]}=\stEnd$ for all $j\in J_i$ and $i\in I$.
	So, $\unfoldOne{\stT[i]}=\stEnd$ for all $i\in I$.
    \item[\inferrule{$\mergeRuleOut$}]
        If $\unfoldOne{\stT}=\stIntSum{\roleQ}{k \in K}{\stChoice{\stLab[k]}{\tyGround[k]} \stSeq \stTii[k]}$, then
        $\unfoldOne{\stTi[j]}=\stIntSum{\roleQ}{k \in K}{\stChoice{\stLab[k]}{\tyGround[k]} \stSeq \stTii[j,k]}$ for all $j\in J_i$ and $i\in I$
        where for all $k\in K$ $\stMergeRel{\left\{\stTii[j,k]\suchthat j\in J_i, i\in I\right\}}{\stTii[k]}$.
	So for $i\in I$, $\unfoldOne{\stT[i]}\stIntSum{\roleQ}{k \in K}{\stChoice{\stLab[k]}{\tyGround[k]} \stSeq \stTii[i,k]}$
	where for all $k\in K$ $\stMergeRel{\left\{\stTii[j,k]\suchthat j\in J_i\right\}}{\stTii[i,k]}$.
	Therefore, $\left\{\stTii[j,k]\suchthat j\in J_i, i\in I\right\}\Re\,\stTii[k]$.
    \item[\inferrule{$\mergeRuleIn$}]
        If $\unfoldOne{\stT} = \stExtSum{\roleP}{k \in K}{\stChoice{\stLab[k]}{\tyGround[k]}\stSeq \stTii[k]}$, then
        $\unfoldOne{\stTi[j]}=\stExtSum{\roleP}{k \in K_{j}}{\stChoice{\stLab[k]}{\tyGround[k]} \stSeq \stTii[j,k]}$ for all $j\in J_i$ and $i\in I$
        where $K=\bigcup_{j\in J}K_{j}$ and for all $k\in K$
        $\stMergeRel{\left\{\stTii[j,k]\suchthat k\in K_{j} \right\}}{ \stTii[k]}$.
	So for $i\in I$, $\unfoldOne{\stT[i]} = \stExtSum{\roleP}{k \in K_i}{\stChoice{\stLab[k]}{\tyGround[k]}\stSeq \stTii[i,k]}$
	where $K_i = \bigcup_{j\in J_i}K_{j}$ and for all $k\in K_j$
	$\stMergeRel{\left\{\stTii[j,k]\suchthat k\in K_{j}, j\in J_i \right\}}{ \stTii[i,k]}$.
	Therefore, $\left\{\stTii[i,k]\suchthat k\in K_i\right\}\Re\,\stTii[k]$.
\end{itemize}
\end{proof}
In the literature, the typing context transition relation is sometimes built up from
a local type transition relation, using additional rules for queue manipulation and
inter-participant communication.
We instead define context transitions directly~(\cref{def:mpst-env-reduction}),
which streamlines our proofs by avoiding this layered construction.
A separate transition relation on local types alone remains useful when
reasoning about syntactic properties of local types such as standard subtyping of local types alone without queues.
We therefore introduce local type LTSs
in order to state an operational correspondence for merge~(\cref{lem:good-merge-corr}),
which we use in turn to relate merging to standard
subtyping~(\cref{lem:mer:subtype,lem:mer:prevsubtype}).

\begin{definition}[Local type LTSs]
\label{def:ltype:red-rules}
We define the local type LTSs, denoted by
 $\stT \gtMove[\stEnvAnnotGenericSym] \stTi$ by the following rules:  

\smallskip 

\centerline{\(
\begin{array}{c}
\inference[\iruleLtRec]
{\stT{}[\stRec{\stRecVar}{\stT}/\stRecVar]\,\gtMove[\stEnvAnnotGenericSym]\,\stTi}
{\stRec{\stRecVar}{\stT}\,\gtMove[\stEnvAnnotGenericSym]\,\stTi}
\qquad
\inference[\iruleLtOut]
{j\in I}
{\stIntSum{\roleP}{i \in I}{\stChoice{\stLab[i]}{\tyGround[i]} \stSeq \stT[i]}\,\gtMove[{\roleP\stEnvAnnotOutSym\stChoice{\stLab[j]}{\tyGround[j]}}]\,\stT[j]}
\qquad
\inference[\iruleLtIn]
{j\in I}
{\stExtSum{\roleP}{i \in I}{\stChoice{\stLab[i]}{\tyGround[i]} \stSeq \stT[i]}\,\gtMove[{\roleP\stEnvAnnotInSym\stChoice{\stLab[j]}{\tyGround[j]}}]\,\stT[j]}
\end{array}
\)}
\end{definition}
\begin{lemma}[Merge has an Operational Correspondence]
	\label{lem:good-merge-corr}
	\begin{itemize}
		\item If $\stT[i]\,\gtMove[\stEnvAnnotGenericSym]\,\stTi[i]$
        for all $i\in I$ and $\stMerge{}{\{\stT[i]\suchthat i\in I\}}\ni \stT$, then
        there exists $\stTi$ such that
        $\stT\,\gtMove[\stEnvAnnotGenericSym]\,\stTi$ and
        $\stMerge{}{\{\stTi[i]\suchthat i\in I\}}\ni \stTi$.
		\item If $\stMerge{}{\{\stT[i]\suchthat i\in I\}}\ni \stT$ and
        $\stT\,\gtMove[\stEnvAnnotGenericSym]\,\stTi$, then
        there is some non-empty $J\subseteq I$ such that
        $\stT[i]\,\gtMove[\stEnvAnnotGenericSym]\,\stTi[i]$ for $i\in J$ and
        $\stMerge{}{\{\stTi[i]\suchthat i\in J\}}\ni\stTi$.
	\end{itemize}
\end{lemma}

\begin{proof}
	Suppose that $\stMergeRel{\{\stT[i]\suchthat i\in I\}}{\stT}$.
	\begin{itemize}
\item Suppose that $\stT[i]\,\gtMove[\roleP\stEnvAnnotOutSym\stChoice{\stLab}{\tyGround}]\,\stTi[i]$ for all $i\in I$.
So, for $i\in I$, $\unfoldOne{\stT[i]}=\stIntSum{\roleP}{j \in J_i}{\stChoice{\stLab[i,j]}{\tyGround[i,j]} \stSeq \stTi[i,j]}$ and there is $j\in J$ with $\stLab[i,j]=\stLab$, $\tyGround[i,j]=\tyGround$, and $\stTi[i,j]=\stTi[i]$.
$\stMergeRel{\{\stT[i]\suchthat i\in I\}}{ \stT}$, so we may assume that $J_i=J$, $\stLab[i,j]=\stLab[j]$, and $\tyGround[i,j]=\tyGround[j]$ for all $j\in J$.
Say $k\in J$ has $\stLab[k]=\stLab$, $\tyGround[k]=\tyGround$, and $\stTi[i,k]=\stTi[i]$
We also have that $\unfoldOne{\stT}=\stIntSum{\roleP}{j \in J}{\stChoice{\stLab[j]}{\tyGround[j]} \stSeq \stTii[j]}$, where $\stMergeRel{\{\stTi[i,j]\suchthat i\in I\}}{\stTii[j]}$ for $j \in J$.
So, $\stT\,\gtMove[\roleP\stEnvAnnotOutSym\stChoice{\stLab}{\tyGround}]\,\stTii[k]$ and $\stMergeRel{\{\stTi[i]\suchthat i\in I\}}{ \stTii[k]}$.
The proof in the case that $\stT[i]\,\gtMove[\roleP\stEnvAnnotInSym\stChoice{\stLab}{\tyGround}]\,\stTi[i]$ for all $i\in I$ is similar.
		\item Suppose that $\stT\,\gtMove[\roleP\stEnvAnnotOutSym\stChoice{\stLab}{\tyGround}]\,\stTi$.
So, $\unfoldOne{\stT}=\stIntSum{\roleP}{j \in J}{\stChoice{\stLab[j]}{\tyGround[j]} \stSeq \stTi[j]}$ and there is $k\in J$ such that
$\stLab[k]=\stLab$, $\tyGround[k]=\tyGround$, and $\stTi[k]=\stTi$.
$\stMergeRel{\{\stT[i]\suchthat i\in I\}}{\stT}$ so for $i\in I$, $\unfoldOne{\stT[i]}=\stIntSum{\roleP}{j \in J}{\stChoice{\stLab[j]}{\tyGround[j]} \stSeq \stTi[i,j]}$,
where $\stMergeRel{\{\stTi[i,j]\suchthat i\in I\}}{\stTi[j]}$ for $j\in J$.
So, $\stT[i]\,\gtMove[\roleP\stEnvAnnotOutSym\stChoice{\stLab}{\tyGround}]\,\stTi[i,k]$ for $i\in I$,
$\stMergeRel{\{\stTi[i,k]\suchthat i\in I\}}{\stTi}$, and
$\emptyset\not = I\subseteq I$.
		\item Suppose that $\stT\,\gtMove[\roleP\stEnvAnnotInSym\stChoice{\stLab}{\tyGround}]\,\stTi$.
So, $\unfoldOne{\stT}=\stExtSum{\roleP}{j \in J}{\stChoice{\stLab[j]}{\tyGround[j]} \stSeq \stTi[j]}$ and there is $k\in J$ such that
$\stLab[k]=\stLab$, $\tyGround[k]=\tyGround$, and $\stTi[k]=\stTi$.
$\stMergeRel{\{\stT[i]\suchthat i\in I\}}{\stT}$ so there are $J_j\subseteq I$, non-empty, for $j\in J$ such that
for $i\in I$, $\unfoldOne{\stT[i]}=\stIntSum{\roleP}{j \suchthat i\in J_j}{\stChoice{\stLab[j]}{\tyGround[j]} \stSeq \stTi[i,j]}$,
where $\stMergeRel{\{\stTi[i,j]\suchthat i\in J_j\}}{ \stTi[j]}$ for $j\in J$.
So, $\stT[i]\,\gtMove[\roleP\stEnvAnnotOutSym\stChoice{\stLab}{\tyGround}]\,\stTi[i,k]$ for $i\in J_k$,
$\stMergeRel{\{\stTi[i,k]\suchthat i\in J_k\}}{ \stTi}$, and
$\emptyset\not = J_k\subseteq I$.
	\end{itemize}
\end{proof}

We shall now define the standard subtyping relation $\stSub$~\cite{DBLP:journals/toplas/CarboneHY12,DemangeonH11,DBLP:conf/tlca/MostrousY09,Chen2017},  on local types,
which is precise in the synchronous setting.
This standard subtyping is contained within the precise asynchronous subtyping.
We will use the standard subtyping during our proof of safety, deadlock-freedom, and liveness,
as there is an operational correspondence between
instances of the standard subtyping.
In fact, the operational correspondences from Lemma~\ref{lem:good-merge-corr}
will be sufficient to show that merge behaves well
with regard to the standard subtyping.

\begin{definition}[Subtyping]
\label{def:syncsubtyping}
The \emph{session subtyping relation $\stSub$} is coinductively defined:

\smallskip
\centerline{\(
\begin{array}{c}
\cinference[\iruleStSubIn]{
  \forall i \in I
  &
  \tyGround[i] \tyGroundSub \tyGroundi[i]
  &
  \stT[i] \stSub \tyTi[i]
}{
  \stExtSum{\roleP}{i \in I\cup{}J}{\stChoice{\stLab[i]}{\tyGroundi[i]} \stSeq \stT[i]}\stSub
  \stExtSum{\roleP}{i \in I}{\stChoice{\stLab[i]}{\tyGround[i]} \stSeq \stTi[i]}}
\qquad
\cinference[\iruleStSubRecL]{
  \stT{}[\stRec{\stRecVar}{\stT}/\stRecVar] \stSub \stTi
}{
  \stRec{\stRecVar}{\stT} \stSub \stTi
}
\qquad

\cinference[\iruleStSubRecR]{
  \stT \stSub \stTi{}[\stRec{\stRecVar}{\stTi}/\stRecVar]
}{
  \stT \stSub \stRec{\stRecVar}{\stTi}
}
\\[2ex]
\cinference[\iruleStSubEnd]{}{
  \stEnd \stSub \stEnd
}
\qquad
\cinference[\iruleStSubOut]{
  \forall i \in I
  &
  \tyGround[i] \tyGroundSub \tyGroundi[i]
  &
  \stT[i] \stSub \stTi[i]
}{
  \stIntSum{\roleP}{i \in I}{\stChoice{\stLab[i]}{\tyGround[i]} \stSeq \stT[i]}
  \stSub
  \stIntSum{\roleP}{i \in I \cup J}{\stChoice{\stLab[i]}{\tyGroundi[i]} \stSeq \stTi[i]}
}
\end{array}
\)}
\end{definition}

Note that the definition of $\stSub$ is far simpler than $\asubt$,
and this reflects in the fact that
$\stSub$ is computable
while $\asubt$ is not~\cite{BravettiCZ17,LY2017}.
The intuition behind $\stSub$
is that if you expect a type to
send a message from a set of messages,
then a type sending a message from
a smaller set will be permissible,
and similarly if you
expect a type to be able to
receive a message from a set of messages,
then a type receiving a message
from a larger set will be permissible.
In fact, $\stSub$ is sound and complete
with respect to local type transitions,
if we enforce that subtypes can always send fewer messages,
subtypes can always receive more messages,
and $\stEnd$ types (up-to unfolding)
are only related to other $\stEnd$ types.
It is now important to note that the
precise asynchronous subtyping
does not satisfy this operation correspondence.
\begin{example}[Precise Asynchronous Subtyping does not provide an Operational Correspondence]
\label{ex:sub-not-corr}
Recall the projection of the ring protocol onto $\roleQ$
and its optimisation.
\begin{equation}
\stT[\roleQ]=\stRec{\stRecVar}{
    \stExtSum{\roleP}{}{
      \stChoice{add}{\tyInt} \stSeq
      \stIntSum{\roleR}{}{
        \stChoice{add}{\tyInt} \stSeq \stRecVar,
        \stChoice{sub}{\tyInt} \stSeq \stRecVar
      }
    }
  }
\end{equation}
\begin{equation}
\stTopt[\roleQ] = \stRec{\stRecVar}{
    \stIntSum{\roleR}{}{
      \stChoice{add}{\tyInt} \stSeq \stExtSum{\roleP}{}{\stChoice{add}{\tyInt} \stSeq \stRecVar},
      \stChoice{sub}{\tyInt} \stSeq \stExtSum{\roleP}{}{\stChoice{add}{\tyInt} \stSeq \stRecVar}
    }
  }
\end{equation}
We have that $\stTopt[\roleQ]\asubt\stT[\roleQ]$,
$\stT[\roleQ]\,\gtMove[{\roleP\stEnvAnnotInSym\stChoice{add}{\tyInt}}]$, and
$\stTopt[\roleQ]\,\gtMove[{\roleR\stEnvAnnotOutSym\stChoice{add}{\tyInt}}]$,
but
$\stT[\roleQ]\qquad\not\!\!\!\!\!\!\!\!\!\!\!\!\gtMove[{\roleR\stEnvAnnotOutSym\stChoice{add}{\tyInt}}]$ and
$\stTopt[\roleQ]\qquad\not\!\!\!\!\!\!\!\!\!\!\!\!\gtMove[{\roleP\stEnvAnnotInSym\stChoice{add}{\tyInt}}]$.
Thus, although $\stTopt[\roleQ]\asubt\stT[\roleQ]$, neither type can match the other's immediate transition, so $\asubt$ does not provide the operational correspondence enjoyed by standard subtyping.
\end{example}

\begin{lemma}[Merge is a Subtype]
\label{lem:mer:subtype}  
	If $\stMerge{}{\{\stT[i]\suchthat i\in I\}}\ni \stT$,
    then for all $j\in I$ $\stT \stSub \stT[j]$.
\end{lemma}
\begin{proof}
	Consider $\Re\, =\{(\stT,\stTi)\suchthat \exists \mathcal{T}\ni \stTi \text{, a set of local types, with }\stMergeRel{\mathcal{T}}{\stT} \}$
	This clearly satisfies the laws for subtyping by~\Cref{lem:good-merge-corr}, so $\Re\subseteq\stSub$.
\end{proof}

\begin{lemma}[Merge preserves Subtypes]
\label{lem:mer:prevsubtype}
	If $\stMerge{}{\{\stT[i]\suchthat i\in I\}}\ni \stT$
    and $\stTi\stSub \stT[i]$ for all $i\in I$, then $\stTi \stSub \stT$.
\end{lemma}
\begin{proof}
	Consider $\Re\, =\, \{(\stT,\stTi)\suchthat \exists \mathcal{T}\text{, a set of local types, such that }\stMergeRel{\mathcal{T}}{ \stTi}\text{ and }\stT\stSub\stTii\text{ for all }\stTii\in\mathcal{T}\}$.
	By~\Cref{lem:good-merge-corr}, this clearly satisfies the laws for subtyping, so $\Re\subseteq\stSub$.
\end{proof}

A set of local types is mergeable if they are all `compatible'.
We should expect this `compatibility' property to be decreasing
(with respect to subsets).

\begin{lemma}[Mergeability respects Set Inclusion]
\label{lem:mer:setinclusion}
	If $\{\stT[i]\suchthat i\in I\}$ is mergeable, then $\{\stT[i]\suchthat i\in J\}$ is mergeable for all non-empty $J\subseteq I$.
\end{lemma}
\begin{proof}
	Suppose that $\stMerge{}{\{\stT[i]\suchthat i\in I\}}\ni\stT$. 
	We construct an LTS as follows:
	\begin{enumerate}
		\item Start at node $(\stT, \stT[i] : i\in I)$.
		\item At node $(\stTi, \stTi[i] :i\in I')$, if $\stTi\,\gtMove[\stEnvAnnotGenericSym]\,\stTii$,
 then let $I''=\{i\in I\suchthat \exists \stTii[i],\stTi[i]\,\gtMove[\stEnvAnnotGenericSym]\,\stTii[i]\}$.
Add an edge from $(\stTi, \stTi[i] :i\in I')$ to $(\stTii, \stTii[i] :i\in I'')$ labelled by $\stEnvAnnotGenericSym$.
	\end{enumerate}
	Now, let $J\subseteq I$ be non-empty.
	Remove all nodes, $(\stTi, \stTi[i] :i\in I')$, with $I'\cap J=\emptyset$.
	This graph is a correct type graph, so let $\stU$ be the corresponding type.
	We relabel the nodes.
	\begin{enumerate}
		\item The root node is relabeled $(\stU, \stT[i] : i\in J)$.
		\item If there is an edge from $(\stUi, \stTi[i] : i\in J')$ to $(\stTii, \stTii[i] :i\in I'')$ with $\stEnvAnnotGenericSym$, 
then there is $\stUii$ such that $\stUi\,\gtMove[\stEnvAnnotGenericSym]\,\stUii$.
Relabel $(\stTii, \stTii[i] :i\in I'')$ to $(\stUii, \stTii[i] :i\in I''\cap J)$.
	\end{enumerate}
	The nodes of this graph satisfy the coinductive rules of $\stMergeRel{}{}$,
    so $\stMergeRel{\{\stT[i]\suchthat i\in J\}}{\stU}$.
\end{proof}
\section{Global and Local Type Asynchronous Association}
\label{sec:gtype:relating}
Following the introduction of  
LTS semantics for global types~(\Cref{def:gtype:lts-gt}) 
and typing contexts~(\Cref{def:mpst-env-reduction}), 
we establish a relationship between these two semantics 
using the projection relation $\gtProjRel{}{\roleP}{}$~(\Cref{def:global-proj}) 
and the subtyping relation $\asubt$~(\Cref{def:subtyping}).
By analysing the structural consequences~(\cref{sec:async-sub-sem})
of asynchronous subtyping
and projection,
we derive an operational correspondence
from this relation~(\cref{thm:gtype:proj-comp,thm:gtype:proj-sound}).
Additionally, we use this relation
and the structural consequences of projection
to guarantee safety, deadlock-freedom, and liveness~(\cref{thm:assoc-live}).
\begin{definition}[Association of Global Types and Typing Contexts]
\label{def:assoc}
        A typing context $\stEnv$ is associated with a global type $\gtG$ written $\stEnvAssoc{\stEnv}{\gtG}{a}$, 
        iff
        $\stEnv$ contains projections of $\gtG$:
        $\dom{\stEnv}
        \supseteq \gtRoles{\gtG}$; and
        $\forall \roleP \in \gtRoles{\gtG}
        \,
        \exists
        \stT[\roleP],\quH[\roleP]$
        such that
        $\stEnvApp{\stEnv}{\roleP} = (\quHi[\roleP], \stTi[\roleP])$ and
        $\stTi[\roleP] \asubt \stT[\roleP]$ and
        $\quHi[\roleP]\tySub\quH[\roleP]$ and
        $\gtProjRel{\gtG}{\roleP}{(\quH[\roleP], \stT[\roleP])}$.
\end{definition}

The association $\stEnvAssoc{\gtFmt{\cdot}}{\stFmt{\cdot}}{a}$ is a
binary relation over typing contexts $\stEnv$ and global types $\gtG$. 
There are two requirements for the association: 
\begin{enumerate*}[label={(\arabic*)}]
\item the typing context $\stEnv$ must include an entry for each role; and 
\item
for each role $\roleP$, its corresponding entry in the
typing context~($\stEnvApp{\stEnv}{\roleP}$) must be a subtype~(\Cref{def:subtyping}) \
of the projection of the global type
onto this role.
\end{enumerate*}
Looking again at the ring protocol example~(\cref{sec:ring}), 
we can observe how the transitions of the global type corresponds to updates in the local context. 
This forms an \emph{operational correspondence} between the global semantics and local process configurations. 
Each global step is matched by a change in the local context.
\begin{equation}
\begin{array}{ccccccccccc}
  \gtG[\text{ring}]
  & \gtMove[\ltsSel{\roleP}{\roleQ}{\gtMsgFmt{add}}] &
  \gtGone[\text{ring}]
  & \gtMove[\ltsSel{\roleQ}{\roleR}{\gtMsgFmt{sub}}] &
  \gtGtwo[\text{ring}]
  & \gtMove[\ltsBra{\roleQ}{\roleP}{\gtMsgFmt{add}}] &
  \gtGthree[\text{ring}]
  & \gtMove[\ltsBra{\roleR}{\roleQ}{\gtMsgFmt{sub}}] &
  \gtGfour[\text{ring}]
  & \gtMove[\ltsSel{\roleR}{\roleP}{\gtMsgFmt{sub}}] &
  \gtGfive[\text{ring}]
  \\
  \vertAssoc
  &  &
  \vertAssoc
  &  &
  \vertAssoc
  &  &
  \vertAssoc
  &  &
  \vertAssoc
  &  &
  \vertAssoc
  \\
  \stEnv[0]
  & \stEnvMoveAnnot{\ltsSel{\roleP}{\roleQ}{\gtMsgFmt{add}}} &
  \stEnv[1]
  & \stEnvMoveAnnot{\ltsSel{\roleQ}{\roleR}{\gtMsgFmt{sub}}} &
  \stEnv[2]
  & \stEnvMoveAnnot{\ltsBra{\roleQ}{\roleP}{\gtMsgFmt{add}}} &
  \stEnv[3]
  & \stEnvMoveAnnot{\ltsBra{\roleR}{\roleQ}{\gtMsgFmt{sub}}} &
  \stEnv[4]
  & \stEnvMoveAnnot{\ltsSel{\roleR}{\roleP}{\gtMsgFmt{sub}}} &
  \stEnv[5]
\end{array}
\end{equation}
This idea is illustrated through 
two main theorems: 
\Cref{thm:gtype:proj-sound} shows that the reducibility of a global type aligns with that of 
its associated typing context; while  
\Cref{thm:gtype:proj-comp} illustrates that each possible transition of a typing context is simulated by 
an action in the transitions of the associated global type.
The following two subsections~(\cref{sec:async-sub-sem,sec:comp_sound})
are dedicated to the proofs of these results.

\begin{restatable}[Completeness of Association]{theorem}{thmProjCompleteness}\label{thm:gtype:proj-comp}
  Given 
  associated global type $\gtG$ and typing context $\stEnv$
  such that $\stEnvAssoc{\stEnv}{\gtG}{a}$. 
  If $\stEnv \stEnvMoveGenAnnot \stEnvi$,
  then there exists $\gtGi$ and $\stEnvAnnotGenericSymi$ such that
$\stEnvAnnotGenericSym\preccurlyeq\stEnvAnnotGenericSymi$,
  $\stEnvAssoc{\stEnvi}{\gtGi}{a}$, and
  $\gtG  \,\gtMove[\stEnvAnnotGenericSymi]\, 
    \gtGi$.
\end{restatable}

\begin{restatable}[Soundness of Association]{theorem}{thmProjSoundness}\label{thm:gtype:proj-sound}
    Let $\stEnvAssoc{\stEnv}{\gtG}{a}$ and assume
    $\gtG \,\gtMove[\stEnvAnnotGenericSym]\, \gtGi$.
    Then there exist actions $\stEnvAnnotGenericSymi\preccurlyeq\stEnvAnnotGenericSymii$,
    a context $\stEnvi$, and a global type $\gtGii$ such that
         $\gtG \,\gtMove[\stEnvAnnotGenericSymii]\, \gtGii$,
         $\stEnv \stEnvMoveAnnot{\stEnvAnnotGenericSymi} \stEnvi$, and
         $\stEnvAssoc{\stEnvi}{\gtGii}{a}$.
\end{restatable}
\begin{remark}[\cref{thm:gtype:proj-comp,thm:gtype:proj-sound} from the Perspective of Processes]
In~\cref{the:srtop,the:sftop},
we shall see that~\cref{thm:gtype:proj-comp,thm:gtype:proj-sound}
are analogues for subject reduction
and session fidelity, respectively,
with typing contexts and global types
taking the place of
sessions and typing contexts.
In both cases,~\cref{thm:gtype:proj-comp,thm:gtype:proj-sound}
bridge the gap between typing
sessions with contexts
and typing them with global types.
\end{remark}
\begin{remark}[Sufficiency of Soundness]
\label{rem:sound-suff}
As in \cite{YH2024},
our soundness condition necessitates
that we change our transition label.
However, unlike in the synchronous case,
we can no longer guarantee that the participants remain the same.
In the synchronous setting,
subtyping will allow for fewer selections
in associated contexts than in global types,
meaning that the global transition label
may use a selection label that is absent
locally, forcing us to change to a label that is present.
In the asynchronous setting this can still occur,
but now the asynchronous subtyping will
also prevent us from transitioning on the
same participants locally.
For example, if the global type has
a branching label from $\roleP$ to $\roleQ$
(corresponding to a head transition
and so reflected in the projected context),
then the asynchronous subtyping may have
moved a branching from $\roleR$ to $\roleQ$
to the head of the context's type for $\roleQ$,
preventing $\roleQ$
from performing any action.
This behaviour can be witnessed below in~\cref{eq:soundsuff}.
\begin{equation}\label{eq:soundsuff}
\stEnvAssoc{
\stEnvMap{\roleP}{(\quMsg{\roleQ}{\stLab}{},\stEnd)},
\stEnvMap{\roleQ}{(\quEmpty,\stExtSum{\roleR}{}{\stLabi\stSeq\stExtSum{\roleP}{}{\stLab}})},
\stEnvMap{\roleR}{(\quEmpty,\stIntSum{\roleQ}{}{\stLabi})}
}
{
\gtCommRawSquig{\roleP}{\roleQ}{\gtLab}{\gtLab\gtSeq
	\gtCommRaw{\roleR}{\roleQ}{\gtLabi}
}
}{a}
\end{equation}
\end{remark}
Before we can use the association to
define the top-down typing system
as a bottom-up typing system,
we must ensure that associated contexts are
safe, deadlock-free, and live.

\begin{restatable}[Associated Context Properties]{theorem}{thmProjLive}
\label{thm:assoc-live}
If $\stEnvAssoc{\stEnv}{\gtG}{a}$, then $\stEnv$ is safe, deadlock-free, and live.
\end{restatable}
Section~\ref{sec:live} is dedicated to the proof of this result.

\subsection{Structure of Asynchronous Subtyping}
\label{sec:async-sub-sem}
The issue that we introduce by using asynchronous subtyping,
is that local asynchronous transitions, with participant $\roleP$, need not correspond
to the shallowest occurrences of $\roleP$
in the corresponding global type.
We follow Ghilezan et al.~\cite{Ghilezan2019} to resolve this issue,
by introducing local and global type contexts,
where the holes correspond to
redexes that are reduced during
asynchronous communication.
However, our choice of global type transition has introduced non-determinism,
which we must resolve when identifying local transitions to global transitions.
Our solution is to introduce `forgetful holes', $\forgetHole{\cdot}{i}$,
both locally and globally;
these holes will denote where the asynchronous subtyping
forgets selections,
so that we can identify which branches to forget
when making global transitions.
\begin{remark}[Inductive Contexts over Coinductive Types]
\label{rem:inductive-contexts}
The contexts defined below (\cref{def:global-ctx,def:local-tree-ctx})
are \emph{inductively} defined and do not include a case
for recursive types.
This is because they operate on the \emph{tree representations}
$\ttree{\gtG}$ and $\ttree{\stT}$,
which are coinductive (potentially infinite) trees
obtained by fully unfolding all $\mu$-binders.
Since the tree representations contain no $\mu$-binders,
the context grammars need no corresponding $\mu$ case.
The contexts themselves are always finite,
capturing only the structure between the root of the tree
and the finitely many redex positions,
while the (potentially infinite) subtrees
are used to fill in the holes.
\end{remark}

\begin{definition}[Global Type Contexts]
    \label{def:global-ctx}
    Analogously to Ghilezan et al.~\cite[Def A.19]{Ghilezan2019}, we 
    define global type contexts, with indexed holes, inductively as follows,
where we assume $(\roleFmt{a}, \roleFmt{b}) \neq (\roleP, \roleQ)$, $\roleFmt{c} \ne \roleP$, and $I_R\ne\emptyset$:
    \begin{align*}
        \trCtxGSelpq \ &= \gtComm{\roleFmt{a}}{\roleFmt{b}}{i \in I}{\stLab[i]}{\tyGround[i]}{\trCtxGSelpq} 
        ~\mid~ \gtCommSquig{\roleR}{\roleS}{\stLab}{}{\stLab}{\tyGround}{\trCtxGSelpq}~\mid~ [\cdot]_i\\
	&\mid~\gtComm{\roleP}{\roleFmt{b}}{i \in I_L}{\stLab[i]}{\tyGround[i]}{\forgetHole{\cdot}{l_i}}+
	\gtComm{\roleP}{\roleFmt{b}}{i \in I_R}{\stLab[i]}{\tyGround[i]}{\trCtxGSelpq}\\
        \trCtxGBraqp \ &= \gtComm{\roleFmt{c}}{\roleR}{i \in I}{\stLab[i]}{\tyGround[i]}{\trCtxGBraqp} ~\mid~ \gtCommSquig{\roleFmt{a}}{\roleFmt{c}}{\stLab}{}{\stLab}{\tyGround}{\trCtxGBraqp} ~\mid~ [\cdot]_i
    \end{align*}
    We define separate contexts for situations in which it is safe to reduce transmissions ($\trCtxGSelpq$) and en-route transmissions ($\trCtxGBraqp$) respectively. 
    In either case, given $\mathbb{G}_{\star}^{(\roleP, \roleQ)}$
where $\star \in \{\oplus, \&\}$,
the context has some hole indexing set $I$ and some forgetful hole indexing set $I'$.
We require the indexing sets of continuations to be disjoint
and for $I$ and $I'$ to be disjoint.
We write $\mathbb{G}_{\star}^{(\roleP, \roleQ)}[\gtG_i]_{i\in I}\forgetHole{\gtGi_i}{i\in I'}$ for the tree given by populating
the holes labelled with $i$ by the global type tree $\gtG_i$ for $i\in I$
and the forgetful holes labelled with $i$ by the global type tree $\gtGi_i$ for $i\in I'$.
\end{definition}

Of course, we will need to specify
the shape of a global type after
it performs an asynchronous transition.
Via \RULE{GR-Ctx-I'}, we can do this
by explicitly removing the choices that result
in a forgetful hole.

\begin{definition}[Forgetting Forgetful Holes]
	Given a context $\mathbb{G}$, we define $\overline{\mathbb{G}}$ to be the context
	where we have removed all of the forgetful holes from $\mathbb{G}$ recursively:
	\begin{itemize}
		\item $\overline{\gtComm{\roleP}{\roleQ}{i \in I_L}{\stLab[i]}{\tyGround[i]}{\forgetHole{\cdot}{l_i}}+
		\gtComm{\roleP}{\roleQ}{i \in I_R}{\stLab[i]}{\tyGround[i]}{\mathbb{G}_i}}=
		\gtComm{\roleP}{\roleQ}{i \in I_R}{\stLab[i]}{\tyGround[i]}{\overline{\mathbb{G}_i}}$
		\item $\overline{[\cdot]_i}=[\cdot]_i$
		\item $\overline{\gtCommSquig{\roleP}{\roleQ}{\stLab}{}{\stLab}{\tyGround}{\mathbb{G}}}=
		\gtCommSquig{\roleP}{\roleQ}{\stLab}{}{\stLab}{\tyGround}{\overline{\mathbb{G}}}$
	\end{itemize}
\end{definition}

We now provide similar definitions for local type contexts,
so that we can observe asynchronous semantics
in projected types.

\begin{definition}[Local Type Tree Context]
    \label{def:local-tree-ctx}
    Define a local type tree context inductively as follows ($\roleQ \ne \roleP$):
    \begin{align*}
        \trCtxSelp \ &= \ [\cdot]_i
	\ ~\mid~ \stExtSum{\roleR}{i \in I}{\stChoice{\stLab[i]}{\tyGround[i]} \stSeq \trCtxSelp}
	 ~\mid~ \stIntSum{\roleQ}{i \in I_L}{\stChoice{\stLab[i]}{\tyGround[i]} \stSeq \forgetHole{\cdot}{l_i}} + 
	\stIntSum{\roleQ}{i \in I_R}{\stChoice{\stLab[i]}{\tyGround[i]} \stSeq \trCtxSelp}\\
        \trCtxBrap \ &= \ [\cdot]_i \ ~\mid~ \stExtSum{\roleQ}{i \in I}{\stChoice{\stLab[i]}{\tyGround[i]} \stSeq \trCtxBrap}\\
\end{align*}
We define separate contexts for situations in which it is safe to reduce transmissions ($\trCtxSelp$) and en-route transmissions ($\trCtxBrap$) respectively. 
    In either case, given $\mathbb{L}_{\star}^{\roleP}$
where $\star \in \{\oplus, \&\}$,
the context has some hole indexing set $I$ and some forgetful hole indexing set $I'$.
We require the indexing sets of continuations to be disjoint
and for $I$ and $I'$ to be disjoint.
We write $\mathbb{L}_{\star}^{\roleP, \roleQ}[\stT_i]_{i\in I}\forgetHole{\stTi_i}{i\in I'}$ for the tree given by populating
the holes labelled with $i$ by the local type tree $\stT_i$ for $i\in I$
and the forgetful holes labelled with $i$ by the local type tree $\stTi_i$ for $i\in I'$.
\end{definition}

\begin{figure}[]
\centering
\begin{tikzpicture}[
    scale=0.88,
intchoice/.style={draw, rounded corners=3pt, fill=blue!10, 
        minimum width=1.6cm, minimum height=0.55cm, font=\footnotesize},
    extchoice/.style={draw, rounded corners=3pt, fill=orange!12, 
        minimum width=1.6cm, minimum height=0.55cm, font=\footnotesize},
    normalhole/.style={draw, circle, minimum size=0.75cm, line width=1pt, 
        fill=green!18, draw=green!55!black, font=\footnotesize\bfseries},
    forgethole/.style={draw, circle, minimum size=0.75cm, line width=1pt, 
        fill=red!10, draw=red!45!black, densely dashed, font=\footnotesize\bfseries},
    arr/.style={-stealth, thick, shorten >=1pt, shorten <=1pt},
    lbl/.style={font=\scriptsize, fill=white, inner sep=1pt},
]

\begin{scope}[shift={(-3.5,0)}]
\node[font=\small\bfseries] at (0,3.2) {Before: $\mathbb{L}_\oplus^{\roleQ}$};
    
\node[intchoice] (L-root) at (0,2.2) {$\roleR\oplus\{\cdots\}$};
    
\node[extchoice] (L-ext1) at (-1.5,0.8) {$\roleP\mathbin{\&}\{\cdots\}$};
    \node[normalhole] (L-h1) at (-1.5,-0.6) {$[\cdot]_1$};
    
\node[forgethole] (L-h2) at (1.5,0.8) {$\langle\cdot\rangle_2$};
    
\draw[arr] (L-root) -- node[lbl, above left=-1pt] {\textsf{add}} (L-ext1);
    \draw[arr] (L-root) -- node[lbl, above right=-1pt] {\textsf{sub}} (L-h2);
    \draw[arr] (L-ext1) -- node[lbl, left] {\textsf{add}} (L-h1);
    
\node[font=\tiny, text=green!50!black] at (-1.5,-1.25) {(tracked)};
    \node[font=\tiny, text=red!40!black] at (1.5,0.05) {(forgotten)};
\end{scope}

\draw[-stealth, line width=1.5pt, gray!60] (0.3,1.2) -- (1.5,1.2)
    node[midway, above, font=\scriptsize] {$\overline{\mathbb{L}}$};

\begin{scope}[shift={(3.8,0)}]
\node[font=\small\bfseries] at (0,3.2) {After: $\overline{\mathbb{L}}$};
    
\node[intchoice] (R-root) at (0,2.2) {$\roleR\oplus\{\cdots\}$};
    
\node[extchoice] (R-ext1) at (0,0.8) {$\roleP\mathbin{\&}\{\cdots\}$};
    
\node[normalhole] (R-h1) at (0,-0.6) {$[\cdot]_1$};
    
\draw[arr] (R-root) -- node[lbl, right] {\textsf{add}} (R-ext1);
    \draw[arr] (R-ext1) -- node[lbl, right] {\textsf{add}} (R-h1);
    
\node[font=\tiny\itshape, text=red!50!black, align=center] at (1.3,1.5) 
        {\textsf{sub} branch\\removed};
\end{scope}

\end{tikzpicture}
\caption{Local type context with normal and forgetful holes for $\stTopt[\roleQ]$ from the ring protocol.
\textbf{Left:} The context where $\roleQ$ selects a branch to send to $\roleR$.
The \textsf{add} branch leads to a normal hole $[\cdot]_1$ (tracked continuation),
while \textsf{sub} leads to a forgetful hole $\langle\cdot\rangle_2$ (to be forgotten).
\textbf{Right:} After applying $\overline{\mathbb{L}}$, only the \textsf{add} branch remains.}
\label{fig:holed-context}
\end{figure}

\begin{definition}[Forgetting Forgetful Holes Locally]
	Given a context $\mathbb{L}$, we define $\overline{\mathbb{L}}$ to be the context
	where we have removed all of the forgetful holes from $\mathbb{L}$ recursively:
	\begin{itemize}
		\item $\overline{\stIntSum{\roleQ}{i \in I_L}{\stChoice{\stLab[i]}{\tyGround[i]} \stSeq \forgetHole{\cdot}{l_i}} + 
	\stIntSum{\roleQ}{i \in I_R}{\stChoice{\stLab[i]}{\tyGround[i]} \stSeq \mathbb{L}_i}}=\stIntSum{\roleQ}{i \in I_R}{\stChoice{\stLab[i]}{\tyGround[i]} \stSeq \overline{\mathbb{L}_i}}$
	\item $\overline{\stExtSum{\roleQ}{i \in I}{\stChoice{\stLab[i]}{\tyGround[i]} \stSeq \mathbb{L}_i}}
	=
	\stExtSum{\roleQ}{i \in I}{\stChoice{\stLab[i]}{\tyGround[i]} \stSeq \overline{\mathbb{L}_i}}$
		\item $\overline{[\cdot]_i}=[\cdot]_i$
	\end{itemize}
\end{definition}

Our aim is to use local contexts to
determine redexes for the asynchronous subtyping,
to identify them in global contexts.
Therefore, it will be necessary to
link global contexts with local contexts in some way.
Taking inspiration from our types,
we provide a partial projection function.
Note that as en-route transmissions
induced by a global transition may occur
both
in front of some global redexes and
following some global redexes,
we must keep track of the queue up to
each hole, instead
of requiring uniformity of the queues.

\begin{definition}[Context Projection]
\label{def:context-proj}
     We define a partial function from global type contexts
     with indexed holes to
     queues indexed by the global holes and local type contexts,
     written $\gtProj{\mathbb{G}}{\roleR}=(\{\quH[i]\}_{i\in I},\mathbb{L})$
     with indexes being sets of indexes from the global context, 
     inductively as follows:
     \begin{itemize}
        \item $\gtProj{[\cdot]_i}{\roleR}=(\{i\mapsto \quEmpty\},[\cdot]_{\{i\}})$
	\item $\gtProj{\forgetHole{\cdot}{i}}{\roleR}=(\{i\mapsto \quEmpty\},\forgetHole{\cdot}{\{i\}})$
        \item $\gtProj
        {\gtComm{\roleP}{\roleR}{i \in I}{\stLab[i]}{\tyGround[i]}{\mathbb{G}_i}}
        {\roleR}=(\bigcup_{i\in I} (\gtProj{\mathbb{G}_i}{\roleR})_1,
        \stExtSum{\roleP}
        {i \in I}
        {\stChoice{\stLab[i]}
        {\tyGround[i]} \stSeq 
        (\gtProj{\mathbb{G}_i}{\roleR})_2
        })$
        \item $\gtProj
        {\gtCommSquig{\roleP}{\roleR}{\stLab}{}{\stLab}{\tyGround}{\mathbb{G}'}}
        {\roleR}=((\gtProj{\mathbb{G}'}{\roleR})_1,
        \stExtSum{\roleP}
        {}
        {\stChoice{\stLab}
        {\tyGround} \stSeq 
        (\gtProj{\mathbb{G}'}{\roleR})_2
        })$
        \item $\gtProj
        {\gtComm{\roleR}{\roleQ}{i \in I}{\stLab[i]}{\tyGround[i]}{\mathbb{G}_i}}
        {\roleR}=(\bigcup_{i\in I} (\gtProj{\mathbb{G}_i}{\roleR})_1,
        \stIntSum{\roleQ}
        {i \in I}{\stChoice{\stLab[i]}{\tyGround[i]} \stSeq
        (\gtProj{\mathbb{G}_i}{\roleR})_2
        })$
        \item $\gtProj
        {\gtCommSquig{\roleR}{\roleQ}{\stLab}{}{\stLab}{\tyGround}{\mathbb{G}'}}
        {\roleR}=(\{\quCons{
        \quMsg{\roleQ}{\stLab}{\tyGround}
        }{
        (\gtProj{\mathbb{G}'}{\roleR})_1(j)
        }\}_j,
        \stIntSum{\roleQ}
        {}{\stChoice{\stLab}{\tyGround} \stSeq
        (\gtProj{\mathbb{G}'}{\roleR})_2
        })$
         \item $\gtProj
        {\gtComm{\roleP}{\roleQ}{i \in I}{\stLab[i]}{\tyGround[i]}{\mathbb{G}_i}}
        {\roleR}=( \bigcup_{i\in I} (\gtProj{\mathbb{G}_i}{\roleR})_1,
        \stMerge{i\in I}{(\gtProj{\mathbb{G}_i}{\roleR})_2})$
        \item $\gtProj
        {\gtCommSquig{\roleP}{\roleQ}{\stLab}{}{\stLab}{\tyGround}{\mathbb{G}'}}
        {\roleR}=( (\gtProj{\mathbb{G}'}{\roleR})_1,
        (\gtProj{\mathbb{G}'}{\roleR})_2 )$
     \end{itemize}
     where we define the merge of contexts, where the hole indexes are sets of indexes, by
     \begin{itemize}
        \item $\stMerge{I\in \mathcal{I}}{[\cdot]_I}= [\cdot]_{\{i\suchthat i\in I\in\mathcal{I}\}}$
        \item $\stMerge{I\in \mathcal{I}}{\forgetHole{\cdot}{I}}= \forgetHole{\cdot}{\{i\suchthat i\in I\in\mathcal{I}\}}$
        \item $\stMerge{I\in \mathcal{I}}{
        \stExtSum{\roleP}{j \in J_I}{\stChoice{\stLab[j]}{\tyGround[j]} \stSeq \mathbb{L}_{I,j}}
        }
        =
        \stExtSum{\roleP}{j \in J}{\stChoice{\stLab[j]}{\tyGround[j]} \stSeq
        \stMerge{}{\{\mathbb{L}_{I,j}\suchthat j\in J_I\}}}$
        where $J=\bigcup_{I\in \mathcal{I}} J_I$ and $\{\stLab[j]\}_{j\in J}$ are distinct
        \item $\stMerge{I\in \mathcal{I}}{
        \stIntSum{\roleP}{j \in J}{\stChoice{\stLab[j]}{\tyGround[j]} \stSeq \mathbb{L}_{I,j}}
        }
        =
        \stIntSum{\roleP}{j \in J}{\stChoice{\stLab[j]}{\tyGround[j]} \stSeq 
        \stMerge{I\in \mathcal{I}}{\mathbb{L}_{I,j}}
        }$
     \end{itemize}
The definition of context projections and context merges,
follow the structure of type projections and type merges.
The exception being that, in the context case,
we will not determine the entire queue before reaching a hole
so we must remember the queue upto each hole.
\end{definition}

Now to make use of the context projection,
we observe that it indeed
corresponds to type projections.

\begin{lemma}[Context Merges correspond to Merges]
    \label{lem:contextmerge}
    Suppose that for $i\in I$, $\mathbb{L}_i$
    are local contexts where the hole indexes are
    disjoint sets.
    Say $\mathbb{L}_i$ has hole index set $J_i$ and
    forgetful hole index set $J'_i$.
    Suppose that $\mathbb{L}=\stMerge{i\in I}{\mathbb{L}_i}$,
    with hole indexes $J$ and forgetful hole indexes $J'$.
    Suppose that $\ttree{\stT}=\mathbb{L}
    [\ttree{\stT[j]}]_{j\in J}
    \forgetHole{\ttree{\stT[j]}}{j\in J'}$,
    for all $i\in I$
    $\ttree{\stTi[i]}=\mathbb{L}_i
    [\ttree{\stTi[i,j]}]_{j\in J_i}
    \forgetHole{\ttree{\stTi[i,j]}}{j\in J'_i}$, and
    for all $j\in J$
    $\stMerge{}{\{\stTi[i,j']\suchthat i\in I,j'\in J_i\cup J'_i,j'\subseteq j\}}\ni \stT[j]$.
    Then $\stMerge{i\in I}{\stTi[i]}\ni\stT$.
\end{lemma}

\begin{proof}
By induction on local contexts.
\end{proof}

\begin{lemma}[Context Projections correspond to Projections]
    Suppose that $\gtProj{\mathbb{G}}{\roleR}=(\{\quH[i]\}_{i\in I\cup J},\mathbb{L})$
    with index sets $I\cup J$ and $I'\cup J'$, respectively.
    Suppose that $\ttree{\gtG}=\mathbb{G}[\gtGi[i]]_{i\in I}\forgetHole{\gtGi[i]}{i\in J}$,
    $\ttree{\stT}=\mathbb{L}[\stTi[i]]_{i\in I'}\forgetHole{\stTi[i]}{i\in J'}$,
    $\gtGi[i]=\ttree{\gtG[i]}$ for $i\in I\cup J$,
    $\stTi[i]=\ttree{\stT[i]}$ for $i\in I'\cup J'$, and
    for $i\in I\cup J$ there is $\stTii[i]$ such that
    $\gtProjRel{\gtG_i}{\roleR}{(\quHi[i],\stTii[i])}$, and
    $\stMerge{i\in J}{\stTii[i]}\ni \stT[J]$ for $J\in I'$, and
    $\quH=\quCons{\quH[i]}{\quHi[i]}$ for $i\in I$.
    Then,
    $\gtProjRel{\gtG}{\roleR}{(\quH, \stT)}$.
\end{lemma}

\begin{proof}
    We proceed by induction on global contexts, using~\cref{lem:contextmerge}
\end{proof}

With all of our tools in place,
we can now see how local contexts
and global contexts
can be used to capture
redexes up to
asynchronous subtyping.
Asynchronous subtyping allows redexes to move,
thus the top-level redex of an asynchronous subtype
appears somewhere in its supertype.

\begin{lemma}[Inversion of subtyping]
    \label{lem:inv-subtyping}
    \ \\
    \begin{itemize}
    \item[(1)] If $\stIntSum{\roleP}{i \in I}{\stChoice{\stLab[i]}{\tyGround[i]} \stSeq \stT[i]} \asubt \stT$,
    then  $\ttree{\stT} = \trCtxSelp[\stIntSum{\roleP}{i \in I_j}{\stChoice{\stLab[i]}{\tyGround[i,j]} \stSeq \trTi[i,j]}]_{j\in J}\forgetHole{\trTii[j]}{j\in J'}$
    with $I \subseteq I_j$ for all $j\in J$ and for all $i \in I$ and $j\in J$,
    we have $\tyGround[i] \tyGroundSub \tyGround[i,j]$,
    and if $\ttree{\stTi[i]}=\trCtxSelp[\trTi[i,j]]_{j\in J}\forgetHole{\trTii[j]}{j\in J'}$
    then $\stT[i] \asubt \stTi[i] $.
    \item[(2)] If $\stExtSum{\roleP}{i \in I}{\stChoice{\stLab[i]}{\tyGround[i]} \stSeq \stT[i]} \asubt \stT$,
    then   $\ttree{\stT} = \trCtxBrap[\stExtSum{\roleP}{i \in I_j}{\stChoice{\stLab[i]}{\tyGround[i,j]} \stSeq \trTi[i,j]}]_{j\in J}$
    with $I_j \subseteq I$ for all $j\in J$,
    $I=\bigcup_{j\in J}I$,
    for all $j\in J$ and $i\in I$ $\tyGround[i,j] \tyGroundSub \tyGround[i]$, and
    if $\ttree{\stTi[i]}=\trCtxBrap[\trTi[i,j]]_{j\in J}$ for some $i\in\bigcap_{j\in J}I_j$ then
    $\stT[i] \asubt \stTi[i]$.
    \end{itemize}
\end{lemma}
\begin{proof}
  We only show item~(1) (as item~(2) is dual).
  Write $\stS=\stIntSum{\roleP}{i \in I}{\stChoice{\stLab[i]}{\tyGround[i]} \stSeq \stT[i]}$ and assume $\stS \asubt \stT$.

  By \cref{def:subtyping}, for every $\stU\in\singleOut{\ttree{\stS}}$ and every $\stVi\in\singleIn{\ttree{\stT}}$ there exist $\stW\in\singleIn{\stU}$ and $\stWi\in\singleOut{\stVi}$ with $\stW\subttt\stWi$.
  The clause for $\singleOut{\cdot}$ on $\stS$ yields, for each $i\in I$, an element
  \[
    \stU[i]=\stSend{\roleP}{\stLab[i]}{\tyGround[i]}{\stUi[i]}
    \quad\text{with}\quad
    \stUi[i]\in\singleOut{\ttree{\stT[i]}}.
  \]
  Fix $i\in I$ and an arbitrary $\stVi\in\singleIn{\ttree{\stT}}$. The subtyping premise gives SISO trees $\stW_i\in\singleIn{\stU_i}$ and $\stWi_i\in\singleOut{\stVi}$ with $\stW_i\subttt\stWi_i$.
  Unfolding $\singleIn{\cdot}$ on $\stU[i]$ forces $\stW[i]=\stSend{\roleP}{\stLab[i]}{\tyGround[i]}{\stWi[i]}$ for some $\stWi[i]\in\singleIn{\stUi[i]}$. 
  Invert the last rule of the refinement derivation of $\stW[i] \subttt \stWi[i]$.
  Because $\stW[i]$ is headed by an output, the final rule is either \RULE{Ref-Out} or \RULE{Ref-$\mathcal{B}$}.
  \begin{itemize}
    \item \RULE{Ref-Out}: $\stWi[i] = \stSend{\roleP}{\stLab[i]}{\tyGround[i,j]}{\stWi'_i}$ with $\tyGround[i]\tyGroundSub\tyGround[i,j]$ and premise $\stWi[i] \subttt \stWii[i]$. 
    Then we must have that $\stT = \stIntSum{\roleP}{i \in I'}{\stChoice{\stLab[i]}{\tyGround[i]} \stSeq \stT[i]}$ with $I \subseteq I'$.
    \item \RULE{Ref-$\mathcal{B}$}: $\stWi[i] = \Bp ; \stSend{\roleP}{\stLab[i]}{\tyGround[i,j]}{\stWi'_i}$ for some non-empty finite $\Bp$, with $\tyGround[i]\tyGroundSub\tyGround[i,j]$, $\act{\stWi[i]}=\act{\Bp ; \stWii[i]}$, and premise $\stWi[i] \subttt \Bp ; \stWii[i]$.
  \end{itemize}
  For \RULE{Ref-$\mathcal{B}$}, write $\Bp=a;\Bp_{\mathrm{tail}}$ with $a$ the first action. We reason by case analysis on $a$, using that $\stWi[i]\in\singleOut{\stVi}$ and the definitions of $\singleOut{\cdot}$ (where $\stSendOne{\roleR}{\stLab}{\tyGround}$ and $\stRecvOne{\roleR}{\stLab}{\tyGround}$ are the single-branch forms of internal and external choice, respectively):
  \begin{itemize}
    \item If $a=\stRecvOne{\roleR}{\stLab'}{\tyGround'}$, then $\singleOut{\stVi}$ can contain such a prefix only if $\stT=\stIntSum{\roleR}{k\in K}{\stChoice{\stLab[k]}{\tyGround[k]} \stSeq \stVi[k]}$ with some $k$ satisfying $\stLab[k]=\stLab'$ and $\tyGround'\tyGroundSub\tyGround[k]$; moreover $\stWi[i]\in\singleOut{\stVi[k]}$.
    \item If $a=\stSendOne{\roleQ}{\stLab'}{\tyGround'}$, with $\roleP \ne \roleQ$, and then $\singleOut{\stVi}$ can contain such a prefix only if $\stT=\stExtSum{\roleQ}{k\in K}{\stChoice{\stLab[k]}{\tyGround[k]} \stSeq \stVi[k]}$ with some $k$ satisfying $\stLab[k]=\stLab'$ and $\tyGround[k]\tyGroundSub\tyGround'$, and again $\stWi[i]\in\singleOut{\stVi[k]}$.
  \end{itemize}
  In either subcase we strip the head constructor of $\ttree{\stT}$, using forgetful holes to cover any branches not in $I$, 
  obtaining a continuation, yielding a shorter prefix $\Bp_{\mathrm{tail}}$. The induction hypothesis on $\Bp_{\mathrm{tail}}$ yields a decomposition in which every branch reached after $\Bp_{\mathrm{tail}}$ still contains an internal choice by $\roleP$ whose branch-set includes all labels in $I$ with payloads extending the corresponding $\tyGround[i]$. 
  Re-wrapping the stripped constructor $a$ preserves this property for the continuation, showing that along every branch of $\ttree{\stT}$ compatible with $\Bp$ there is an occurrence of $\stIntSum{\roleP}{i\in I_j}{\stChoice{\stLab[i]}{\tyGround[i,j]}\stSeq\trTi[i,j]}$ with $I\subseteq I_j$ and $\tyGround[i]\tyGroundSub\tyGround[i,j]$.
  So the overall shape of $\ttree{\stT}$ is $\trCtxSelp[\stIntSum{\roleP}{i \in I_j}{\stChoice{\stLab[i]}{\tyGround[i,j]} \stSeq \trTi[i,j]}]_{j\in J}\forgetHole{\trTii[j]}{j\in J'}$ as claimed.
\end{proof}

Recall the context rules
\inferrule{\iruleGtMoveCtx},
\inferrule{\iruleGtMoveCtx'} and
\inferrule{\iruleGtMoveCtxII} from
\cref{def:gtype:lts-gt};
these exactly say that
sends occur under $\trCtxGSelpq$ contexts
and receives occur under $\trCtxGBraqp$ contexts.

\begin{lemma}[Global Type Tree Context Transitions]
    \label{lem:global-ctx-red}
    \ \\
    \begin{enumerate}
        \item Let $\trCtxGSelpq$ be a context with index set $J$ and
        forgetful index set $J'$.
        Let $k\in \bigcap_{j\in J} I_j$.
        We have that
        $\trCtxGSelpq[\gtComm{\roleP}{\roleQ}{i \in I_j}{\stLab[i]}{\tyGround[i]}{\gtG[i,j]}]_{j\in J}\forgetHole{\gtGi[j]}{j\in J'}
        \gtMove[\ltsSel{\roleP}{\roleQ}{\gtLab[k](\tyGround[k])}]
        \overline{\trCtxGSelpq}[{\gtG[k,j]}]_{j\in J}$.
        \item Let $\trCtxGBraqp$ be a context with index set $J$.
        We have that
        $\trCtxGBraqp[\gtCommSquig{\roleQ}{\roleP}{\stLab}{}{\stLab}{\tyGround}{\gtG[j]}]_{j\in J}
        \gtMove[\ltsBra{\roleP}{\roleQ}{\gtLab(\tyGround)}]
        \trCtxGSelpq[{\gtG[j]}]_{j\in J}$
    \end{enumerate}
\end{lemma}
\begin{proof}
  The proof in each case is by induction on the structure of the context. 
  \begin{enumerate}
    \item By induction on the structure of $\trCtxGSelpq$. 
  \begin{itemize}
    \item \textbf{Case $\trCtxGSelpq = [\cdot]$:} Immediate from \RULE{GR-$\oplus$}. 
    \item \textbf{Case $\trCtxGSelpq = \gtComm{\roleFmt{a}}{\roleFmt{b}}{i \in I}{\stLabi[i]}{\tyGroundi[i]}{\trCtxGSelpq} $:} Applying \RULE{GR-Ctx-I} to the induction hypothesis. 
\item \textbf{Case $\trCtxGSelpq = \gtComm{\roleP}{\roleFmt{b}}{i \in I_L}{\stLab[i]}{\tyGround[i]}{\forgetHole{\cdot}{l_i}}+
	\gtComm{\roleP}{\roleFmt{b}}{i \in I_R}{\stLab[i]}{\tyGround[i]}{\trCtxGSelpq} $:}Applying \RULE{GR-Ctx-I'} to the induction hypothesis. 
    \item \textbf{Case $\trCtxGSelpq = \gtCommSquig{\roleR}{\roleS}{\gtLab[k]}{i \in I}{\stLabi[i]}{\tyGroundi[i]}{\trCtxGSelpq} $:} Applying \RULE{GR-Ctx-II} to the induction hypothesis. 
  \end{itemize}
  \item By induction on the structure of $\trCtxGBraqp$. 
  \begin{itemize}
    \item \textbf{Case $\trCtxGBraqp = [\cdot]$:} Immediate from \RULE{GR-$\&$}. 
    \item \textbf{Case $\trCtxGBraqp = \gtComm{\roleFmt{c}}{\roleR}{i \in I}{\stLabi[i]}{\tyGroundi[i]}{\trCtxGSelpq} $:} Applying \RULE{GR-Ctx-I} to the induction hypothesis. 
    \item \textbf{Case $\trCtxGBraqp = \gtCommSquig{\roleFmt{a}}{\roleFmt{b}}{\gtLabi[k]}{i \in I}{\stLabi[i]}{\tyGroundi[i]}{\trCtxGBraqp} $:} Applying \RULE{GR-Ctx-II} to the induction hypothesis. 
  \end{itemize}
\end{enumerate}
\end{proof}

\subsection{Completeness and Soundness}
\label{sec:comp_sound}
In~\cref{sec:async-sub-sem},
we related local asynchronous semantics
to global semantics,
in a manner that can
be reflected by our projections,
via type contexts,
and in~\cref{sec:bal},
we determined properties
that being $\text{balanced}^+$ ensures.
We can now combine these results;
we may now proceed by induction on
the depth functions provided by $\text{balanced}^+$,
to invert projections \ie
produce global contexts
from local contexts.
\\[2ex]
We begin by observing that
the semantics of a typing context
informs us of
the structure of its local types
and queues.

\begin{lemma}[Inversion of Context Semantics]
  \label{lem:inv-ctx-lts}
  Suppose that $\stEnv\,\stEnvMoveGenAnnot\,\stEnvi$.
  \begin{enumerate}
    \item If $\stEnvAnnotGenericSym = \ltsSel{\roleP}{\roleQ}{\stLab(\tyGround)}$
    then $\stEnvApp{\stEnv}{\roleP} = (\quH[\roleP], \stT[\roleP])$,
    with $\stT[\roleP] = \stIntSum{\roleQ}{i\in I}{\stChoice{\stLab[i]}{\tyGround[i]}\stSeq\stT[i]}$ 
    where $\stLab[j]=\stLab,\,\tyGround[j]=\tyGround$ for some $j \in I$. 
    And $\stEnvApp{\stEnvi}{\roleP} = (\quCons{\quH[\roleP]}{\quMsg{\roleQ}{\stLab}{\tyGround}}, \stT[j])$ with $\stEnvApp{\stEnvi}{\roleR} = \stEnvApp{\stEnv}{\roleR}$ 
    for all $\roleR \in \dom{\stEnv}$ with $\roleR \ne \roleP$.
    \item If $\stEnvAnnotGenericSym = \ltsBra{\roleP}{\roleQ}{\stLab(\tyGround)}$
    then $\stEnvApp{\stEnv}{\roleP} = (\quH[\roleP], \stT[\roleP])$,
    with $\stT[\roleP] = \stExtSum{\roleQ}{i\in I}{\stChoice{\stLab[i]}{\tyGround[i]}\stSeq\stT[i]}$ 
    where $\stLab[j]=\stLab,\,\tyGround[j]=\tyGround$ for some $j \in I$. 
    And $\stEnvApp{\stEnv}{\roleQ} = (\quCons{\quH[\roleQ]}{\quMsg{\roleP}{\stLab}{\tyGroundi}}, \stT[\roleQ])$
    with $\tyGroundi \tyGroundSub \tyGround$. 
    Finally, $\stEnvApp{\stEnvi}{\roleP} = (\quH[\roleP], \stT[j])$ and 
    $\stEnvApp{\stEnvi}{\roleQ} = (\quH[\roleQ], \stT[\roleQ])$ 
    with $\stEnvApp{\stEnvi}{\roleR} = \stEnvApp{\stEnv}{\roleR}$ 
    for all $\roleR \in \dom{\stEnv}$ with $\roleR \notin \{\roleP, \roleQ\}$.
  \end{enumerate}
  \end{lemma}
  \begin{proof}
    By rule-induction on
    $\stEnv \stEnvMoveGenAnnot \stEnvi$
    (see \Cref{def:mpst-env-reduction}).
  \end{proof}

As alluded to above,
we can proceed by induction on our depth functions
to match local sends and receives in projected types
to communications in global types.

\begin{lemma}[Head Inversion of Projection]
\label{lem:head-inv-proj}
	Assume $\gtProjRel{\gtG}{\roleP}{(\quH, \stT)}$.
	\begin{enumerate}
		\item If $\stT=\stIntSum{\roleQ}{i \in I}{\stLab[i]:\stT[i]}$, then
		$\ttree{\gtG}=\mathbb{G}[
		\gtComm{\roleP}{\roleQ}
	        {i \in I}
	        {\stLab[i]}
	        {\tyGround[i]}
	        {\gtG[i,j]}
		]_{j\in J}$ where
		$\gtProjRel{\gtGi[i,j]}{\roleP}{(\quH[j],\stT[i,j])}$,
		$\stMerge{}{\{\stT[i,j]\suchthat j\in J\}}\ni \stT[i]$,
		$\ttree{\gtGi[i,j]}=\gtG[i,j]$,
		$\gtProj{\mathbb{G}}{\roleP}=(\{\quHi[j]\}_{j\in J}, [\cdot]_J)$,
		for $j\in J$, $\quH=\quCons{\quHi[j]}{\quH[j]}$, and
		$\mathbb{G}$ does not contain $\roleP$.
		\item If $\stT=\stExtSum{\roleQ}{i \in I}{\stLab[i]:\stT[i]}$ and
		$\gtProjRel{\gtG}{\roleQ}{(\quCons{\quMsg{\roleP}{\stLab}{\tyGround}}{\quHii}, \stTii)}$, then
		$\ttree{\gtG}=\mathbb{G}[
		\gtCommSquig{\roleQ}{\roleP}{\stLab}
	        {}
	        {\stLab}
	        {\tyGround}
	        {\gtG[j]}
		]_{j\in J}$ and
		$\stT=\stExtSum{\roleQ}{}{\stLab:\stTi}$
		where
		$\gtProjRel{\gtGi[j]}{\roleP}{(\quH[j],\stT[j])}$ for all $j\in J$,
		$\stMerge{}{\{\stT[j]\suchthat j\in J\}}\ni \stTi$,
		$\ttree{\gtGi[j]}=\gtG[j]$,
		$\gtProj{\mathbb{G}}{\roleP}=(\{\quHi[j]\}_{j\in J}, [\cdot]_J)$
		for $j\in J$, $\quH=\quCons{\quHi[j]}{\quH[j]}$, and
		$\mathbb{G}$ does not contain $\roleP$ except as the sender of an en-route transmission.
		\item If $\stT=\stExtSum{\roleQ}{i \in I}{\stLab[i]:\stT[i]}$ and
		$\gtProjRel{\gtG}{\roleQ}{(\quH[\roleQ], \stTii)}$ with no element with destination $\roleP$, then\\
		$\ttree{\gtG}=\mathbb{G}[
		\gtComm{\roleQ}{\roleP}
	        {i \in I_j}
	        {\stLab[i]}
	        {\tyGround[i]}
	        {\gtG[i,j]}
		]_{j\in J}$ where
		$I_j \subseteq I$,
		$I=\bigcup_{j\in J}I_j$,
		$\gtProjRel{\gtGi[i,j]}{\roleP}{(\quH[j],\stT[i,j])}$ for all $j\in J$ and $i\in I_j$,
		$\stMerge{}{\{\stT[i,j]\suchthat i\in I_j\}}\ni \stT[i]$ for all $i\in I$,
		$\ttree{\gtGi[i,j]}=\gtG[i,j]$,
		$\gtProj{\mathbb{G}}{\roleP}=(\{\quHi[j]\}_{j\in J}, [\cdot]_J)$
		for $j\in J$, $\quH=\quCons{\quHi[j]}{\quH[j]}$, and
		$\mathbb{G}$ does not contain $\roleP$.
	\end{enumerate}
\end{lemma}

\begin{proof}
    As $\unfoldOne{\stT}\ne\stEnd$,
    $\roleP\in\gtRoles{\gtG}$.
    As $\gtG$ is balanced, $\gtDepth{\gtG}{\roleP}$ exists.
    We shall proceed by induction on $\gtDepth{\gtG}{\roleP}$.
    Suppose that $\gtDepth{\gtG}{\roleP}=1$,
    so $\roleP$ appears at the head of $\gtG$.
	\begin{enumerate}
        \item
        If $\stT=\stIntSum{\roleQ}{i \in I}{\stLab[i]\stSeq\stT[i]}$, then
        by the definition of the projection,
        $\unfoldOne{\gtG}=\gtComm{\roleP}{\roleQ}
	        {i \in I}
	        {\stLab[i]}
	        {\tyGround[i]}
	        {\gtG[i]}$
        where $\gtProjRel{\gtG[i]}{\roleP}{(\quH,\stT[i])}$.
        $\gtProj{[\cdot]_l}{\roleP}=(\{l\mapsto \quEmpty\},[\cdot]_{\{l\}})$.
        \item 
        If $\stT=\stExtSum{\roleQ}{i \in I}{\stLab[i]\stSeq\stT[i]}$ and
		$\gtProjRel{\gtG}{\roleQ}{(\quH[\roleQ], \stTii)}$
        with $\quH[\roleQ]$ having no element with destination $\roleP$,
        then by definition of the projection,
        $\unfoldOne{\gtG}=\gtComm{\roleQ}{\roleP}
	        {i \in I}
	        {\stLab[i]}
	        {\tyGround[i]}
	        {\gtG[i]}$
        where $\gtProjRel{\gtG[i]}{\roleP}{(\quH,\stT[i])}$.
        \item
        If $\stT=\stExtSum{\roleQ}{i \in I}{\stLab[i]\stSeq\stT[i]}$ and
		$\gtProjRel{\gtG}{\roleQ}{(\quCons{\quMsg{\roleP}{\stLab}{\tyGround}}{\quHii}, \stTii)}$,
        then by definition of the projection
        and as $\gtMCount{\gtG}{\roleQ}{\roleP}$ exists,
        we have that
        $\unfoldOne{\gtG}=\gtCommSquig{\roleQ}{\roleP}{\stLab}
	        {}
	        {\stLab}
	        {\tyGround}
	        {\gtGi}$,
        and
        $\stT=\stExtSum{\roleQ}{}{\stLab:\stTi}$,
        where $\gtProjRel{\gtGi}{\roleP}{(\quH,\stTi)}$.
	\end{enumerate}
    Suppose that $\gtDepth{\gtG}{\roleP}>1$, so
    $\roleP$ is not at the head of $\gtG$.
    We consider the cases for the head of $\gtG$.
    \begin{enumerate}
        \item $\unfoldOne{\gtG}=\gtComm{\rolePi}{\roleQi}
	    {i \in I}
	    {\stLab[i]}
	    {\tyGround[i]}
	    {\gtG[i]}$ where $\roleP\not\in \{\rolePi,\roleQi\}$.
        For $i\in I$, we have $\stT[i]$ such that
        $\gtProjRel{\gtG[i]}{\roleP}{(\quH,\stT[i])}$
        and $\stMerge{i\in I}{\stT[i]}\ni\stT$, and
        if $\gtProjRel{\gtG}{\roleQ}{(\quHi,\stTi)}$ then
        $\gtProjRel{\gtG[i]}{\roleQ}{(\quHi,\star)}$.
        We then apply the I.H.,
        reindex the contexts, and
        use associativity of the full coinductive merge to conclude.
        \item $\unfoldOne{\gtG}=\gtCommSquig{\rolePi}{\roleQi}{\gtLab}
	    {}
	    {\stLab}
	    {\tyGround}
	    {\gtGi}$ where $\roleP\not\in \{\rolePi,\roleQi\}$.
        We have
        $\gtProjRel{\gtGi}{\roleP}{(\quH,\stT)}$, and
        if $\gtProjRel{\gtG}{\roleQ}{(\quHi,\stTi)}$ then
        $\gtProjRel{\gtGi}{\roleQ}{(\quHii,\star)}$
        where $\quHii$ and $\quHi$, restricted to
        $\roleP$, are the same.
        We then apply the I.H. to conclude.
        \item $\unfoldOne{\gtG}=\gtCommSquig{\roleP}{\roleQi}{\gtLab}
	    {}
	    {\stLab}
	    {\tyGround}
	    {\gtGi}$.
        We have $\quHiii$ such that
        $\gtProjRel{\gtGi}{\roleP}{(\quHiii, \stT)}$ and
        $\gtProjRel{\gtGi}{\roleQ}{(\quCons{\quMsg{\roleP}{\stLab}{\tyGround}}{\quHii},\stTii)}$.
        We then apply the I.H. to conclude.
    \end{enumerate}
\end{proof}

Since we are working with asynchronous subtyping,
the head of a local type is insufficient.
However by~\cref{lem:inv-subtyping},
it will suffice to consider local asynchronous redexes
in their respective local contexts.
Further, it suffices to apply~\cref{lem:head-inv-proj}
at each step to extend it to local contexts.

  \begin{lemma}[Inversion of Projection]
    \label{lem:inv-proj}
    Assume $\gtProjRel{\gtG}{\roleP}{(\quH, \stT)}$ 
    \begin{enumerate}
        \item If $\ttree{\stT} =
        \trCtxSelq[
        \stIntSum{\roleQ}{i \in I_j}{\stLab[i,j]\stSeq
        \stT[i,j]}]_{j\in J}\forgetHole{\stT[j]}{j\in K}$
        then
	\newline
        $\ttree{\gtG} =
        \trCtxGSelpq[
        \gtComm{\roleP}{\roleQ}
        {i \in I_{f(j)}}
        {\stLab[i,f(j)]}
        {\tyGround[i,f(j)]}
        {\gtG[i,j]}]_{j\in J'}\forgetHole{\gtGiii[j]}{j\in K'}$
        where $\sigma$ is a reindexing of $\trCtxSelq$,
        $\gtProj{\trCtxGSelpq}{\roleP}=(\{\quHi[j]\}_{j\in J'\cup K'},
        \sigma \trCtxSelq)$, and
        $f(j)=k\in J\cup K\suchthat j\in\sigma k$ for $j\in J\cup K$.
        We additionally have that
        $\gtProjRel{\gtGi[i,j]}{\roleR}{(\quHii[j],\stTi[i,j])}$ for $j\in J', i\in I_{f(j)}$,
	$\gtProjRel{\gtGi[j]}{\roleR}{(\quHii[j],\stTi[j])}$ for $j\in K'$,
        $\stTii[i,k]\in\stMerge{j\in \sigma k}{\stTi[i, j]}$ for $k\in J, i\in I_k$, and
	$\stTii[k]\in\stMerge{j\in \sigma k}{\stTi[j]}$ for $k\in K$,
        where $\ttree{\stTii[i,k]}=\stT[i,k]$ for $k\in J$ and $i\in I_k$,
	$\ttree{\stTii[k]}=\stT[k]$ for $k\in K$,
        $\ttree{\gtGi[i,j]}=\gtG[i,j]$ for $j\in J'$ and $i\in I_{f(j)}$,
	$\ttree{\gtGi[j]}=\gtG[j]$ for $j\in K'$, and
        $\quH=\quCons{\quHi[j]}{\quHii[j]}$ for $j\in J'\cup K'$.
        
        \item If $\ttree{\stT} =
        \trCtxBraq[
        \stExtSum{\roleQ}{i \in I_j}{\stLab[i,j]\stSeq
        \stT[i,j]}]_{j\in J}$
        and $\gtProjRel{\gtG}{\roleQ}{
        (\quCons{\quMsg{\roleP}{\stLab}{\tyGround}}{\quH[\roleQ]},\stTi[\roleQ])
        }$,
        then for all $j\in J$,
	$|I_j|=1$ so relabel $\stT[i,j]$ to $\stT[j]$, also $\stLab[i,j]=\stLab$ and $\tyGround[i,j]=\tyGround$ for the unique $i\in I$,
        $\ttree{\gtG} =
        \trCtxGBraqp[
        \gtCommSquig{\roleQ}{\roleP}
        {\gtLab}
        {}
        {\stLab}
        {\tyGround}
        {\gtG[j]}]_{j\in J'}$
        where $\sigma$ is a reindexing of $\trCtxBraq$,
        $\gtProj{\trCtxGBraqp}{\roleP}=(\{\quHi[j]\}_{j\in J},\sigma \trCtxBraq)$.
        We additionally have that
        $\gtProjRel{\gtGi[j]}{\roleR}{(\quHii[j], \stTi[j])}$ for $j\in J'$ and
        $\stTii[k]\in\stMerge{j\in \sigma k}{\stTi[j]}$ for $k\in J$,
        where $\ttree{\stTii[k]}=\stT[k]$,
        $\ttree{\gtGi[j]}=\gtG[j]$, and
        $\quH=\quCons{\quHi[j]}{\quHii[j]}$ for $j\in J'$.
        \end{enumerate}
  \end{lemma}
\begin{proof}
	Induction on local contexts.
	The base case and inductive cases
	are instances of~\cref{lem:head-inv-proj}.
\end{proof}

We will now be able to perform global transitions
upon observing local asynchronous transitions,
such that the subject of the transition
can still be projected onto a supertype of
the resultant local type.
To ensure that association is preserved,
we must also consider the other projections.
It can be easily observed that
performing transitions with subject $\roleP$
does not drastically impact
the projection onto $\roleR(\ne\roleP)$.

\begin{lemma}[Global Type Transitions]
\label{lem:global-red}
Suppose that $\gtProjRel{\gtG}{\roleR}{(\quH,\stT)}$ and
$\gtG\,\gtMove[\stEnvAnnotGenericSym]\,\gtGi$.
\begin{enumerate}
	\item If $\stEnvAnnotGenericSym=\ltsBra{\roleP}{\roleQ}{\stChoice{\stLab}{\tyGround}}$
	and $\roleR\ne\roleQ$, then $\gtProjRel{\gtGi}{\roleR}{(\quH,\stT)}$.
	\item If $\stEnvAnnotGenericSym=\ltsBra{\roleP}{\roleR}{\stChoice{\stLab}{\tyGround}}$, then
	$\quH=\quCons{\quMsg{\roleP}{\stLab}{\tyGround}}{\quHi}$ and
	$\gtProjRel{\gtGi}{\roleR}{(\quHi,\stT)}$.
	\item \label{lem:global-red:3}Otherwise, there is $\stTi$ such that
	$\stMerge{}{\{\stT,\stTi\}}\ni\stT$ and
	$\gtProjRel{\gtGi}{\roleR}{(\quH,\stTi)}$.
\end{enumerate}
Note that $\stMerge{}{\{\stT,\stTi\}}\ni\stT$ implies
that $\stT\stSub\stTi$ by Lemma~\ref{lem:mer:subtype}.
\end{lemma}
\begin{proof}
	Induction on the relation $\gtG\,\gtMove[\stEnvAnnotGenericSym]\,\gtGi$,
	making use of~\cref{lem:good-merge-assoc} and~\cref{lem:mer:setinclusion}
	to maintain the invariant specified in~\cref{lem:global-red:3}.
\end{proof}
Given a transition from $\stEnv$ associated to a $\text{balanced}^+$ global type, $\gtG$,
to $\stEnvi$,
we can now:
identify the redexes inside the projection of $\gtG$;
identify corresponding redexes inside $\gtG$;
perform a corresponding transition of $\gtG$ to $\gtGi$;
and project $\gtGi$ to observe that $\stEnvi$ is associated to $\gtGi$.
We state this formally, focusing on a single participant, in~\cref{lem:gt-move},
before applying this to contexts, proving~\cref{thm:gtype:proj-comp}.

  \begin{lemma}[Global Type Mirrors Local Actions]
    \label{lem:gt-move}
    Assume 
    \[
         \gtProjRel{\gtG}{\roleP}
        {(\quH[\roleP],\stTi[\roleP])}
         \quad\text{with}\quad
         \stT[\roleP]\;\asubt\;\stTi[\roleP].
      \]
    \begin{enumerate}
      \item If $\unfoldOne{\stT[\roleP]} = \stIntSum{\roleQ}{i\in I}{\stChoice{\stLab[i]}{\tyGround[i]}\stSeq\stT[i]}$ then
      for any $k \in I$, $\exists \gtGi$ such that 
      \(
        \gtProjRel{\gtGi}{\roleP}
        {(\quCons{\quH[\roleP]}{\quMsg{\roleQ}{\stLab[k]}{\tyGroundi}},\stTi[k])}
      \) and $\tyGround[k]\tyGroundSub\tyGroundi$
      and $\stT[k] \asubt \stTi[k]$ and 
      $\gtG  \,\gtMove[\ltsSel{\roleP}{\roleQ}{\stLab[k](\tyGroundi)}]\, \gtGi$.
      \item If $\unfoldOne{\stT[\roleP]} = \stExtSum{\roleQ}{i\in I}{\stChoice{\stLab[i]}{\tyGround[i]}\stSeq\stT[i]}$ and 
      $\gtProjRel{\gtG}{\roleQ}
        {(\quCons{\quMsg{\roleP}{\stLab}{\tyGround}}{\quH[\roleQ]},\stTi[\roleQ])} 
         \quad\text{with}\quad
         \stT[\roleQ]\;\asubt\;\stTi[\roleQ]$ where
         $\stLab = \stLab[k]$ and $\tyGround \tyGroundSub \tyGround[k]$ for some $k \in I$, then 
      \(
        \gtProjRel{\gtGi}{\roleP}
        {(\quH[\roleP],\stTi[k])}
      \) and $\stT[k] \asubt \stTi[k]$ and $\gtProjRel{\gtGi}{\roleQ}{\stTi[\roleQ]}$ and 
      $\gtG  \,\gtMove[\ltsBra{\roleP}{\roleQ}{\stLab[k](\tyGround)}]\, \gtGi$.
      \item If $\unfoldOne{\stT[\roleP]} = \stEnd$ then $\quH[\roleP] = \quEmpty[]$ and $\unfoldOne{\stTi[\roleP]} = \stEnd$.

    \end{enumerate}
    \end{lemma}
    \begin{proof} \leavevmode
    \begin{enumerate} 
        \item $\unfoldOne{\stT[\roleP]} = \stIntSum{\roleQ}{i\in I}{\stChoice{\stLab[i]}{\tyGround[i]}\stSeq\stT[i]}$.
        Applying \cref{lem:inv-subtyping}, we have that
        $\ttree{\stT} =
        \trCtxSelp[
        \stIntSum{\roleP}
        {i \in I_j}{\stChoice{\stLab[i]}{\tyGround[i,j]}
        \stSeq \trTi[i,j]}]_{j\in J}
        \forgetHole{\trTii[j]}{j\in J'}$
        with $I \subseteq I_j$ for all $j\in J$ and for all $i \in I$ and $j\in J$,
        we have $\tyGround[i] \tyGroundSub \tyGround[i,j]$,
        and if $\ttree{\stTi[i]}=\trCtxSelp[\trTi[i,j]]_{j\in J}\forgetHole{\trTii[j]}{j\in J'}$
        then $\stT[i] \asubt \stTi[i] $.
        Now we apply \cref{lem:inv-proj} to get that
        $\ttree{\gtG} =
        \trCtxGSelpq[
        \gtComm{\roleP}{\roleQ}
        {i \in I_{f(j)}}
        {\stLab[i,f(j)]}
        {\tyGround[i,f(j)]}
        {\gtG[i,j]}]_{j\in K}\forgetHole{\gtGiii[j]}{j\in K'}$.
        We apply \cref{lem:gt-move} to get that
        $\ttree{\gtG} =
        \trCtxGSelpq[
        \gtComm{\roleP}{\roleQ}
        {i \in I_{f(j)}}
        {\stLab[i,f(j)]}
        {\tyGround[i,f(j)]}
        {\gtG[i,j]}]_{j\in K}\forgetHole{\gtGiii[j]}{j\in K'}
        \gtMove \trCtxGSelpq{}'[\gtG[k,j]]_{j\in K}$ for $k\in I$. 
        Then we use the additional conclusions
        to reconstruct the subtyping derivation from before,
        with updated premises taken from the results of applying the lemma, 
        to deduce that $\gtProjRel{\gtGi}{\roleP}{\stTi[k]}$ where
        $\ttree{\stTi[k]} = \trCtxSel{}'[\trT[k,j]]_{j\in J}$ and
        $\ttree{\gtGi}=\trCtxGSelpq{}'[\gtG[k,j]]_{j\in K}$.
        \item (Similar to previous case).
        \item From the projection rules and applying \cref{lem:end-sub-end}.
    \end{enumerate}
    \end{proof}

    \thmProjCompleteness*
    \begin{proof}
        By case analysis on $\stEnvAnnotGenericSym$. 
        \bigskip
        \noindent
        \textbf{Case $\stEnvAnnotGenericSym = \ltsSel{\roleP}{\roleQ}{\stLab(\tyGround)}$
        (a send by a participant).}
      
        \smallskip
        Applying \Cref{lem:inv-ctx-lts}, for some label~$\stLab$,
        payload $\tyGround$, sender~$\roleP$
        and receiver~$\roleQ$, we have
      
        \[
          \begin{aligned}
            &\stEnvApp{\stEnv}{\roleP}
                =(\quH[\roleP], \stT[\roleP]) \\ 
                & \text{with } \stT[\roleP] = \stIntSum{\roleQ}{i\in I}{\stChoice{\stLab[i]}{\tyGround[i]}\stSeq\stT[i]}
                \text{ where }\,\stLab[j]=\stLab,\,\tyGround[j]=\tyGround.
          \end{aligned}
        \]
        We also know the endpoint at $\roleP$ is updated in the new context,
        \(
            \stEnvApp{\stEnvi}{\roleP}
               =(\quCons{\quH[\roleP]}{\quMsg{\roleQ}{\stLab}{\tyGround}},\stT[j]),
        \)
        while all other endpoints stay unchanged.
        As we know $\stT[\roleP] \ne \stEnd$, then $\roleP \in \dom{\stEnv}$, 
        and so by association there exists $\stTi[\roleP]$ such that
        \begin{equation}
           \gtProjRel{\gtG}{\roleP}
             {(\quHi[\roleP],\stTi[\roleP])}
           \quad\text{with}\quad
           \stT[\roleP]\;\asubt\;\stTi[\roleP]
	\text{ and }\quH[\roleP]\tySub\quHi[\roleP]
        \end{equation}
        Hence we can apply \Cref{lem:gt-move} to obtain the desired result.
        \begin{equation}
          \gtProjRel{\gtGi}{\roleP}
            {(\quCons{\quHi[\roleP]}{\quMsg{\roleQ}{\stLab}{\tyGroundi}},\stTi[j])}
          \quad\text{with}\quad
          \stT[j]\;\asubt\;\stTi[j]
	\quad\text{and}\quad
	\tyGround\tyGroundSub\tyGroundi
       \end{equation}
        And so, by also applying~\cref{lem:global-red},we have 
        $\stEnvAssoc{\stEnvi}{\gtGi}{a}$
        and
        $\gtG  \,\gtMove[\ltsSel{\roleP}{\roleQ}{\stLab(\tyGroundi)}]\, 
          \gtGi$ in this case.
      
        \bigskip
        \noindent
        \textbf{Case $\stEnvAnnotGenericSym = \ltsBra{\roleQ}{\roleP}{\stLab(\tyGround)}$
        (receive by a participant).}

        \smallskip
        Dually to the above case, we can again apply \Cref{lem:inv-ctx-lts}, for some label~$\stLab$,
        payload $\tyGround$, sender~$\roleQ$
        and receiver~$\roleP$. Then we have
          \begin{equation}
          \begin{split}
            &\stEnvApp{\stEnv}{\roleP}
                =(\quH[\roleP], \stT[\roleP]) \quad \text{and} \quad \stEnvApp{\stEnv}{\roleQ}
                =(\quCons{\quMsg{\roleP}{\stLab}{\tyGroundi}}{\quH[\roleQ]}, \stT[\roleQ])\\ 
                & \text{with } \stT[\roleP] = \stExtSum{\roleQ}{i\in I}{\stChoice{\stLab[i]}{\tyGround[i]}\stSeq\stT[i]}
                \text{ where }\,\stLab[j]=\stLab,\,\tyGroundi\tyGroundSub\tyGround[j]=\tyGround.
		\end{split}
          \end{equation}
        Now the endpoints at $\roleP$ and $\roleQ$ are updated, so 
        \(
            \stEnvApp{\stEnvi}{\roleP}
               =(\quH[\roleP],\stT[j]),
        \)
        and 
        \(
            \stEnvApp{\stEnvi}{\roleQ}
               =(\quH[\roleQ],\stT[\roleQ]),
        \)
        while other endpoints in $\stEnvi$ stay the same.
        Again, given $\stT[\roleP] \ne \stEnd$, there exists $\stTi[\roleP]$ such that
        \begin{equation}
           \gtProjRel{\gtG}{\roleP}
             {(\quHi[\roleP],\stTi[\roleP])}
           \quad\text{with}\quad
           \stT[\roleP]\;\asubt\;\stTi[\roleP].
	\text{ and }\quH[\roleP]\tySub\quHi[\roleP].
        \end{equation}
        and 
        \begin{equation}
          \gtProjRel{\gtG}{\roleQ}
            {(\quCons{\quMsg{\roleP}{\stLab}{\tyGroundii}}{\quHi[\roleQ]},\stTi[\roleQ])}
          \quad\text{with}\quad
          \stT[\roleQ]\;\asubt\;\stTi[\roleQ]\quad\text{and}\quad
	\tyGroundi\tyGroundSub\tyGroundii
	\text{ and }\quH[\roleQ]\tySub\quHi[\roleQ].
       \end{equation}
        Hence we can apply \Cref{lem:gt-move} to obtain the desired result.
        \begin{equation}
          \gtProjRel{\gtGi}{\roleP}
            {(\quHi[\roleP],\stTi[j])}
          \quad\text{with}\quad
          \stT[j]\;\asubt\;\stTi[j].
       \end{equation} and 
       \begin{equation}
        \gtProjRel{\gtGi}{\roleQ}
          {(\quHi[\roleQ],\stTi[\roleQ])}.
     \end{equation}
        And so we have 
        $\stEnvAssoc{\stEnvi}{\gtGi}{a}$
        and
        $\gtG  \,\gtMove[ \ltsBra{\roleQ}{\roleP}{\stLab(\tyGroundii)}]\, 
          \gtGi$.
      \end{proof}

Now that we have~\cref{thm:gtype:proj-comp},
\cref{thm:gtype:proj-sound} is truly quite weak.
Indeed, we only need to know that
contexts, associated to a non-$\gtEnd$ global type,
can move.

\begin{lemma}[Projection preserves operational semantics]
  \label{lem:proj-preserves-sem}
  If $\gtG\,\gtMove$, $\gtProjRel{\gtG}{\roleP}{\stEnvApp{\stEnv}{\roleP}}$ for all $\roleP \in \dom{\stEnv}\supseteq\gtRoles{\gtG}$,
  then $\stEnv \, \stEnvMoveAnnot{}$.
\end{lemma}  
\begin{proof}
If $\gtG\,\gtMove$, then $\unfoldOne{\gtG}\neq\gtEnd$ so
write $\unfoldOne{\gtG} =
\gtComm{\roleP}{\roleQ}{i \in I}{\gtLab[i]}{\tyGround[i]}{\gtG[i]}$ or
$\gtCommSquig{\roleP}{\roleQ}{\gtLab[j]}{i \in I}{\gtLab[i]}{\tyGround[i]}{\gtG[i]}$.
If $\unfoldOne{\gtG} = \gtComm{\roleP}{\roleQ}{i \in I}{\gtLab[i]}{\tyGround[i]}{\gtG[i]}$, then
let $\stEnvApp{\stEnv[\gtG]}{\roleP}=(\quH,\stT)$ so
$\unfoldOne{\stT}=\stIntSum{\roleQ}{i \in I}{\stChoice{\stLab[i]}{\tyGround[i]} \stSeq \stT[i]}$ so
$\stEnv \, \stEnvMoveAnnot{\ltsSel{\roleP}{\roleQ}{\gtLab[j]}}$ for $j\in I$.
If $\unfoldOne{\gtG} = \gtCommSquig{\roleP}{\roleQ}{\gtLab[j]}{i \in I}{\gtLab[i]}{\tyGround[i]}{\gtG[i]}$, then
let $\stEnvApp{\stEnv[\gtG]}{\roleQ}=(\quH,\stT)$ and $\stEnvApp{\stEnv[\gtG]}{\roleP}=(\quHi,\stTi)$ so
$\quHi = \quCons{\quMsg{\roleQ}{\gtLab[j]}{\tyGround[j]}}{\quHii}$ and
$\unfoldOne{\stT}=\stExtSum{\roleP}{i \in I}{\stChoice{\stLab[i]}{\tyGround[i]} \stSeq \stT[i]}$ so
$\stEnv \, \stEnvMoveAnnot{\ltsBra{\roleQ}{\roleP}{\gtLab[j]}}$.
\end{proof}

\thmProjSoundness*
\begin{proof}
By association, we know there exists a context $\stEnvii$ such that $\gtProjRel{\gtG}{\roleP}{\stEnvApp{\stEnvii}{\roleP}}$ for all $\roleP \in \dom{\stEnvii}$ with $\stEnv \asubt \stEnvii$. 
Then we can apply \cref{lem:proj-preserves-sem} to get that $\stEnvii \, \stEnvMoveAnnot{\stEnvAnnotGenericSym}$.
So, $\stEnv \, \stEnvMoveAnnot{\stEnvAnnotGenericSymi}$ as
the precise asynchronous subtyping relates non-deadlocked contexts to non-deadlocked contexts.
Then the desired result is obtained by applying \cref{thm:gtype:proj-comp}.
\end{proof}

\subsection{Properties of Associated Contexts}
\label{sec:live}
Now only~\cref{thm:assoc-live} remains to be shown.
Similarly to the synchronous setting,
liveness implies deadlock-freedom.
However unlike in the synchronous setting,
liveness also implies safety.
Therefore, it suffices to prove that
associated contexts are live.
With the observation that liveness is downward closed under
the precise asynchronous subtyping,
we need only focus on projected contexts.

\begin{lemma}[Liveness is Downwards closed under subtyping~{\cite[Lemma 4.10]{GPPSY2023}}]
  \label{lem:live-down-closed}
  If $\stEnv$ is live and $\stEnvi \asubt \stEnv$ then $\stEnvi$ is live.
\end{lemma}  

Now, the precise asynchronous subtyping cannot cause any further issues.
The proof that a projected context is live is very similar
to the proofs in the inductive case and in the synchronous setting.
We have already observed sufficient operational correspondence
between global types and projected contexts to derive liveness.

\begin{lemma}[Standard Association is Complete]
\label{lem:std-complete}
	If $\gtProjRel{\gtG}{\roleP}{\stEnvApp{\stEnvi}{\roleP}}$
    for all $\roleP \in \dom{\stEnv}$,
    $\gtRoles{\gtG}\subseteq\dom{\stEnv}$,
    $\stEnv\stSub\stEnvi$, and
    $\stEnv\,\stEnvMoveGenAnnot\,\stEnvii$, then
    there exist $\gtGi$, $\stEnvAnnotGenericSymi\succcurlyeq\stEnvAnnotGenericSym$, and $\stEnviii$ such that
    $\gtProjRel{\gtGi}{\roleP}{\stEnvApp{\stEnviii}{\roleP}}$
    for all $\roleP \in \dom{\stEnv}$,
    $\stEnvi\stSub \stEnviii$, and
    $\gtG\,\gtMove[\stEnvAnnotGenericSymi]\,\gtGi$.
\end{lemma}
\begin{proof}
	Follows the same structure as the proof of~\cref{thm:gtype:proj-comp}.
\end{proof}

\begin{lemma}[Standard Association is Head-Sound]
\label{lem:std-sound}
    If $\gtProjRel{\gtG}{\roleP}{\stEnvApp{\stEnv}{\roleP}}$
    for all $\roleP \in \dom{\stEnv}$,
    $\gtRoles{\gtG}\subseteq\dom{\stEnv}$, then:
    \begin{itemize}
        \item If $\unfoldOne{\gtG}=\gtComm{\roleP}{\roleQ}{i \in I}{\gtLab[i]}{\tyGround[i]}{\gtG[i]}$, then
        $\stEnv\,\stEnvMoveAnnot{\ltsSel{\roleP}{\roleQ}{\stChoice{\stLab[k]}{\tyGround[k]}}}$
        for some $k\in I$.
        \item If $\unfoldOne{\gtG}=\gtCommSquig{\roleP}{\roleQ}{\stLab}{}{\gtLab}{\tyGround}{\gtGi}$, then
        $\stEnv\,\stEnvMoveAnnot{\ltsBra{\roleP}{\roleQ}{\stChoice{\stLab}{\tyGround}}}$.
    \end{itemize}
\end{lemma}
\begin{proof}
Clear from the definition of projection.
\end{proof}

With this operational correspondence,
we can proceed normally.
We take fair sequences from a projected context
and use~\cref{lem:std-complete} to
determine a corresponding sequence of global types.
By~\cref{lem:std-sound} and fairness,
every global type head must be reduced.
Finally by $\text{balanced}^+$, we deduce that
every redex will either be reduced or,
eventually, become the head of the global type,
and then be reduced.

\begin{lemma}[Projection ensures Liveness]
  \label{lem:proj-live}
  If $\gtProjRel{\gtG}{\roleP}{\stEnvApp{\stEnv}{\roleP}}$ for all $\roleP \in \dom{\stEnv}$ and $\gtRoles{\gtG}\subseteq\dom{\stEnv}$,
then $\stEnv$ is live.
\end{lemma}  
\begin{proof}
Let $\left(\stEnv[n]\right)_{n\in N}$ with $\left(\stEnvAnnotGenericSym_{n}\right)_{n,n+1\in N}$ be a fair sequence starting from $\stEnv$.
By~\cref{lem:std-complete}, find
$\left(\gtG[n]\right)_{n\in N}$ (\gtG[0]=\gtG) and 
$\left(\stEnvi[n]\right)_{n\in N}$ (\stEnvi[0]=\stEnv) such that
$\gtProjRel{\gtG[n]}{\roleP}{\stEnvApp{\stEnvi[n]}{\roleP}}$ for all $\roleP \in \dom{\stEnv}$,
$\stEnv[n]\stSub\stEnvi[n]$, and
$\gtG[n]\,\gtMove[{\stEnvAnnotGenericSym_n}]\,\gtG[n+1]$.
If $\stEnvApp{\stEnv[i]}{\roleP}$ has a message for $\roleQ$ in its queue, then
$(\roleP,\roleQ)\in\gtMRoles{\gtG[i]}$ so
$\gtMDepth{\gtG[i]}{\roleP}{\roleQ}$ exists.
Proceed by induction on $\gtMDepth{\gtG[i]}{\roleP}{\roleQ}$.
Suppose that this message is never dequeued.
Now apply~\cref{lem:std-sound} via fairness to get to $j>i$ and~\cref{lem:depth-decr} by our assumption.
The message is still on the queue by our assumption, and $\gtMDepth{\gtG[j]}{\roleP}{\roleQ}<\gtMDepth{\gtG[i]}{\roleP}{\roleQ}$.
By the I.H, the message will be dequeued.
The case for $\stEnvApp{\stEnv[i]}{\roleP}$ having a receive at the head is similar,
but we proceed by induction on $\gtDepth{\gtG[i]}{\roleP}$ instead of $\gtMDepth{\gtG[i]}{\roleP}{\roleQ}$.
\end{proof}

\thmProjLive*
\begin{proof}
Let $\gtProjRel{\gtG}{\roleP}{\stEnvApp{\stEnvi}{\roleP}}$ for all $\roleP \in \dom{\stEnv}$
and $\stEnv\asubt\stEnvi$, by association.
By~\cref{lem:proj-live} and~\cref{lem:live-down-closed}, $\stEnvi$ is live so $\stEnv$ is live.
\end{proof}
\section{Deriving the Main Theorems from Associations}
\label{sec:properties}
This section demonstrates how to derive the main theorems 
using soundness and completeness of the associations, together with 
the corresponding results from Ghilezan et al.~\cite[Theorems~4.11, 4.12 and
  4.13]{GPPSY2023}.

\subsection{Asynchronous Multiparty Session Processes}
We summarise the calculus from Ghilezan et al.~\cite{GPPSY2023} and
main properties.

\myparagraph{Process Syntax.}
We define \emph{processes} ($\mpP, \mpQ, \mpP_i, \ldots$) by the following grammar:
\[
\mpP \ ::= \ \ 
\roleQ\mathord{!}\ell\langle e \rangle \mpSeq \mpP 
\bnfsep \textstyle\sum_{i\in I} \roleQ\mathord{?}\ell_i(x_i)\mpSeq \mpP_i 
\bnfsep \mpIf{e}{\mpP}{\mpQ} 
\bnfsep \mpNil 
\bnfsep \mpRec{\mpX}{\mpP} 
\bnfsep \mpX
\]
where $\roleQ\mathord{!}\ell\langle e \rangle \mpSeq \mpP$ is a \emph{send} (selection) to role $\roleQ$ with label $\ell$ and payload expression $e$, 
followed by continuation $\mpP$; 
$\sum_{i\in I} \roleQ\mathord{?}\ell_i(x_i)\mpSeq \mpP_i$ is a \emph{receive} (branching) from role $\roleQ$, 
where a message with label $\ell_k$ binds variable $x_k$ and continues as $\mpP_k$;
$\mpIf{e}{\mpP}{\mpQ}$ is a conditional;
$\mpNil$ is the inactive process;
$\mpRec{\mpX}{\mpP}$ and $\mpX$ are for recursion. 
\myparagraph{Queue Syntax.}
We define \emph{message queues} ($h,h_i,\ldots$) by the following grammar:
\[
h\ ::= \ \
\epsilon
\bnfsep (\roleQ, \ell(\mpV))
\bnfsep h_1 \cdot h_2
\]
where $\epsilon$ is the empty queue;
$(\roleQ, \ell(\mpV))$ is the queue denoting that the message with
label $\ell$ and value $\mpV$
has been sent to $\roleQ$; and
$h_1 \cdot h_2$ is the concatenation
of the queues $h_1$ and $h_2$.
We consider message queues up to
an equivalence relation $\equiv$,
similar to~\cref{def:queue-equiv}
but with $\epsilon$ for $\quEmpty$,
values for basic types,
and message queues for syntactic queue types.

\begin{remark}[Session Interleaving and Delegation]
This calculus excludes session interleaving
and delegation elements.
This choice was made by Ghilezan et al.~\cite{GPPSY2023},
to simplify the calculus by Coppo et
       al.~\cite{CDPY2015}
	and focus on subtyping.
	For more detail and examples of how to tackle this,
	see the work by Coppo et al.~\cite{CDYP13} and 
        van den Heuvel and P\'erez~\cite{Van_den_Heuvel2022-hr}.
\end{remark}

\myparagraph{Typing Judgements.}
We write $\stJudge{\mpEnv}{P}{\stT}$ to denote that process $P$ has local type $\stT$ 
under process variable environment $\mpEnv$, 
and $\stJudge{}{h}{\quH}$ to denote that queue $h$ has queue type $\quH$. 
The typing rules for processes are standard and are shared by both the top-down and
bottom-up systems. We write $\mpEnv$ for the process-variable and expression-variable context,
mapping process variables to their local types and expression variables to their sorts.
\RULE{T-$\mpNil$} types the inactive process; \RULE{T-Var} and
\RULE{T-Rec} allow recursion at the process level; \RULE{T-$\oplus$} and
\RULE{T-$\&$} allow typing of selections and branchings; \RULE{T-Cond}
types conditionals; and the subsumption rule \RULE{T-Sub} allows for subtyping.
Queue rules type the empty queue, single messages, and concatenation.
We assume a separate sorting system for expressions, at least providing
judgements of the form $\stJudge{\mpEnv}{e}{\tyBool}$ and $\stJudge{\mpEnv}{e}{\tyInt}$ for
boolean and integer expressions.
\begin{definition}[Typing Rules for Processes~\cite{GPPSY2023}]
\label{def:process-typing}
\leavevmode
\[
\begin{array}{c}
\inference[T-$\mpNil$]{}{\stJudge{\mpEnv}{\mpNil}{\stEnd}}
\qquad
\inference[T-Var]{\mpX : \stRecVar \in \mpEnv}{\stJudge{\mpEnv}{\mpX}{\stRecVar}}
\qquad
\inference[T-Sub]{\stJudge{\mpEnv}{\mpP}{\stT} \quad \stT \asubt \stTi}{\stJudge{\mpEnv}{\mpP}{\stTi}}
\\[3ex]
\inference[T-$\oplus$]{
  \stJudge{\mpEnv}{e_j}{\tyGround[j]} \quad
  \stJudge{\mpEnv}{\mpP_j}{\stT[j]} \quad j \in I
}{
  \stJudge{\mpEnv}{\roleQ\mathord{!}\ell_j\langle e_j \rangle \mpSeq \mpP_j}{
    \stIntSum{\roleQ}{i \in I}{\stChoice{\ell_i}{\tyGround[i]} \stSeq \stT[i]}
  }
}
\qquad
\inference[T-$\&$]{
  \forall i \in I \quad \stJudge{\mpEnv, x_i : \tyGround[i]}{\mpP_i}{\stT[i]}
}{
  \stJudge{\mpEnv}{\textstyle\sum_{i\in I} \roleQ\mathord{?}\ell_i(x_i)\mpSeq \mpP_i}{
    \stExtSum{\roleQ}{i \in I}{\stChoice{\ell_i}{\tyGround[i]} \stSeq \stT[i]}
  }
}
\\[3ex]
\inference[T-Cond]{
  \stJudge{\mpEnv}{e}{\tyBool} \quad
  \stJudge{\mpEnv}{\mpP_1}{\stT} \quad
  \stJudge{\mpEnv}{\mpP_2}{\stT}
}{
  \stJudge{\mpEnv}{\mpIf{e}{\mpP_1}{\mpP_2}}{\stT}
}
\qquad
\inference[T-Rec]{
  \stJudge{\mpEnv, \mpX : \stRecVar}{\mpP}{\stT}
}{
  \stJudge{\mpEnv}{\mpRec{\mpX}{\mpP}}{\stRec{\stRecVar}{\stT}}
}
\end{array}
\]

\noindent
The typing rules for queues are:
\[
\inference{}{\stJudge{}{\epsilon}{\quEmpty}}
\qquad
\inference{
  \stJudge{}{\mpV}{\tyGround}
}{
  \stJudge{}{(\roleQ, \ell(\mpV))}{\quMsg{\roleQ}{\ell}{\tyGround}}
}
\qquad
\inference{
  \stJudge{}{h_1}{\quH[1]} \quad
  \stJudge{}{h_2}{\quH[2]}
}{
  \stJudge{}{h_1 \cdot h_2}{\quCons{\quH[1]}{\quH[2]}}
}
\]
\end{definition}
\noindent
The subsumption rule \RULE{T-Sub} can be explained by Liskov and Wing's substitution principle \cite{Liskov1994-ra}:
if $\stT$ is a subtype of $\stTi$,
then an object of type $\stT$
can always replace an object of type $\stTi$.

\begin{example}[Ring Protocol Processes]
\label{ex:ring-processes}
Returning to the ring-choice protocol $\gtG[\text{\tiny ring}]$
from \cref{sec:ring},
a set of processes implementing the protocol can be given as follows:
\begin{align}
  \mpP[\roleP] &= \mpRec{\mpX}{
    \roleQ\mathord{!}\mathsf{add}\langle n \rangle \mpSeq
    \left(
      \roleR\mathord{?}\mathsf{add}(x) \mpSeq \mpX
      +
      \roleR\mathord{?}\mathsf{sub}(x) \mpSeq \mpX
    \right)
  } \\
  \mpP[\roleQ] &= \mpRec{\mpX}{
    \mpIf{e}{
      \roleR\mathord{!}\mathsf{add}\langle m \rangle \mpSeq
      \roleP\mathord{?}\mathsf{add}(y) \mpSeq \mpX
    }{
      \roleR\mathord{!}\mathsf{sub}\langle m \rangle \mpSeq
      \roleP\mathord{?}\mathsf{add}(y) \mpSeq \mpX
    }
  } \\
  \mpP[\roleR] &= \mpRec{\mpX}{
    \roleQ\mathord{?}\mathsf{add}(z) \mpSeq
      \roleP\mathord{!}\mathsf{add}\langle z{+}k \rangle \mpSeq \mpX
    +
    \roleQ\mathord{?}\mathsf{sub}(z) \mpSeq
      \roleP\mathord{!}\mathsf{sub}\langle z{-}k \rangle \mpSeq \mpX
  }
\end{align}
Processes $\mpP[\roleP]$ and $\mpP[\roleR]$ are typed by the
corresponding projected local types from \cref{sec:ring}:
$\stJudge{}{\mpP[\roleP]}{\stT[\roleP]}$ and
$\stJudge{}{\mpP[\roleR]}{\stT[\roleR]}$.

Observe that,
as far as $\stT[\roleQ]$ is concerned,
$\roleQ$'s outgoing message to $\roleR$
need not depend on the value
received from $\roleP$.
A programmer can therefore write $\mpP[\roleQ]$ to send to $\roleR$
\emph{before} receiving from $\roleP$,
allowing the $\roleQ\to\roleR$ and $\roleP\to\roleQ$ communications
to proceed concurrently,
improving throughput
as illustrated in \cref{fig:optimisation}(b).
The projected type $\stT[\roleQ]$ requires the receive first,
but $\mpP[\roleQ]$ is typed by $\stTopt[\roleQ]$,
and since $\stTopt[\roleQ] \asubt \stT[\roleQ]$,
the subsumption rule \RULE{T-Sub} accepts $\mpP[\roleQ]$
in place of any process typed by $\stT[\roleQ]$.
In this way, the global protocol $\gtG[\text{\tiny ring}]$
guarantees safety, deadlock-freedom, and liveness for the system,
while $\asubt$ gives the programmer the freedom
to choose a more efficient implementation for $\roleQ$.
\end{example}

Figure~\ref{fig:overview-of-bottomup-properties} highlights the bottom-up workflow that we use in this section:
local types are jointly model-checked against the property $\phi$ and guide the typing of the programs.
\begin{figure}[]
\centering
{\footnotesize
  \begin{tikzpicture}[node distance=12mm and 26mm, every node/.style={align=center}]
    \node (phi) {Property $\varphi$};
    \node[above=2mm of phi] (mc) {\small {\bf{model-check}}\,($\models$)};
    \node[below=12mm of phi] (LB) {\footnotesize Local Type for $\roleQ$\\ \small \boxed{\stT_{\roleQ}}};
    \node[left=26mm of LB] (LA) {\footnotesize Local Type for $\roleP$\\ \small \boxed{\stT_{\roleP}}};
    \node[right=26mm of LB] (LC) {\footnotesize Local Type for $\roleR$\\ \small \boxed{\stT_{\roleR}}};
    \node[below=12mm of LA] (PA) {\, \, \, \footnotesize Program for $\roleP$\\ \small \boxed{P_{\roleP}}};
    \node[below=12mm of LB] (PB) {\, \, \, \footnotesize Program for $\roleQ$\\ \small \boxed{P_{\roleQ}}};
    \node[below=12mm of LC] (PC) {\, \, \, \footnotesize Program for $\roleR$\\ \small \boxed{P_{\roleR}}};
    \draw[->] (LA) -- node[left=3mm]{\small {\bf{typing}}\,($\vdash$)} (PA);
    \draw[->] (LB) -- (PB);
    \draw[->] (LC) -- (PC);
    \draw[->] (LA) -- (phi);
    \draw[->] (LB) -- (phi);
    \draw[->] (LC) -- (phi);
\end{tikzpicture}
}
\caption{Bottom-up methodology for multiparty session types. 
Processes $P_{\roleP}$, $P_{\roleQ}$ and $P_{\roleR}$ are typed by
  local types $\stT_{\roleP}$, $\stT_{\roleQ}$ and $\stT_{\roleR}$, which are
  jointly model-checked against the property $\varphi$.}
\label{fig:overview-of-bottomup-properties}
\end{figure}

\myparagraph{Asynchronous Multiparty Sessions.}
We define \emph{asynchronous multiparty session} ($\N,\N_i,...$) as:
\[\N \ ::= \ \ 
\roleP \triangleleft P_{\roleP} \ | \ \roleP \triangleleft h_{\roleP} \bnfsep \N \ | \ \N'\] 
where $\roleP \triangleleft P_{\roleP}$ denotes a process $P_{\roleP}$
plays as a role $\roleP$.

The structural precongruence relation $\prestruct$ used in
\cref{def:mpst-reduction-relation} below allows us to disregard inactive processes,
unfolds recursion, and is closed under reflexivity and transitivity:
\[
\begin{array}{c}
\inference[SC-Idle]{}{\roleP \triangleleft \mpNil \;\mpPar\; \roleP \triangleleft \epsilon \;\mpPar\; \mathcal{M} \prestruct \mathcal{M}}
\qquad
\inference[SC-Refl]{}{\mathcal{M} \prestruct \mathcal{M}}
\\[2.2ex]
\inference[SC-Unfold]{}{\mpRec{\mpX}{\mpP} \prestruct \mpP\{\mpRec{\mpX}{\mpP}/\mpX\}}
\qquad
\inference[SC-Trans]{\mathcal{M} \prestruct \mathcal{M}' \quad \mathcal{M}' \prestruct \mathcal{M}''}{\mathcal{M} \prestruct \mathcal{M}''}
\end{array}
\]

 \begin{definition}[Reduction relation on sessions (asynchronous)]
 \label{def:mpst-reduction-relation}
\leavevmode\par\centering
\centerline{
\begingroup
\setlength{\arraycolsep}{2pt}\def\arraystretch{1.12}\begin{adjustbox}{width=\columnwidth,center}
\ensuremath{\begin{array}{l@{\;\;}lcl}
\RULE{R-Send} &
  \roleP \triangleleft \roleQ\mathord{!}\ell\langle e\rangle \mpSeq \mpP
  \;\mpPar\; \roleP \triangleleft h_{\roleP}
  \;\mpPar\; \mathcal{M}
  &\mpMove&
  \roleP \triangleleft \mpP
  \;\mpPar\; \roleP \triangleleft h_{\roleP}\mathbin{\cdot}(\roleQ,\ell(\mpV))
  \;\mpPar\; \mathcal{M}
  \quad (e \downarrow \mpV)
\\[0.35em]
\RULE{R-Rcv} &
  \begin{aligned}[t]
    \roleP \triangleleft \textstyle\sum_{i\in I} \roleQ\mathord{?}\ell_i(x_i)\mpSeq \mpP_i
    &\;\mpPar\; \roleP \triangleleft h_{\roleP} \\
    &\;\mpPar\; \roleQ \triangleleft \mpQ
    \;\mpPar\; \roleQ \triangleleft ( \roleP,\ell_k(\mpV))\mathbin{\cdot} h
    \;\mpPar\; \mathcal{M}
  \end{aligned}
  &\mpMove&
  \begin{aligned}[t]
    \roleP \triangleleft \mpP_k\{\mpV/x_k\}
    &\;\mpPar\; \roleP \triangleleft h_{\roleP} \\
    &\;\mpPar\; \roleQ \triangleleft \mpQ
    \;\mpPar\; \roleQ \triangleleft h
    \;\mpPar\; \mathcal{M}
  \end{aligned}
  \quad (k\in I)
\\[0.35em]
\RULE{R-Cond-T} &
  \roleP \triangleleft \mpIf{e}{\mpP}{\mpQ}
  \;\mpPar\; \roleP \triangleleft h
  \;\mpPar\; \mathcal{M}
  &\mpMove&
  \roleP \triangleleft \mpP
  \;\mpPar\; \roleP \triangleleft h
  \;\mpPar\; \mathcal{M}
  \quad (e \downarrow \mpTrue)
\\[0.35em]
\RULE{R-Cond-F} &
  \roleP \triangleleft \mpIf{e}{\mpP}{\mpQ}
  \;\mpPar\; \roleP \triangleleft h
  \;\mpPar\; \mathcal{M}
  &\mpMove&
  \roleP \triangleleft \mpQ
  \;\mpPar\; \roleP \triangleleft h
  \;\mpPar\; \mathcal{M}
  \quad (e \downarrow \mpFalse)
\\[0.35em]
\RULE{R-Struct} &
  \mathcal{M}_1 \prestruct \mathcal{M}'_1 \quad
  \mathcal{M}'_1 \mpMove \mathcal{M}'_2 \quad
  \mathcal{M}'_2 \prestruct \mathcal{M}_2
  &\Longrightarrow&
  \mathcal{M}_1 \mpMove \mathcal{M}_2
\\[0.35em]
\RULE{Err-Mism} &
  \roleP \triangleleft \textstyle\sum_{i\in I} \roleQ\mathord{?}\ell_i(x_i)\mpSeq \mpP_i
  \;\mpPar\; \roleP \triangleleft h_{\roleP}
  \;\mpPar\; \roleQ \triangleleft \mpQ
  \;\mpPar\; \roleQ \triangleleft (\roleP,\ell(\mpV))\mathbin{\cdot}h
  \;\mpPar\; \mathcal{M}
  &\mpMove& \mpErr
  \quad (\forall i\in I.\ \ell_i \neq \ell)
\\[0.35em]
\RULE{Err-Eval} &
  \roleP \triangleleft \mpIf{e}{\mpP}{\mpQ}
  \;\mpPar\; \roleP \triangleleft h
  \;\mpPar\; \mathcal{M}
  &\mpMove& \mpErr
  \quad (e \ndownarrow \mpTrue \text{ and } e \ndownarrow \mpFalse)
\end{array}
}\end{adjustbox}
\endgroup}
\end{definition}

\myparagraph{Safety and Deadlock-Freedom.}
We first define safety and deadlock-freedom for multiparty sessions.

\begin{definition}[Safety and Deadlock-Freedom]
\label{def:safety-deadlock-free}
A session $\N$ is:
\begin{enumerate}
\item \emph{communication safe} iff for all $\N'$ such that $\N \mpMove^\ast \N'$, we have $\N' \not\mpMove \mpErr$;
\item \emph{deadlock-free} iff for all $\N'$ such that $\N \mpMove^\ast \N'$ and $\N' \not\mpMove$, 
we have $\N' \equiv \mpBigPar{i \in I}{(\roleP_i \triangleleft \mpNil \;\mpPar\; \roleP_i \triangleleft \epsilon)}$.
\end{enumerate}
\end{definition}

\noindent
Intuitively, safety ensures that no reachable session can reduce to an error state.
Deadlock-freedom ensures that if a session cannot reduce further, 
then all processes have terminated and all queues are empty.

\myparagraph{Session Liveness.}
Liveness is a stronger property that subsumes both safety and deadlock-freedom (recall \cref{def:safe,def:df,def:live} for the corresponding context-level properties).
We define liveness for multiparty sessions, ensuring that
under fair scheduling, all pending actions are eventually performed.
A \emph{session path} is a (potentially infinite) sequence of sessions $(\N_i)_{i\in I}$, 
where $I = \{0,1,2,\ldots\}$ is a set of consecutive natural numbers, 
and for all $i \in I$, $\N_i \mpMove \N_{i+1}$.

\begin{definition}[Fair and Live Session Paths~\cite{GPPSY2023}]
\label{def:session-liveness}
A session path $(\N_i)_{i\in I}$ is \emph{fair} iff, for all $i \in I$:
\begin{itemize}
\item[\textsf{SF1}] if $\N_i \equiv \roleP \triangleleft \roleQ\mathord{!}\ell\langle e\rangle \mpSeq \mpP 
\;\mpPar\; \roleP \triangleleft h \;\mpPar\; \N'$,
then $\exists k$ such that $I \ni k \geq i$ and $\N_k \mpMove \N_{k+1}$ 
via \RULE{R-Send} at role $\roleP$

\item[\textsf{SF2}] if $\N_i \equiv \roleP \triangleleft \textstyle\sum_{j\in J} \roleQ\mathord{?}\ell_j(x_j)\mpSeq \mpP_j 
\;\mpPar\; \roleP \triangleleft h \;\mpPar\; \N'$
and $\roleQ \triangleleft (\roleP, \ell_k(\mpV))\mathbin{\cdot}h' \;\mpPar\; \N''$ for some $k \in J$,
then $\exists k'$ such that $I \ni k' \geq i$ and $\N_{k'} \mpMove \N_{k'+1}$ 
via \RULE{R-Rcv} at role $\roleP$

\item[\textsf{SF3}] if $\N_i \equiv \roleP \triangleleft \mpIf{e}{\mpP}{\mpQ} 
\;\mpPar\; \roleP \triangleleft h \;\mpPar\; \N'$,
then $\exists k$ such that $I \ni k \geq i$ and $\N_k \mpMove \N_{k+1}$ 
via \RULE{R-Cond-T} or \RULE{R-Cond-F} at role $\roleP$
\end{itemize}

\noindent
A session path $(\N_i)_{i\in I}$ is \emph{live} iff, for all $i \in I$:
\begin{itemize}
\item[\textsf{SL1}] if $\N_i \equiv \roleP \triangleleft \mpIf{e}{\mpP}{\mpQ} 
\;\mpPar\; \roleP \triangleleft h \;\mpPar\; \N'$,
then $\exists k$ such that $I \ni k \geq i$ and, for some $\N''$,
$\N_k \equiv \roleP \triangleleft \mpP \;\mpPar\; \roleP \triangleleft h \;\mpPar\; \N''$ or 
$\N_k \equiv \roleP \triangleleft \mpQ \;\mpPar\; \roleP \triangleleft h \;\mpPar\; \N''$

\item[\textsf{SL2}] if $\N_i \equiv \roleP \triangleleft \roleQ\mathord{!}\ell\langle e\rangle \mpSeq \mpP 
\;\mpPar\; \roleP \triangleleft h \;\mpPar\; \N'$,
then $\exists k$ such that $I \ni k \geq i$ and $\N_k \mpMove \N_{k+1}$ 
via \RULE{R-Send} at role $\roleP$

\item[\textsf{SL3}] if $\N_i \equiv \roleP \triangleleft \mpP \;\mpPar\; \roleP \triangleleft (\roleQ, \ell(\mpV))\mathbin{\cdot}h \;\mpPar\; \N'$,
then $\exists k$ such that $I \ni k \geq i$ and $\N_k \mpMove \N_{k+1}$ 
via \RULE{R-Rcv} at role $\roleQ$

\item[\textsf{SL4}] if $\N_i \equiv \roleP \triangleleft \textstyle\sum_{j\in J} \roleQ\mathord{?}\ell_j(x_j)\mpSeq \mpP_j 
\;\mpPar\; \roleP \triangleleft h \;\mpPar\; \N'$,
then $\exists k, \ell'$ such that $I \ni k \geq i$, $\ell' \in \{\ell_j\}_{j\in J}$, and 
$\N_k \mpMove \N_{k+1}$ via \RULE{R-Rcv} at role $\roleP$ receiving label $\ell'$
\end{itemize}

\noindent
A session $\N$ is \emph{live} iff all fair paths starting from $\N$ are live.
\end{definition}

Intuitively, fairness (\textsf{SF1}--\textsf{SF3}) requires that if a component (process or queue) 
is ready to fire an action, then that action will eventually be performed.
Liveness (\textsf{SL1}--\textsf{SL4}) then ensures that under fair execution:
conditionals are eventually resolved (\textsf{SL1}), 
pending outputs are eventually performed (\textsf{SL2}), 
queued messages are eventually received (\textsf{SL3}), 
and pending inputs eventually receive a message (\textsf{SL4}).

\subsection{Deriving the Main Theorems} 
We recall the bottom-up typing system for a multiparty session:
\[
\inference[SessBot]{\forall \roleP \in \dom{\Delta}
  \quad \stJudge{}{P_{\roleP}}{\stT[\roleP]} \quad 
\stJudge{}{h_{\roleP}}{\sigma_{\roleP}} \quad 
\stEnv(\roleP)=(\sigma_{\roleP},\stT[\roleP])\quad 
\varphi(\stEnv)}{\stJudgeBot{}{{\Pi_{\roleP\in \dom{\Delta}}{\ (\roleP \triangleleft 
P_{\roleP} \ | \ \roleP \triangleleft h_{\roleP})\ }}}{\stEnv}}
\]
where
$\stJudge{}{P_{\roleP}}{\stT[\roleP]}$ states that
process $P_{\roleP}$ has type $\stT[\roleP]$, 
$\stJudge{}{h_{\roleP}}{\sigma_{\roleP}}$ types message queue
$h_{\roleP}$ has a queue type $\sigma_{\roleP}$, and 
$\varphi$ is some desired property, which is usually a
\emph{safety} property--a selected label is always available at the
branching process 
\cite{POPL19LessIsMore,YH2024}.
In the work of Ghilezan et al.~\cite{GPPSY2023}, a \emph{liveness} property \cite[Definition 4.17]{GPPSY2023} is used instead for proving the preciseness of
$\asubt$. To derive the main theorems,
we do not require the details of typing rules, hence we omit. 
See Ghilezan et al.~\cite[\S~7.1]{GPPSY2023} for the full typing system and its explanations.   

\paragraph{\bf Deriving Subject Reduction Theorem.}
We prove the subject reduction theorem of the top-down system using
the completeness of the association with 
the following subject reduction theorem of the bottom-up system. 

\begin{theorem}[Subject~Reduction,~{\cite[Theorem~4.11]{GPPSY2023}}]
\label{the:srbot}
Assume 
${\stJudgeBot{}{\N}{\stEnv}}$ with $\stEnv$ live and 
$\N\hpMoveStar\N'$. 
Then there exist live $\stEnvi$, $\stEnvii$ such that 
${\stJudgeBot{}{\N'}{\stEnvi}}$ 
with $\stEnvii \asubt \stEnv$ and $\stEnvii\stEnvMoveStar\stEnvi$.
\end{theorem}

\begin{theorem}[Subject Reduction of the Top-Down System]
\label{the:srtop}
Assume $\stJudgeTop{}{\N}{\stEnv}$, $\stEnvAssoc{\stEnv}{\gtG}{a}$, and  
$\N\hpMoveStar\N'$. 
Then there exists $\gtGi$ such that 
${\stJudgeTop{}{\N'}{\stEnvi}}$,  
$\gtG\,\,\,\,\gtMoveStar\,\gtGi$ and
$\stEnvAssoc{\stEnvi}{\gtGi}{a}$.
\end{theorem}
\begin{proof}
We prove the following stronger statement, from which we can derive the
above theorem
\begin{quote}  
Assume $\stJudgeTop{}{\N}{\stEnv}$, $\stEnvAssoc{\stEnv}{\gtG}{a}$, and  
$\N\hpMoveStar\N'$. 
Then there exist $\stEnvi$, $\stEnvii$, and $\gtGi$ such that 
${\stJudgeTop{}{\N'}{\stEnvi}}$ 
with $\stEnvii\asubt\stEnv$,
$\stEnvii\stEnvMoveStar\stEnvi$,
$\gtG\,\,\,\,\gtMoveStar\,\gtGi$, and
$\stEnvAssoc{\stEnvi}{\gtGi}{a}$.
\end{quote}
Assume $\N\equiv 
\Pi_{\roleP\in \dom{\Delta}}
{(\roleP \triangleleft P_{\roleP} \ | \ \roleP \triangleleft h_{\roleP})}$ 
and $\stJudgeTop{}{\N}{\stEnv}$
is derived with 
\begin{equation}
\label{eq:typingrule}
\forall \roleP\in \dom{\stEnv}\quad  \stJudge{}{P_{\roleP}}{\stT[\roleP]} \quad \stJudge{}{h_{\roleP}}{\sigma_{\roleP}} \quad \stEnv(\roleP)=(\sigma_{\roleP},\stT[\roleP])\quad \stEnvAssoc{\stEnv}{\gtG}{a}
\end{equation}
by $\inferrule{SessST}$. 
Suppose $\N\hpMove \N'$. 
We need to prove that there exist $\gtGi$ and $\stEnvii$ such that 
$\Pi_{\roleP\in \roleset{\gtGi}}\ 
(\roleP \triangleleft P'_{\roleP} \ | \ \roleP \triangleleft h'_{\roleP})$
with $\stEnvAssoc{\stEnvi}{\gtGi}{a}$.

Note that $\stEnv$ is live by~\cref{thm:assoc-live}. 
Hence by~\cref{the:srbot},  
there exist live $\stEnvi$, $\stEnvii$ such that 
${\stJudgeBot{}{{\Pi_{\roleP\in \roleset{\gtG}}\ {P'_{\roleP}}}}{\stEnvii}}$ 
with $\stEnvii \asubt \stEnv$ and $\stEnvii\stEnvMoveStar\stEnvi$.
By~\cref{def:assoc}, 
$\stEnvAssoc{\stEnvi}{\gtG}{a}$.  
Then by~\cref{thm:gtype:proj-comp}, 
$\stEnvii\stEnvMoveStar\stEnvi$ implies 
$\gtG\,\,\,\,\gtMoveStar\,\gtGi$ and 
$\stEnvAssoc{\stEnvi}{\gtGi}{a}$. 
Hence ${\stJudgeTop{}{{\Pi_{\roleP\in \dom{\stEnvi}}\ {P'_{\roleP}}}}{\stEnvi}}$
as desired.
\end{proof}
\begin{remark}[Subject Reduction Labels]
\label{rem:sr-label}
By examining the proof of~\cref{the:srbot},
we can deduce the labels of the context
transition sequence
from the reduction rules
used in the session reduction sequence.
For simplicity, we consider a single session reduction
as this implies the full subject reduction.
Suppose that $\N\hpMove \N'\not\equiv\mpErr$ and consider the
non-\RULE{R-Struct} rule deriving it.
\begin{itemize}[leftmargin=0.7in,labelindent=-\leftmargin]
\item[\RULE{R-Send}] If the reduction occurred via \RULE{R-Send}
	at role $\roleP$ with label $\ell$ and destination $\roleQ$,
	then
	$\stEnvi\,\stEnvMoveAnnot{\ltsSel{\roleP}{\roleQ}{\stLab(\tyGround)}}\,\stEnvii$
	for some ground type $\tyGround$.
\item[\RULE{R-Rcv}] If the reduction occurred via \RULE{R-Rcv}
	at role $\roleP$ receiving label $\ell$ and with source $\roleQ$,
	then
	$\stEnvi\,\stEnvMoveAnnot{\ltsBra{\roleP}{\roleQ}{\stLab(\tyGround)}}\,\stEnvii$
	for some ground type $\tyGround$.
\item[\RULE{R-Cond-T}] $\stEnvi=\stEnvii$.
\item[\RULE{R-Cond-F}] $\stEnvi=\stEnvii$.
\end{itemize}
\end{remark}

\paragraph{\bf Deriving Session Fidelity.}
We derive session fidelity of the top-down system. 
We use the soundness and completeness of the association with 
session fidelity of the bottom-up system 

\begin{theorem}[Session Fidelity,~{\cite[Theorem~4.13]{GPPSY2023}}]
\label{the:sfbot}
Assume ${\stJudgeBot{}{\N}{\stEnv}}$ 
with $\stEnv$ live. Assume $\stEnv\stEnvMove$.
Then there exist $\N'$, $\stEnvi$, and $\stEnvii$ such that 
$\N\hpMove^+\N'$, $\stEnvii\stEnvMove\stEnvi$, $\stEnvii\asubt\stEnv$,
and 
${\stJudgeBot{}{\N'}{\stEnvi}}$.
\end{theorem}

\begin{theorem}[Session~Fidelity of the Top-Down System]
\label{the:sftop}
Assume $\stJudgeTop{}{\N}{\stEnv}$ is derived by 
$\stEnvAssoc{\stEnv}{\gtG}{a}$ and 
$\gtG\,\gtMove$. Then there exist 
$\N'$, $\stEnvi$, and $\gtGi$ such that 
$\N\hpMove^+\N'$,
$\gtG\,\gtMove\,\gtG'$, and   
${\stJudgeTop{}{\N'}{\stEnvi}}$ with 
$\stEnvAssoc{\stEnvi}{\gtGi}{a}$. 
\end{theorem}
\begin{proof}
We prove the following stronger statement, by which the above theorem
is derived.
\begin{quote}
Assume $\stJudgeTop{}{\N}{\stEnv}$ is derived by 
$\stEnvAssoc{\stEnv}{\gtG}{a}$ and 
$\gtG\,\gtMove$. Then there exist 
$\N'$, $\stEnvi$, $\stEnvii$, and $\gtGi$ such that 
$\N\hpMove^+\N'$,
$\stEnvii\stEnvMove\stEnvi$,
$\stEnvii\asubt\stEnv$,
$\gtG\,\gtMove\,\gtG'$, and   
${\stJudgeTop{}{\N'}{\stEnvi}}$ with 
$\stEnvAssoc{\stEnvi}{\gtGi}{a}$. 
\end{quote}  
Assume 
$\stEnvAssoc{\stEnv}{\gtG}{a}$. By the soundness of the association, 
$\gtG\,\gtMove$ implies  
$\stEnv\stEnvMove$. 
Suppose $\N\equiv 
\Pi_{\roleP\in \dom{\Delta}}
{(\roleP \triangleleft P_{\roleP} \ | \ \roleP \triangleleft h_{\roleP})}$ 
and $\stJudgeTop{}{\N'}{\stEnv}$
is derived with 
(\ref{eq:typingrule}) above. 
By~\cref{the:sfbot}, 
there exist $\N'$, $\stEnvi$, and $\stEnvii$ such that 
$\N\hpMove^+\N'$, $\stEnvii\stEnvMove\stEnvi$, and $\stEnvii\asubt\stEnv$.
Hence by the completeness of the association, 
and~\cref{the:srtop}, 
$\gtG\,\gtMove\,\gtGi$ and 
$\stEnvAssoc{\stEnvi}{\gtGi}{a}$ with  
${\stJudgeTop{}{\N'}{\stEnvi}}$, as desired. 
\end{proof}
\begin{remark}[No Session Fidelity Labels]
Unlike~\cref{rem:sr-label},
we cannot identify the rules used in
the session reduction
from the initial global transition.
We cannot even identify
the participant undertaking
the reduction!
The reason is the same as in~\cref{rem:sound-suff},
the asynchronous subtyping
allows us to reorganise participants' actions.
Therefore, the global type transition
may reflect an action
that can be brought forward by
supertyping,
but that the session cannot do right now.
Regardless,
the concluding transitions
are induced by~\cref{the:srtop},
hence the labels can be deduced
from~\cref{rem:sr-label}.
\end{remark}
Next we show that typed multiparty sessions 
are live.

\begin{theorem}[Session Liveness,~{\cite[Theorem~4.12]{GPPSY2023}}]
\label{thm:livebot}
Assume ${\stJudgeBot{}{\N}{\stEnv}}$ 
with $\stEnv$ live.
Then $\N$ is live.
\end{theorem}

\begin{theorem}[Liveness of the Top-Down System]
\label{the:livetop}
Assume $\stJudgeTop{}{\N}{\stEnv}$. Then, $\N$ is safe, deadlock-free
and live.
\end{theorem}
\begin{proof}
We first note that if $\N$ is live, then $\N$ is safe and
deadlock-free. 
If $\stEnvAssoc{\stEnv}{\gtG}{a}$, then $\stEnv$ is live
by~\cref{thm:assoc-live}
and we have that $\stJudgeBot{}{\N}{\stEnv}$
as $\stJudgeTop{}{\N}{\stEnv}$.
Then by~\cref{thm:livebot}, $\N$ is live.
\end{proof}  

\begin{example}[Program Requiring Message Reordering]
	Consider a program that may exhibit
	different behaviours depending on non-deterministic control flow;
	this may require the precise asynchronous subtyping to type,
	since the different behaviours may need to be reconciled via
	message reordering.
	
	\noindent
	Recall $\gtG[\text{\tiny non-det}]$ from \cref{ex:non-det},
	its projections,
	and their subtypes:
	\begin{equation}
		\gtG[\text{\tiny non-det}] = \gtRec{\gtRecVar}{
			\gtCommRaw{\roleP}{\roleQ}{
				\begin{array}{@{}l@{}}
					\gtLab[1]\gtSeq
					\gtCommSingle{\roleQ}{\roleR}{\gtLab[1]}{}{\gtRecVar}
					\\
					\gtLab[2]\gtSeq
					\gtCommSingle{\roleQ}{\roleR}{\gtLab[2]}{}{
						\gtCommSingle{\roleP}{\roleR}{\gtLab}{}{}
					}
				\end{array}
			}
		}
	\end{equation}
	\begin{equation}
		\stT[\roleP] = \stRec{\stRecVar}{
			\stIntSumRaw{\roleQ}{
				\begin{array}{l}
					\stChoice{\stLab[1]}{}\stSeq\stRecVar\\
					\stChoice{\stLab[2]}{}\stSeq \roleR\stFmt{\oplus}\stChoice{\stLab}{}
				\end{array}
			}
		}
		\qquad
		\stT[\roleQ] = \stRec{\stRecVar}{
			\stExtSumRaw{\roleP}{
				\begin{array}{l}
					\stChoice{\stLab[1]}{}\stSeq \roleR\stFmt{\oplus}\stChoice{\stLab[1]}{}\stSeq\stRecVar\\
					\stChoice{\stLab[2]}{}\stSeq \roleR\stFmt{\oplus}\stChoice{\stLab[2]}{}
				\end{array}
			}
		}
		\qquad
		\stT[\roleR] = \stRec{\stRecVar}{
			\stExtSumRaw{\roleQ}{
				\begin{array}{l}
					\stChoice{\stLab[1]}{}\stSeq\stRecVar\\
					\stChoice{\stLab[2]}{}\stSeq \roleP\stFmt{\&}\stChoice{\stLab}{}
				\end{array}
			}
		}
\end{equation}	
\begin{equation}
	\stT[n]=
	\roleR\stFmt{\oplus}\stChoice{\stLab}{}\,
	\underbrace{
	\stSeq \,\roleQ\stFmt{\oplus}\stChoice{\stLab[1]}{}
	\dots
	\stSeq \roleQ\stFmt{\oplus}\stChoice{\stLab[1]}{}
	}_{n}
	\,\stSeq\, \roleQ\stFmt{\oplus}\stChoice{\stLab[2]}{}
	\qquad
	\stTi=\stRec{\stRecVar}{\roleQ\stFmt{\oplus}\stChoice{\stLab[1]}{}\stSeq\stRecVar}
\end{equation}
	We can use these to type a protocol that
	cannot be typed without the precise asynchronous subtyping.
	Suppose that $e_1$ and $e_2$
	are non-deterministic boolean typed expressions.

	\begin{equation}
	\begin{array}{c}
		\mpP[\roleP] =
			\begin{array}{l}
				\mpFmt{\mathbf{if}\,e_1\,\mathbf{then}\,}
				\roleR\mathord{!}\stLab\mpSeq\roleQ\mathord{!}\stLab[2]
				\\
				\mpFmt{\mathbf{else\,if}\,e_2\,\mathbf{then}\,}
				\roleR\mathord{!}\stLab\mpSeq\roleQ\mathord{!}\stLab[1]\mpSeq\roleQ\mathord{!}\stLab[2]
				\\
				\mpFmt{\mathbf{else}\,}
				\mpRec{\mpX}{
					\roleQ\mathord{!}\stLab[1]\mpSeq\mpX
				}
			\end{array}
			\qquad
		\mpP[\roleQ] =
			\mpRec{\mpX}{
				\sum
				\left\{
				\begin{array}{l}
				\roleP\mathord{?}\stLab[1]\mpSeq\roleR\mathord{!}\stLab[1]\mpSeq\mpX
				\\
				\roleP\mathord{?}\stLab[2]\mpSeq\roleR\mathord{!}\stLab[2]
				\end{array}
				\right\}
			}
			\qquad
			\mpP[\roleR] =
			\mpRec{\mpX}{
				\sum\left\{
					\begin{array}{l}
						\roleQ\mathord{?}\stLab[1]\mpSeq\mpX
						\\
						\roleQ\mathord{?}\stLab[2]\mpSeq
							\roleP\mathord{?}\stLab
					\end{array}
				\right\}
			}
	\end{array}
	\end{equation}
	We can derive the following type judgements without using subtyping:
	\begin{equation}
		\stJudge{}{\mpP[\roleQ]}{\stT[\roleQ]}
		\qquad
		\stJudge{}{\mpP[\roleR]}{\stT[\roleR]}
		\qquad
		\stJudge{}{\roleR\mathord{!}\stLab\mpSeq\roleQ\mathord{!}\stLab[2]}{\stT[0]}
		\qquad
		\stJudge{}{\roleR\mathord{!}\stLab\mpSeq\roleQ\mathord{!}\stLab[1]\mpSeq\roleQ\mathord{!}\stLab[2]}{\stT[1]}
		\qquad
		\stJudge{}{\mpRec{\mpX}{
						\roleQ\mathord{!}\stLab[1]\mpSeq\mpX
				}}{\stTi}
	\end{equation}
	We have that, for all $n\geq 0$, $\stT[n]\asubt\stT[\roleP]$
	and $\stTi\asubt\stT[\roleP]$,
	so, using \RULE{T-Sub}, we get the following type judgements:
	\begin{equation}
		\stJudge{}{\roleR\mathord{!}\stLab\mpSeq\roleQ\mathord{!}\stLab[2]}{\stT[\roleP]}
		\qquad
		\stJudge{}{\roleR\mathord{!}\stLab\mpSeq\roleQ\mathord{!}\stLab[1]\mpSeq\roleQ\mathord{!}\stLab[2]}{\stT[\roleP]}
		\qquad
		\stJudge{}{\mpRec{\mpX}{
						\roleQ\mathord{!}\stLab[1]\mpSeq\mpX
				}}{\stT[\roleP]}
	\end{equation}
	Hence, we can type $\stJudge{}{\mpP[\roleP]}{\stT[\roleP]}$,
	allowing us to type the session using $\gtG[\text{\tiny non-det}]$:
	\begin{equation}
	\stJudge{}{\left(
		\roleP\triangleleft\mpP[\roleP]\mpPar\roleP\triangleleft\epsilon
		\mpPar
		\roleQ\triangleleft\mpP[\roleQ]\mpPar\roleQ\triangleleft\epsilon
		\mpPar
		\roleR\triangleleft\mpP[\roleR]\mpPar\roleR\triangleleft\epsilon
	\right)}
	{\stEnvAssoc{\left(
		\stEnvMap{\roleP}{\stT[\roleP]}
		\stEnvComp
		\stEnvMap{\roleQ}{\stT[\roleQ]}
		\stEnvComp
		\stEnvMap{\roleR}{\stT[\roleR]}
		\right)}{\gtGi[\text{\tiny non-det}]}{a}
	}
	\end{equation}
	Therefore, the session is safe, deadlock-free and live by \cref{the:livetop}.
	If our subtyping did not include asynchronous message reordering,
	then we would be unable to reconcile the control flows of $\mpP[\roleP]$,
	making it untypable.
	
	\noindent
	We use this session to witness the utility of the
	non-deterministic global transition relation,
	as in \cref{ex:non-det},
	by considering the three possible control flows:
	\\[2ex]
	\centerline{
	\begin{tikzpicture}
		\node[draw] (P1) at (0,0) {
		$
		\begin{array}{clclclcl}
			&\roleP&\triangleleft&\mpP[\roleP]&\mpPar&\roleP&\triangleleft&\epsilon\\
			\mpPar&
			\roleQ&\triangleleft&\mpP[\roleQ]&\mpPar&\roleQ&\triangleleft&\epsilon\\
			\mpPar&
			\roleR&\triangleleft&\mpP[\roleR]&\mpPar&\roleR&\triangleleft&\epsilon
			\end{array}
		$
		};
		\node[draw] (P2) at (-5,-2) {
			$
		\begin{array}{clclclcl}
			&\roleP&\triangleleft&\roleR\mathord{!}\stLab\mpSeq\roleQ\mathord{!}\stLab[2]&\mpPar&\roleP&\triangleleft&\epsilon\\
			\mpPar&
			\roleQ&\triangleleft&\mpP[\roleQ]
			&\mpPar&\roleQ&\triangleleft&\epsilon\\
			\mpPar&
			\roleR&\triangleleft&\mpP[\roleR]&\mpPar&\roleR&\triangleleft&\epsilon
			\end{array}
		$
		}; 
		\node[draw] (P3) at (5,-2) {
			$
		\begin{array}{clclclcl}
			&\roleP&\triangleleft&\roleR\mathord{!}\stLab\mpSeq\roleQ\mathord{!}\stLab[1]\mpSeq\roleQ\mathord{!}\stLab[2]&\mpPar&\roleP&\triangleleft&\epsilon\\
			\mpPar&
			\roleQ&\triangleleft&\mpP[\roleQ]
			&\mpPar&\roleQ&\triangleleft&\epsilon\\
			\mpPar&
			\roleR&\triangleleft&\mpP[\roleR]&\mpPar&\roleR&\triangleleft&\epsilon
			\end{array}
		$
		}; 
		
		\node[draw] (P4) at (0,-4) {
			$
		\begin{array}{clclclcl}
			&\roleP&\triangleleft&\mpRec{\mpX}{
					\roleQ\mathord{!}\stLab[1]\mpSeq\mpX
				}&\mpPar&\roleP&\triangleleft&\epsilon\\
			\mpPar&
			\roleQ&\triangleleft&\mpP[\roleQ]
			&\mpPar&\roleQ&\triangleleft&\epsilon\\
			\mpPar&
			\roleR&\triangleleft&\mpP[\roleR]&\mpPar&\roleR&\triangleleft&\epsilon
			\end{array}
		$
		}; 
		\draw[->,semithick] (P1) to node[very near end, yshift = 2mm, xshift=-1mm] {$*$} (P2);
		\draw[->,semithick] (P1) to node[very near end, yshift = 2mm, xshift=1mm] {$*$} (P3);
		\draw[->,semithick] (P1) to node[very near end, yshift = 0mm, xshift=1mm] {$*$} (P4);
	\end{tikzpicture}}
	All three sessions are typed by $\gtG[\text{\tiny non-det}]$,
	but, after $\roleP$ enqueues the next message to $\roleR$
	in the left and right sessions,
	the session on the left can be typed by
	$\gtGi$
	while the session on the right cannot be. Where:
	\[
		\gtGi=\gtCommSingle{\roleP}{\roleQ}{\gtLab[2]}{}{
					\gtCommSingle{\roleQ}{\roleR}{\gtLab[2]}{}{
						\gtCommSquigSingle{\roleP}{\roleR}{\gtLab}{\gtLab}{}{}
					}
			}
	\]
	This explains why the transitions of $\gtG[\text{\tiny non-det}]$
	allow for non-determinism:
	it needs to handle all processes of the above form
	(with any number of communications from $\roleP$ to $\roleQ$),
	which no single transition can handle simultaneously.
\end{example}

\section{Related Work}
\label{sec:related}
\myparagraph{Preciseness of Subtyping. } 
For synchronous programs, two subtyping approaches have been
proposed: one allowing the safe substitution of channels (called \emph{channel  subtyping})\cite{Gay2005} and
the other allowing the safe substitution of processes (called \emph{process subtyping}) \cite{DemangeonH11}. In this 
work, we use process subtyping, 
which represents refinement of process behaviours. 

\emph{Preciseness of subtyping} was first proposed for the call-by-value
$\lambda$-calculus with iso-recursive types by Ligatti et al.~\cite{Jay2017}. 
Chen et al.~\cite{Chen2017} introduced 
and proved \emph{preciseness} of  
synchronous \cite{Gay2005,DemangeonH11} and asynchronous subtyping for 
the binary (2-party) session types. 
Later, the asynchronous subtyping 
was found to be \emph{undecidable}, independently by Bravetti et
al.~\cite{BravettiCZ17}, 
and by Lange and Yoshida~\cite{LY2017}. 
This provoked active studies on identifying 
a set of binary session types where asynchronous subtyping 
is decidable; and proposing \emph{sound} algorithms 
\cite{lmcs:7238,BocchiKM24}  
extending the formalism to \emph{binary} communicating automata 
\cite{Brand1983}. 
A variant of asynchronous binary session types
based on service contract theory, called       
the \emph{fair} asynchronous subtyping relation, was also proposed and
proved sound and complete in  
\cite{bravetti_et_al:LIPIcs.CONCUR.2025.11}.  

In the multiparty setting,
Ghilezan et al.~\cite{Ghilezan2019,GPPSY2023} 
proposed precise synchronous and asynchronous 
session subtyping, and our theory is 
centred on precise multiparty asynchronous subtyping. 
Recently, Bocchi et al.~\cite{bocchi_et_al:CONCUR.2025.10} 
have shown a sound algorithm for asynchronous multiparty session types 
that can decide subtyping
for non-trivial multiparty communication patterns, extending  
an abstract interpretation framework from Bocchi et al~\cite{BocchiKM24}. 
Li et al.~\cite{Li2024-mu} proposed
a sound and complete algorithm
for asynchronous behavioural contract refinement
subtyping on communicating state machines (CSMs) \cite{Brand1983},
which is distinct from the precise asynchronous subtyping:
as it is defined to be complete with respect to subprotocol fidelity.
As noted in \cite[\S~7]{BocchiKM24}, classically, the subtyping 
question is answered pairwise: by comparing a type against a candidate
subtype, without reference to a global type. Our classical formulation
is favourable when types are mined from source code, allowing
\emph{local} optimisations of specifications and code per participant.
See the paragraph on \textbf{Implementations} below. 

\myparagraph{Top-Down Asynchronous Systems.}  
The first multiparty session types system 
\cite{HYC08,HYC2016} (MPST) 
introduced \emph{asynchronous multiparty session processes} which interact through 
FIFO queues. The typing system uses 
the most restricted end-point projection (called an 
inductive projection with plain merging)
and does not consider subtyping. Later this projection was extended to 
an inductive projection with full merging with synchronous subtyping 
relations, which is used in several asynchronous MPST systems 
(e.g., \cite{DBLP:journals/corr/abs-1208-6483}). 
Scalas and Yoshida have discovered that proofs of  
an inductive projection with full merging in the literature 
are flawed \cite{POPL19LessIsMore}.  
Recently this proof was fixed by Hou and Yoshida \cite{YH2024} for 
a \emph{synchronous calculus}, using 
an association relation between a global type and typing 
context up to the synchronous subtyping relation. 
We apply their proof method to an asynchronous 
calculus up to the precise asynchronous subtyping, 
with a coinductive projection with full merging; the coinductive full merge subsumes the inductive full merge of earlier MPST systems, since any inductively-derivable merge also satisfies our coinductive rules. Hence
our type system offers strictly larger typable processes than those typed
by the association relation in \cite{YH2024}.
Barbanera and Dezani-Ciancaglini \cite{DBLP:conf/isola/BarbaneraD24}
extend global types with syntax for delegation
and separate the sending and receiving syntax,
allowing some of the expressivity given by the precise asynchronous
subtyping.
They give a decidable type system directly
between sessions and global types,
without using a projection.
Notably, their global types do not include en-route transmissions
and are not given semantics;
instead they define \emph{type configurations},
which are pairs of global types and queues,
and give these semantics.
Unlike \cref{the:srtop,the:sftop}, their subject reduction
and session fidelity theorems
refer to type configurations \emph{not} global types.
This subsumes points (1) and (2)
from \cref{def:count}, and
avoids point (3) since
messages can be globally dequeued
at only point,
independent of the structure of the global type.

Asynchronous local types are often viewed as 
Communicating Finite State Machines (CFSMs)~\cite{Brand1983}.
The work \cite{DY12} which first studied 
a connection between multiparty local types and CFSMs 
proposed a graphical asynchronous 
calculus with mixed choice which is typed by a graphical global type 
with mixed choice and parallel composition. 
After this work, researchers started studying properties 
of global types or global protocols using the theory of 
CFSMs and applied their theories to implementations of MPSTs, 
see \cite{YZF2021}. Extending methodologies of message sequence charts (MSC), 
recent work focus on \emph{implementability} or \emph{realisability}--whether there is a 
CFSM implementation of 
the specified global protocol (or MSC)
\cite{DBLP:conf/concur/MajumdarMSZ21,Li2023-cr,stutz:LIPIcs.ECOOP.2023.32,Giusto2025}.
Their approach aims to produce decidable,
sound and complete (with respect to traces)
projections from sender-driven global types 
(or global state machines) to CFSMs.

While CFSMs related session type works only focus on 
CFSMs (and global protocols), 
recent work in \cite{10.1007/978-3-031-91121-7_13}
has proposed Protocol State Machines (PSMs) and
Automata-based Multiparty Protocols (AMP)
as an alternative framework for 
asynchronous message passing processes. 
Their PSMs project onto CFSMs,
which are then used to type processes,
similar to how global types project onto
local types. 
While a typing context projected from 
a global type (modulo the synchronous subtyping) can only produce 
1-synchronous
and half-duplex behaviours (as proven \emph{multiparty
compatibility} in \cite{DY13}), 
their work \cite{DY12,10.1007/978-3-031-91121-7_13} 
and ours present more expressive global type behaviours. 
By extending en-route transmissions notation
and context transmission rules for global types, 
our global and local type semantics do not impose half-duplex
limitations.

The most important distinction between all of the above CFSM systems  
and ours is the property of \emph{liveness};
using well-formed global types enforces liveness in the presence of
single selection and conditional processes (features of the original session
calculi \cite{Takeuchi1994,honda.vasconcelos.kubo:language-primitives}), 
while the calculi from~\cite{DY12,10.1007/978-3-031-91121-7_13}
exclude these
constructs. 
In our system,
liveness is downward closed with respect to $\asubt$
as proven in \cite[Lemma 4.10]{GPPSY2023}. 
See \cite[Example 5.14]{POPL19LessIsMore}
for a counterexample which demonstrates liveness without
the fair path assumption (weaker liveness)
is \emph{not} downward closed with respect to
subtyping relations, thus the typing system is unsound under their weaker
liveness (which corresponds to orphan message freedom in 
\cite{DY12,DY13,10.1007/978-3-031-91121-7_13}).  
In contrast, we have proven that the
coinductive projection with full merging 
is sound and complete with respect to asynchronous trace reordering, 
demonstrating the correctness of the top-down 
MPST with asynchronous precise multiparty session subtyping relation.

\myparagraph{Implementations of Asynchronous Multiparty Session
  Subtyping.\ }
Asynchronous session subtyping was first introduced for 
messaging optimisation in session-based high-performance 
computing platforms,
i.e., multicore C programming \cite{HVY2009,YoshidaVPH08} and MPI-C \cite{NYH2012,NCY2015}. 
Castro-Perez and Yoshida~\cite{CastroPerezY20} proposed CAMP, which is a static 
performance analysis framework for message-passing concurrent and
distributed systems based on MPST. CAMP augments MPST with
annotations of communication latency and local computation cost,
defined as estimated execution times, which is used to extract cost
equations from protocol descriptions and to statically predict the
communication cost. CAMP is also extended to analyse asynchronous communication optimised programs. The tool, based on cost
theory, is applicable to different existing benchmarks and usecases
in the literature with a wide range of communication protocols. 

The Rust programming framework, Rumpsteak~\cite{CYV2022}, 
incorporates multiparty asynchronous subtyping to optimise 
asynchronous message-passing in the Rust programming language. 
It proposes an algorithm for asynchronous subtyping
based on the session decomposition technique that is
bounded by a number of iterations and proved to be sound and
decidable. They evaluate the performance and expressiveness of
Rumpsteak against previous Rust implementations and 
algorithms,  and show that
Rumpsteak is more efficient and can safely implement case studies 
by offering arbitrary ordering of messages. 

\myparagraph{Mechanisations of Asynchronous Subtyping Relations. }
Hinrichsen et al.~\cite{lmcs:6869} introduce Actris, a Iris Rocq that
integrates separation logics and asynchronous binary session types
with the asynchronous subtyping by Chen et al.~\cite{Chen2017}. 
The first mechanisation of multiparty asynchronous subtyping
by the Rocq proof assistant is proposed by Ekici and Yoshida~\cite{EY2024,EY2026}.
Their Rocq formalisation follows 
a session tree decomposition approach, 
using the greatest fixed point of the least fixed point technique, 
facilitated by the paco library, to define coinductive predicates.
Li et al. \cite{li_et_al:LIPIcs.ITP.2025.15} have formalised
correctness of their implementation problem 
of refinement global specifications \cite{Li2025c} in Rocq.  
It is an interesting future work to combine it with 
the Rocq mechanisation of the top-down approach by Ekici et al.~\cite{EKY2025}, 
to formalise the subject reduction, deadlock-freedom and 
liveness theorems based on this work.

\section{Conclusion}
\label{sec:conclusion}
We have proposed an asynchronous association relation 
and proved its sound and complete operational correspondence.
This work is the first to prove these results based on 
(1) asynchronous precise subtyping and (2) projection with
coinductive full merging. We introduced a new operational semantics
for global types, which captures more behaviours allowed by permuting actions 
than the previous asynchronous global type
semantics provided by Barwell et al.~\cite{BHYZ2025} and Honda et al.~\cite{HYC2016}.
We developed a new projection relation which associates global types with
a pair of a local type and a queue type for each participant.
Using this correspondence, we derived 
the subject reduction theorem and 
the session fidelity theorem of the top-down system 
from the corresponding theorems of the bottom-up system 
\cite[Theorem~4.11 and 4.13]{GPPSY2023}. 
By proving that the projection $\stEnv$ 
of $\gtG$ is safe, deadlock-free and live, we can derive that asynchronous 
multiparty session processes typed by 
the top-down typing system (\inferrule{SessTop}) 
are also safe, deadlock-free and live (Theorem~\ref{the:livetop}).  

While Ghilezan et al.~\cite{GPPSY2023} have proved the subject reduction theorem 
and session fidelity theorem under the subsumption rule 
of $\asubt$, it does not use the top-down typing system.    
On the other hand, Ghilezan et al.~\cite{Ghilezan2019} have shown that multiparty synchronous 
subtyping is precise in the synchronous multiparty session calculus 
using the top-down system. None of the previous works 
has studied:
(1) properties of coinductive full merge projection;
(2) liveness of asynchronously associated typing contexts; nor 
(3) association with respect to asynchronous subtyping.
An interesting open question is whether the
association theorems hold for the sound decidable asynchronous subtyping
relations \cite{CYV2022,bocchi_et_al:CONCUR.2025.10} 
so that we can derive the subject reduction theorems under those 
relations.  

Pischke and Yoshida \cite{PY2026} have shown that in
synchronous semantics, the top-down system has the same typability 
as the bottom-up system, which ensures liveness,
if the top-down system meets three conditions:
(1) global types are balanced;
(2) the projection is coinductive with full
merging; and
(3) the subsumption rule uses $\stSub$.
We have proven the correctness of the bottom-up 
system with $\asubt$ implies
the correctness of the top-down system where 
(1) a global 
type is $\text{balanced}^+$; (2) a projection is coinductive with full
merging (soundness of typability).
It is an open question whether our top-down system can 
\emph{completely} type sessions which are typable by the
bottom-up system, which ensures asynchronous liveness
(completeness of typability).

We have demonstrated 
the usefulness of association in deriving the main theorems of the
top-down system, by \emph{reusing} the theorems in \cite{GPPSY2023}.
We have not yet reached a stage to claim that 
MPST is a theoretical framework for building asynchronous 
software systems. 
There still needs to be more effort applied to developing practical applications of MPST with asynchronous subtyping relations, 
for testing and maintaining compositionality and reusability of 
protocols. 
The most challenging topic is to type individual
asynchronous \emph{components}, 
each being written in a different programming 
language or running on a different platform, 
while ensuring their type-safety and deadlock-freedom, 
assuming they conform to a shared global protocol.  
Implementing such a distributed architecture requires 
significant engineering effort such as defining system 
requirements, identifying components,
splitting the system into components, 
integrating these components, and designing the interfaces for
components. We will continue our avenue to promote 
asynchronous optimisations with MPST for specifying and implementing
real-world distributed components.

\paragraph{Acknowledgements. } We thank the reviewers for their detailed
comments and helpful suggestions.

\bibliographystyle{plain}
\bibliography{main}

\begin{thebibliography}{10}

\bibitem{ITU-T-Z120-1996}
Message sequence chart (msc).
\newblock Recommendation Z.120, International Telecommunication Union (ITU-T),
  Geneva, October 1996.
\newblock (10/96). Revised and approved by WTSC (Geneva, 9--18 Oct 1996).

\bibitem{DBLP:conf/isola/BarbaneraD24}
Franco Barbanera and Mariangiola Dezani{-}Ciancaglini.
\newblock Asynchronous multiparty sessions with internal delegation - dedicated
  to rocco de nicola on the occasion of his 70th birthday.
\newblock In Tiziana Margaria and Bernhard Steffen, editors, {\em Leveraging
  Applications of Formal Methods, Verification and Validation. REoCAS
  Colloquium in Honor of Rocco De Nicola - 12th International Symposium, ISoLA
  2024, Crete, Greece, October 27-31, 2024, Proceedings, Part {I}}, volume
  15219 of {\em Lecture Notes in Computer Science}, pages 322--339. Springer,
  2024.

\bibitem{BHYZ2025}
Adam~D. Barwell, Ping Hou, Nobuko Yoshida, and Fangyi Zhou.
\newblock {Crash-Stop Failures in Asynchronous Multiparty Session Types}.
\newblock {\em {Logical Methods in Computer Science}}, 2025.

\bibitem{BocchiKM24}
Laura Bocchi, Andy King, and Maurizio Murgia.
\newblock Asynchronous subtyping by trace relaxation.
\newblock In Bernd Finkbeiner and Laura Kov{\'{a}}cs, editors, {\em Tools and
  Algorithms for the Construction and Analysis of Systems - 30th International
  Conference, {TACAS} 2024, Held as Part of the European Joint Conferences on
  Theory and Practice of Software, {ETAPS} 2024, Luxembourg City, Luxembourg,
  April 6-11, 2024, Proceedings, Part {I}}, volume 14570 of {\em Lecture Notes
  in Computer Science}, pages 207--226. Springer, 2024.

\bibitem{bocchi_et_al:CONCUR.2025.10}
Laura Bocchi, Andy King, Maurizio Murgia, and Simon Thompson.
\newblock {Abstract Subtyping for Asynchronous Multiparty Sessions}.
\newblock In Patricia Bouyer and Jaco van~de Pol, editors, {\em 36th
  International Conference on Concurrency Theory (CONCUR 2025)}, volume 348 of
  {\em Leibniz International Proceedings in Informatics (LIPIcs)}, pages
  10:1--10:19, Dagstuhl, Germany, 2025. Schloss Dagstuhl -- Leibniz-Zentrum
  f{\"u}r Informatik.

\bibitem{Brand1983}
Daniel Brand and Pitro Zafiropulo.
\newblock On {{Communicating Finite-State Machines}}.
\newblock {\em Journal of the ACM}, 30(2):323--342, April 1983.

\bibitem{lmcs:7238}
Mario Bravetti, Marco Carbone, Julien Lange, Nobuko Yoshida, and Gianluigi
  Zavattaro.
\newblock {A Sound Algorithm for Asynchronous Session Subtyping and its
  Implementation}.
\newblock {\em {Logical Methods in Computer Science}}, {Volume 17, Issue 1},
  March 2021.

\bibitem{BravettiCZ17}
Mario Bravetti, Marco Carbone, and Gianluigi Zavattaro.
\newblock Undecidability of asynchronous session subtyping.
\newblock {\em Inf. Comput.}, 256:300--320, 2017.

\bibitem{bravetti_et_al:LIPIcs.CONCUR.2025.11}
Mario Bravetti, Luca Padovani, and Gianluigi Zavattaro.
\newblock {A Sound and Complete Characterization of Fair Asynchronous Session
  Subtyping}.
\newblock In Patricia Bouyer and Jaco van~de Pol, editors, {\em 36th
  International Conference on Concurrency Theory (CONCUR 2025)}, volume 348 of
  {\em Leibniz International Proceedings in Informatics (LIPIcs)}, pages
  11:1--11:17, Dagstuhl, Germany, 2025. Schloss Dagstuhl -- Leibniz-Zentrum
  f{\"u}r Informatik.

\bibitem{DBLP:journals/toplas/CarboneHY12}
Marco Carbone, Kohei Honda, and Nobuko Yoshida.
\newblock Structured communication-centered programming for web services.
\newblock {\em {ACM} Trans. Program. Lang. Syst.}, 34(2):8:1--8:78, 2012.

\bibitem{DBLP:conf/pldi/Castro-Perez0GY21}
David Castro{-}Perez, Francisco Ferreira, Lorenzo Gheri, and Nobuko Yoshida.
\newblock Zooid: a {DSL} for certified multiparty computation: from mechanised
  metatheory to certified multiparty processes.
\newblock In Stephen~N. Freund and Eran Yahav, editors, {\em {PLDI} '21: 42nd
  {ACM} {SIGPLAN} International Conference on Programming Language Design and
  Implementation, Virtual Event, Canada, June 20-25, 2021}, pages 237--251.
  {ACM}, 2021.

\bibitem{CastroPerezY20}
David Castro{-}Perez and Nobuko Yoshida.
\newblock {CAMP:} cost-aware multiparty session protocols.
\newblock {\em Proc. {ACM} Program. Lang.}, 4({OOPSLA}):155:1--155:30, 2020.

\bibitem{Chen2017}
Tzu~Chun Chen, Mariangiola {Dezani-Ciancaglini}, Alceste Scalas, and Nobuko
  Yoshida.
\newblock On the preciseness of subtyping in session types.
\newblock {\em Logical Methods in Computer Science}, 13(2), June 2017.

\bibitem{CDY2024}
Tzu-Chun Chen, Mariangiola Dezani-Ciancaglini, and Nobuko Yoshida.
\newblock On the preciseness of subtyping in session types: 10 years later.
\newblock In {\em Proceedings of the 26th International Symposium on Principles
  and Practice of Declarative Programming}, PPDP '24, New York, NY, USA, 2024.
  Association for Computing Machinery.

\bibitem{CDPY2015}
Mario Coppo, Mariangiola Dezani-Ciancaglini, Luca Padovani, and Nobuko Yoshida.
\newblock {A Gentle Introduction to Multiparty Asynchronous Session Types}.
\newblock In {\em 15th International School on Formal Methods for the Design of
  Computer, Communication and Software Systems: Multicore Programming}, volume
  9104 of {\em LNCS}, pages 146--178. Springer, 2015.

\bibitem{CDYP13}
Mario Coppo, Mariangiola Dezani-Ciancaglini, Nobuko Yoshida, and Luca Padovani.
\newblock {Global Progress for Dynamically Interleaved Multiparty Sessions}.
\newblock {\em MSCS}, 26:238--302, 2015.

\bibitem{CutnerYoshida21}
Zak Cutner and Nobuko Yoshida.
\newblock Safe session-based asynchronous coordination in rust.
\newblock In {\em Coordination Models and Languages: 23rd IFIP WG 6.1
  International Conference, COORDINATION 2021, Held as Part of the 16th
  International Federated Conference on Distributed Computing Techniques,
  DisCoTec 2021, Valletta, Malta, June 14–18, 2021, Proceedings}, page
  80–89, Berlin, Heidelberg, 2021. Springer-Verlag.

\bibitem{CYV2022}
Zak Cutner, Nobuko Yoshida, and Martin Vassor.
\newblock {Deadlock-Free Asynchronous Message Reordering in Rust with
  Multiparty Session Types}.
\newblock In {\em 27th ACM SIGPLAN Symposium on Principles and Practice of
  Parallel Programming}, volume abs/2112.12693 of {\em PPoPP '22}, pages
  246--261. ACM, 2022.

\bibitem{DemangeonH11}
Romain Demangeon and Kohei Honda.
\newblock Full abstraction in a subtyped pi-calculus with linear types.
\newblock In {\em Proceedings of CONCUR 2011}, volume 6901 of {\em LNCS}, pages
  280--296. Springer, 2011.

\bibitem{DY12}
Pierre-Malo Deni{\'e}lou and Nobuko Yoshida.
\newblock Multiparty session types meet communicating automata.
\newblock In {\em ESOP}, volume 7211 of {\em LNCS}, pages 194--213. Springer,
  2012.

\bibitem{DY13}
Pierre{-}Malo Deni{\'{e}}lou and Nobuko Yoshida.
\newblock Multiparty compatibility in communicating automata: Characterisation
  and synthesis of global session types.
\newblock In {\em {ICALP} 2013}, pages 174--186, 2013.

\bibitem{DBLP:journals/corr/abs-1208-6483}
Pierre{-}Malo Deni{\'{e}}lou, Nobuko Yoshida, Andi Bejleri, and Raymond Hu.
\newblock Parameterised multiparty session types.
\newblock {\em Logical Methods in Computer Science}, 8(4), 2012.

\bibitem{Giusto2025}
Cinzia Di~Giusto, Etienne Lozes, and Pascal Urso.
\newblock Realisability and complementability of multiparty session types.
\newblock In {\em Proceedings of the 27th International Symposium on Principles
  and Practice of Declarative Programming}, PPDP '25, New York, NY, USA, 2025.
  Association for Computing Machinery.

\bibitem{EKY2025}
Burak Ekici, Tadayoshi Kamegai, and Nobuko Yoshida.
\newblock {Formalising Subject Reduction and Progress for Multiparty Session
  Processes}.
\newblock In {\em 16th International Conference on Interactive Theorem Proving,
  2025, Reykjavik, Iceland}, volume 352 of {\em LIPIcs}, pages 19:1--19:23.
  Schloss Dagstuhl - Leibniz-Zentrum f{"{u}}r Informatik, 2025.

\bibitem{EY2024}
Burak Ekici and Nobuko Yoshida.
\newblock {Completeness of Asynchronous Session Tree Subtyping in Coq}.
\newblock In {\em 15th International Conference on Interactive Theorem Proving,
  2024, Tbilisi, Georgia}, volume 309, pages 6:1--6:20. Schloss Dagstuhl -
  Leibniz-Zentrum f{"{u}}r Informatik, 2024.

\bibitem{EY2026}
Burak Ekici and Nobuko Yoshida.
\newblock {Formalising Asynchronous Session Subtyping}.
\newblock {\em ACM Transactions on Computational Logic}, 27:1--45, 2026.

\bibitem{Gay2005}
Simon Gay and Malcolm Hole.
\newblock Subtyping for session types in the pi calculus.
\newblock {\em Acta Informatica}, 42(2-3):191--225, November 2005.

\bibitem{Ghilezan2019}
Silvia Ghilezan, Svetlana Jak{\v s}i{\'c}, Jovanka Pantovi{\'c}, Alceste
  Scalas, and Nobuko Yoshida.
\newblock Precise subtyping for synchronous multiparty sessions.
\newblock {\em Journal of Logical and Algebraic Methods in Programming},
  104:127--173, April 2019.

\bibitem{GPPSY2023}
Silvia Ghilezan, Jovanka Pantovi\'{c}, Ivan Proki\'{c}, Alceste Scalas, and
  Nobuko Yoshida.
\newblock Precise subtyping for asynchronous multiparty sessions.
\newblock {\em ACM Trans. Comput. Logic}, 24(2), nov 2023.

\bibitem{GraubmannRudolphGrabowski1993MSCStandardization}
Peter Graubmann, Ekkart Rudolph, and Jens Grabowski.
\newblock The standardization of message sequence charts.
\newblock In {\em Proceedings of the {IEEE} Software Engineering Standards
  Symposium ({SESS} '93)}, pages 48--63, Brighton, United Kingdom, September
  1993. {IEEE} Computer Society.

\bibitem{lmcs:6869}
Jonas~Kastberg Hinrichsen, Jesper Bengtson, and Robbert Krebbers.
\newblock Actris 2.0: Asynchronous session-type based reasoning in separation
  logic.
\newblock {\em Logical Methods in Computer Science}, Volume 18, Issue 2, Jun
  2022.

\bibitem{honda.vasconcelos.kubo:language-primitives}
Kohei Honda, Vasco~T. Vasconcelos, and Makoto Kubo.
\newblock Language primitives and type disciplines for structured
  communication-based programming.
\newblock In {\em Proceedings of ESOP 1998}, volume 1381 of {\em LNCS}, pages
  22--138. Springer, 1998.

\bibitem{HVY2009}
Kohei Honda, Vasco~Thudichum Vasconcelos, and Nobuko Yoshida.
\newblock {Type-Directed Compilation for Multicore Programming}.
\newblock {\em ENTCS}, 241:101--111, 2009.

\bibitem{HYC08}
Kohei Honda, Nobuko Yoshida, and Marco Carbone.
\newblock {Multiparty Asynchronous Session Types}.
\newblock In {\em Proceedings of POPL 2008}, pages 273--284. ACM, 2008.

\bibitem{HYC2016}
Kohei Honda, Nobuko Yoshida, and Marco Carbone.
\newblock Multiparty asynchronous session types.
\newblock {\em Journal of the {ACM}}, 63(1):9:1--9:67, 2016.

\bibitem{YH2024}
Ping Hou, Nobuko Yoshida, and Iona Kuhn.
\newblock {Less is More Revisited: Association with Global Protocols and
  Multiparty Sessions}.
\newblock {\em Theoretical Computer Science}, 1076:--, 2026.

\bibitem{LY2017}
Julien Lange and Nobuko Yoshida.
\newblock {On the Undecidability of Asynchronous Session Subtyping}.
\newblock In {\em Proceedings of FOSSACS 2017}, volume 10203 of {\em LNCS},
  pages 441--457. Springer, 2017.

\bibitem{Li2024-mu}
Elaine Li, Felix Stutz, and Thomas Wies.
\newblock Deciding subtyping for asynchronous multiparty sessions.
\newblock In {\em Lecture Notes in Computer Science}, Lecture notes in computer
  science, pages 176--205. Springer Nature Switzerland, Cham, 2024.

\bibitem{Li2023-cr}
Elaine Li, Felix Stutz, Thomas Wies, and Damien Zufferey.
\newblock Complete multiparty session type projection with automata.
\newblock In {\em Lecture Notes in Computer Science}, Lecture notes in computer
  science, pages 350--373. Springer Nature Switzerland, Cham, 2023.

\bibitem{Li2025c}
Elaine Li, Felix Stutz, Thomas Wies, and Damien Zufferey.
\newblock Characterizing implementability of global protocols with infinite
  states and data.
\newblock {\em Proc. ACM Program. Lang.}, 9(OOPSLA1), April 2025.

\bibitem{li_et_al:LIPIcs.ITP.2025.15}
Elaine Li and Thomas Wies.
\newblock {Certified Implementability of Global Multiparty Protocols}.
\newblock In Yannick Forster and Chantal Keller, editors, {\em 16th
  International Conference on Interactive Theorem Proving (ITP 2025)}, volume
  352 of {\em Leibniz International Proceedings in Informatics (LIPIcs)}, pages
  15:1--15:20, Dagstuhl, Germany, 2025. Schloss Dagstuhl -- Leibniz-Zentrum
  f{\"u}r Informatik.

\bibitem{Jay2017}
Jay Ligatti, Jeremy Blackburn, and Michael Nachtigal.
\newblock On subtyping-relation completeness, with an application to
  iso-recursive types.
\newblock {\em ACM Trans. Program. Lang. Syst.}, 39(1), March 2017.

\bibitem{Liskov1994-ra}
Barbara~H Liskov and Jeannette~M Wing.
\newblock A behavioral notion of subtyping.
\newblock {\em ACM Trans. Program. Lang. Syst.}, 16(6):1811--1841, November
  1994.

\bibitem{DBLP:conf/concur/MajumdarMSZ21}
Rupak Majumdar, Madhavan Mukund, Felix Stutz, and Damien Zufferey.
\newblock Generalising projection in asynchronous multiparty session types.
\newblock In Serge Haddad and Daniele Varacca, editors, {\em 32nd International
  Conference on Concurrency Theory, {CONCUR} 2021, August 24-27, 2021, Virtual
  Conference}, volume 203 of {\em LIPIcs}, pages 35:1--35:24. Schloss Dagstuhl
  - Leibniz-Zentrum f{\"{u}}r Informatik, 2021.

\bibitem{DBLP:conf/tlca/MostrousY09}
Dimitris Mostrous and Nobuko Yoshida.
\newblock Session-based communication optimisation for higher-order mobile
  processes.
\newblock In Pierre{-}Louis Curien, editor, {\em Typed Lambda Calculi and
  Applications, 9th International Conference, {TLCA} 2009, Brasilia, Brazil,
  July 1-3, 2009. Proceedings}, volume 5608 of {\em Lecture Notes in Computer
  Science}, pages 203--218. Springer, 2009.

\bibitem{NCY2015}
Nicholas Ng, Jose~G.F. Coutinho, and Nobuko Yoshida.
\newblock {Protocols by Default: Safe MPI Code Generation based on Session
  Types}.
\newblock In {\em 24th International Conference on Compiler Construction},
  volume 9031 of {\em LNCS}, pages 212--232. Springer, 2015.

\bibitem{NYH2012}
Nicholas Ng, Nobuko Yoshida, and Kohei Honda.
\newblock {Multiparty Session C: Safe Parallel Programming with Message
  Optimisation}.
\newblock In {\em 50th International Conference on Objects, Models, Components,
  Patterns}, volume 7304 of {\em LNCS}, pages 202--218. Springer, 2012.

\bibitem{PierceTAPL}
Benjamin Pierce.
\newblock {\em Types and Programming Languages}.
\newblock MIT Press, 2002.

\bibitem{PY2025}
Kai Pischke and Nobuko Yoshida.
\newblock {Asynchronous Global Protocols, Precisely}.
\newblock In {\em Components Operationally: Reversibility and System
  Engineering}, volume 16065 of {\em LNCS}, pages 116--133. Springer Nature,
  2025.

\bibitem{PY2026}
Kai Pischke and Nobuko Yoshida.
\newblock {Top Down$=$Bottom Up: Sound and Complete Characterisations of
  Liveness by Multiparty Global Protocols}.
\newblock {\em Proc. ACM Program. Lang.}, 2026.

\bibitem{POPL19LessIsMore}
Alceste Scalas and Nobuko Yoshida.
\newblock Less is more: Multiparty session types revisited.
\newblock {\em Proceedings of the ACM on Programming Languages}, 3(POPL):1--29,
  January 2019.

\bibitem{stutz:LIPIcs.ECOOP.2023.32}
Felix Stutz.
\newblock {Asynchronous Multiparty Session Type Implementability is Decidable -
  Lessons Learned from Message Sequence Charts}.
\newblock In Karim Ali and Guido Salvaneschi, editors, {\em 37th European
  Conference on Object-Oriented Programming (ECOOP 2023)}, volume 263 of {\em
  Leibniz International Proceedings in Informatics (LIPIcs)}, pages
  32:1--32:31, Dagstuhl, Germany, 2023. Schloss Dagstuhl -- Leibniz-Zentrum
  f{\"u}r Informatik.

\bibitem{10.1007/978-3-031-91121-7_13}
Felix Stutz and Emanuele D’Osualdo.
\newblock An automata-theoretic basis for specification and type checking of
  multiparty protocols.
\newblock In {\em Programming Languages and Systems: 34th European Symposium on
  Programming, ESOP 2025, Held as Part of the International Joint Conferences
  on Theory and Practice of Software, ETAPS 2025, Hamilton, ON, Canada, May
  3–8, 2025, Proceedings, Part II}, page 314–346, Berlin, Heidelberg, 2025.
  Springer-Verlag.

\bibitem{Takeuchi1994}
Kaku Takeuchi, Kohei Honda, and Makoto Kubo.
\newblock An {{Interaction-based Language}} and its {{Typing System}}.
\newblock In Constantine Halatsis, Dimitris~G. Maritsas, George Philokyprou,
  and Sergios Theodoridis, editors, {\em {{PARLE}} '94: {{Parallel
  Architectures}} and {{Languages Europe}}, 6th {{International PARLE
  Conference}}, {{Athens}}, {{Greece}}, {{July}} 4-8, 1994, {{Proceedings}}},
  volume 817 of {\em Lecture {{Notes}} in {{Computer Science}}}, pages
  398--413. Springer, 1994.

\bibitem{thien-nobuko-popl-25}
Thien Udomsrirungruang and Nobuko Yoshida.
\newblock Top-down or bottom-up? complexity analyses of synchronous multiparty
  session types.
\newblock {\em Proc. {ACM} Program. Lang.}, 9({POPL}):1040--1071, 2025.

\bibitem{Van_den_Heuvel2022-hr}
Bas van~den Heuvel and Jorge~A P{\'e}rez.
\newblock A decentralized analysis of multiparty protocols.
\newblock {\em Sci. Comput. Program.}, 222(102840):102840, October 2022.

\bibitem{Yoshida24}
Nobuko Yoshida.
\newblock Programming language implementations with multiparty session types.
\newblock In Frank~S. de~Boer, Ferruccio Damiani, Reiner H{\"{a}}hnle,
  Einar~Broch Johnsen, and Eduard Kamburjan, editors, {\em Active Object
  Languages: Current Research Trends}, volume 14360 of {\em Lecture Notes in
  Computer Science}, pages 147--165. Springer, 2024.

\bibitem{YHNN2013}
Nobuko Yoshida, Raymond Hu, Rumyana Neykova, and Nicholas Ng.
\newblock {The Scribble Protocol Language}.
\newblock In {\em 8th International Symposium on Trustworthy Global Computing},
  volume 8358 of {\em LNCS}, pages 22--41. Springer, 2013.

\bibitem{YoshidaVPH08}
Nobuko Yoshida, Vasco~Thudichum Vasconcelos, Herv{\'{e}} Paulino, and Kohei
  Honda.
\newblock Session-based compilation framework for multicore programming.
\newblock In {\em {FMCO} 2008}, pages 226--246, 2008.

\bibitem{YZF2021}
Nobuko Yoshida, Fangyi Zhou, and Francisco Ferreira.
\newblock {Communicating Finite State Machines and an Extensible Toolchain for
  Multiparty Session Types}.
\newblock In {\em 23rd International Symposium on Fundamentals of Computation
  Theory}, volume 12867 of {\em LNCS}, pages 18--35. Springer, 2021.

\end{thebibliography}

\end{document}